\documentclass[12pt]{article}
\usepackage[font=small,labelfont=bf]{caption}
\usepackage{subcaption}
\usepackage{amsmath,amscd,amssymb,mathrsfs,bm,amsthm}
\DeclareSymbolFontAlphabet{\mathrsfs}{rsfs}
\usepackage[mathscr]{eucal}
\usepackage{mathabx}
\usepackage{changepage}
\usepackage[usenames, dvipsnames]{xcolor}
\usepackage{tikz}
\usepackage{tikz-cd}
\usetikzlibrary{quotes}
\usepackage{savesym}
\usepackage[nosort]{cite}
\usepackage{wasysym}
\usepackage{graphicx}
\graphicspath{{Figs/}}
\usepackage{pifont}
\usepackage{float}
\usepackage{multicol}
\usepackage{amsfonts}
\usepackage{picture}
\usepackage{mathtools}
\savesymbol{iint}
\savesymbol{iiint}
\restoresymbol{TXF}{iint}
\restoresymbol{TXF}{iiint}
\usepackage{setspace}
\usepackage{jheppub}
\usepackage{slashed}
\hypersetup{colorlinks=true}
\hypersetup{linkcolor=blue}
\hypersetup{citecolor=magenta}
\hypersetup{urlcolor=blue}
\usepackage{multirow}
\usepackage{verbatim}
\usepackage{accents}
\usepackage{placeins}
\usepackage{enumitem}
\usepackage[all]{xy}
\definecolor{refkey}{rgb}{0,.8,.2}\definecolor{labelkey}{rgb}{1,0,0}
\numberwithin{equation}{section}
\let\O\relax\DeclareMathOperator{\O}{O}

\newcommand{\beq}{\begin{equation}}
\newcommand{\eeq}{\end{equation}}

\newcommand{\nn}{\nonumber}

\newlength{\offsetpage}
\newenvironment{widepage}{\begin{adjustwidth}{-\offsetpage}{-\offsetpage}%
    \addtolength{\textwidth}{2\offsetpage}}%
{\end{adjustwidth}}

\newtheorem{theorem}{Theorem}[section]

\newtheorem{lemma}[theorem]{Lemma}

\DeclareMathOperator*{\Ext}{Ext}
\DeclareMathOperator*{\Tor}{Tor}
\DeclareMathOperator*{\Hom}{Hom}
\DeclareMathOperator*{\Obj}{Obj}
\DeclareMathOperator*{\Mor}{Mor}
\newcommand{\Z}{\mathbb{Z}}
\newcommand{\R}{\mathbb{R}}

\newcommand{\C}{\mathbb{C}}

\def\SO{\text{SO}}
\def\O{\text{O}}

\def\SU{\text{SU}}

\def\U{\text{U}(1)}

\def\Spin{\text{Spin}}

\def\pt{\text{pt}}

\def\UU{\text{U}}

\def\Tr{\text{Tr}\,}

\def\ii{\text{i}}

\def\dd{\text{d}}

\def\Z{\mathbb{Z}}
\def\R{\mathbb{R}}

\def\C{\mathbb{C}}

\def\G{\mathbb{G}}

\def\sq{\text{~Sq}}
\def\CP{\mathbb{C}P^1}
\def\CPP{\mathbb{C}P^2}
\def\SS{\mathcal{S}}

\def\mod{\text{~mod~}}

\definecolor{bittersweet}{rgb}{1.0, 0.2, 0.6}

\usepackage{array}
\newcolumntype{C}{>{$}c<{$}}

\def\RPn (#1,#2){
  \fill (#1, #2) circle (3pt);
  \fill (#1, #2+1) circle (3pt);
}

\def\sqtwoL (#1,#2,#3){
  \draw[#3] (#1,#2) .. controls (#1-1,#2+1) .. (#1,#2+2);
}

\def\sqtwoR (#1,#2,#3){
  \draw[#3] (#1,#2) .. controls (#1+1,#2+1) .. (#1,#2+2);
}

\def \sqtwoCR (#1,#2,#3){
   \draw[#3] (#1,#2) .. controls (#1+1,#2+.5) and (#1+1.5,#2+2) .. (#1+2,#2+2);
}

\def \sqtwoCL (#1,#2,#3){
   \draw[#3] (#1,#2) .. controls (#1-1,#2+.5) and (#1-1.5,#2+2)  .. (#1-2,#2+2);
}

\def \sqone (#1,#2,#3){
  \draw[#3] (#1,#2) -- (#1,#2+1);
}

\def\Aone (#1,#2){
\fill (#1, #2) circle (3pt);
\fill (#1, #2+1) circle (3pt);
\fill (#1, #2+2) circle (3pt);
\fill (#1, #2+3) circle (3pt);
\fill (#1+2, #2+3) circle (3pt);
\fill (#1+2, #2+4) circle (3pt);
\fill (#1+2, #2+5) circle (3pt);
\fill (#1+2, #2+6) circle (3pt);
\draw (#1, #2) -- (#1, #2+1);
\draw (#1, #2+2) -- (#1, #2+3);
\draw (#1+2, #2+3) -- (#1 + 2, #2+4);
\draw (#1+2, #2+5) -- (#1+2, #2+6);
\draw (#1, #2) .. controls (#1-1, #2+1) .. (#1, #2+2);
\draw (#1+2, #2+4) .. controls (#1+3, #2+5) .. (#1+2, #2+6);
\draw (#1, #2+1) .. controls (#1+1, #2+1.5) and  (#1+1.5 ,#2+3) .. (#1+2,#2+3);
\draw (#1, #2+2) .. controls (#1+1, #2+2.5) and (#1+1.5, #2+4) .. (#1+2, #2+4);
\draw (#1, #2+3) .. controls (#1+1, #2+3.5) and (#1+1.5, #2+5) .. (#1+2, #2+5);
}

\def\rectangle (#1,#2,#3){   \draw[#3] (#1-0.15,#2-0.15) rectangle (#1+0.15,#2+0.15)}

\def\Eone (#1,#2){
\fill (#1, #2) circle (3pt);
\fill (#1, #2+1) circle (3pt);
\fill (#1, #2+2) circle (3pt);
\fill (#1, #2+3) circle (3pt);
\draw (#1, #2) -- (#1, #2+1);
\draw (#1, #2+2) -- (#1, #2+3);
\draw (#1, #2) .. controls (#1-1, #2+1) .. (#1, #2+2);
\draw (#1, #2+1) .. controls (#1+1, #2+2) .. (#1, #2+3);
}

\def\joker (#1,#2){
  \foreach \y in {#2, #2+1, #2+2, #2+3, #2+4}
           {\fill (#1,\y) circle (3pt);}
           \draw (#1,#2) -- (#1, #2+1);
           \draw (#1,#2+3) -- (#1, #2+4);
           \draw (#1,#2+0) .. controls (#1-1,#2+1) .. (#1, #2+2);
           \draw (#1,#2+2) .. controls (#1-1,#2+3) .. (#1, #2+4);
           \draw (#1,#2+1) .. controls (#1+1,#2+2) .. (#1, #2+3);
}

\def\jokercolor (#1,#2, #3){
  \foreach \y in {#2, #2+1, #2+2, #2+3, #2+4}
           {\fill[#3] (#1,\y) circle (3pt);}
           \draw[#3] (#1,#2) -- (#1, #2+1);
           \draw[#3] (#1,#2+3) -- (#1, #2+4);
           \draw[#3] (#1,#2+0) .. controls (#1-1,#2+1) .. (#1, #2+2);
           \draw[#3] (#1,#2+2) .. controls (#1-1,#2+3) .. (#1, #2+4);
           \draw[#3] (#1,#2+1) .. controls (#1+1,#2+2) .. (#1, #2+3);
}

\def\msopart (#1,#2,#3){
    \fill[#3] (#1,#2) circle (3pt); 
      \fill[#3] (#1, #2+1) circle (3pt);
      \fill[#3] (#1, #2+2) circle (3pt);
      \fill[#3] (#1+2, #2+2) circle (3pt);
      \fill[#3] (#1+2, #2+3) circle (3pt);
      \fill[#3] (#1+2, #2+4) circle (3pt);
      \fill[#3] (#1+2, #2+5) circle (3pt);
    \sqtwoCR(#1,#2, #3);
    \sqtwoCR (#1, #2+1, #3);
    \sqtwoCR (#1, #2+2, #3);
    \sqone (#1+2, #2+2, #3);
    \sqone (#1+2, #2+4, #3);
    \sqtwoR(#1+2, #2+3, #3);
    \sqone (#1, #2+1, #3); }

    \def\amme (#1,#2,#3){
       \fill[#3] (#1,#2) circle (3pt) ;
   \sqtwoR(#1,#2,#3);
      \fill[#3] (#1,#2+2) circle (3pt) ;
         \sqone(#1,#2+2,#3);
      \fill[#3] (#1,#2+3) circle (3pt) ;
   \sqtwoR(#1,#2+3,#3);
   \fill[#3] (#1,#2+5) circle (3pt) ;}
   
    \def\questionupsidedon (#1,#2,#3){
       \fill[#3] (#1,#2) circle (3pt) ;
          \sqtwoR(#1,#2,#3);
                 \fill[#3] (#1,#2+2) circle (3pt) ;
             \sqone(#1,#2+2,#3);
                \fill[#3] (#1,#2+3) circle (3pt) ;}

\DeclareRobustCommand*{\RaiseBoxByDepth}{%
    \raisebox{\depth}%
}
\newcommand{\uQ}{\RaiseBoxByDepth{\protect\rotatebox{180}{$Q$}}}

\newcommand{\cA}{\mathcal A}
\newcommand{\Sq}{\mathrm{Sq}}

\usepackage{spectralsequences}
\newcommand{\AdamsTower}[1]{\DoUntilOutOfBounds{
        \class[#1](\lastx, \lasty+1)
        \structline[#1]
}}

\title{Toric 2-group anomalies via cobordism}

\author[a]{Joe Davighi}
\author[b]{and Nakarin Lohitsiri,}
\author[c]{\\ with a mathematical appendix by Arun Debray}

\affiliation[a]{Physics Institute, University of Zurich, Switzerland}
\emailAdd{joe.davighi@physik.uzh.ch}
\affiliation[b]{Department of Mathematical Sciences, Durham University, United Kingdom}
\emailAdd{nakarin.lohitsiri@durham.ac.uk}
\affiliation[c]{Department of Mathematics, Purdue University, USA}
\emailAdd{adebray@purdue.edu}

\abstract{ 2-group symmetries arise in physics when a 0-form symmetry
  $G^{[0]}$ and a 1-form symmetry $H^{[1]}$ intertwine, forming a
  generalised group-like structure.  Specialising to the case where
  both $G^{[0]}$ and $H^{[1]}$ are 
	compact, connected, abelian groups ({\em i.e.} tori), 
	we analyse anomalies in such `toric 2-group symmetries' 
	using the cobordism classification. As a warm up example, we use cobordism to study various 't Hooft anomalies (and the phases to which they are dual) in Maxwell theory defined on non-spin manifolds.
For our main example, we compute the 5th spin bordism group of
  $B|\mathbb{G}|$ where $\mathbb{G}$ is any 2-group whose 0-form and
  1-form symmetry parts are both $\mathrm{U}(1)$, and $|\mathbb{G}|$ is the
  geometric realisation of the nerve of the 2-group $\mathbb{G}$. 
By leveraging a variety of algebraic methods, we
  show that $\Omega^{\mathrm{Spin}}_5(B|\mathbb{G}|) \cong \mathbb{Z}/m$ where $m$ is the
  modulus of the Postnikov class for $\mathbb{G}$, and we reproduce
  the expected physics result for anomalies in 2-group symmetries that
  appear in 4d QED.  Moving down two dimensions, we recap that any
  (anomalous) $\mathrm{U}(1)$ global symmetry in 2d can be enhanced to a toric
  2-group symmetry, before showing that its associated local anomaly
  reduces to at most an order 2 anomaly, when the theory is defined
  with a spin structure. }

\begin{document}
\maketitle

\setcounter{tocdepth}{3}

\section{Introduction}

In quantum field theory, anomalies are loosely defined to be quantum obstructions to symmetries. More precisely, anomalies can themselves be identified with (special classes of) quantum field theories in one dimension higher than the original theory, via the idea of `anomaly inflow'.\footnote{Mathematically, any anomalous theory in $d$ dimensions can be described by a relative quantum field theory~\cite{Freed:2012bs} between an extended field theory in one dimension higher (the `anomaly theory'), and the trivial extended field theory (see {\em e.g.}~\cite{Freed:2014iua}).}
This modern viewpoint led to an algebraic classification of anomalies via cobordism, which was made rigorous by Freed and Hopkins~\cite{Freed:2016rqq} following many important works (including~\cite{Witten:1985xe,Dai:1994kq,Kapustin:2014dxa,Witten:2015aba,Witten:2019bou}). The cobordism classification includes all known anomalies afflicting chiral symmetries of massless fermions in any number of dimensions, as well as other more subtle anomalies that involve discrete spacetime symmetries (see {\em e.g.}~\cite{Tachikawa:2018njr,Witten:2016cio,Wang:2018qoy}). 

The cobordism group that classifies anomalies, for $d$ spacetime dimensions and symmetry type $\SS$,\footnote{A given symmetry {\em type} $\SS$, as defined precisely in~\cite{Freed:2016rqq}, includes both the spacetime symmetry and an internal symmetry group, as we as maps between them, extended to all dimensions. } is the Anderson dual of bordism~\cite{Freed:2016rqq},
\beq \label{eq:intro_cobord}
H_{I\Z}^{d+2}\left(MT\SS\right) := [MT\SS, \Sigma^{d+2} I \Z],
\eeq
where $MT\SS$ denotes the Madsen--Tillman spectrum associated to symmetry type $\SS$, 
$I\Z$ denotes the Anderson dual of the sphere spectrum (with $\Sigma^{d+2}$ denoting the $d+2$-fold suspension),
and $[X, Y]$ denotes the homotopy classes of maps between a pair of spectra $X$ and $Y$. 
In fact, this cobordism classification goes beyond what we would normally think of as `anomalies' to cover all reflection positive invertible field theories (with or without fermions) with symmetry $\SS$ in $d+1$-dimensions. 

To unpack the meaning and significance of the cobordism group $H_{I\Z}^{d+2}\left(MT\SS \right)$, it is helpful to recall that it fits inside a defining short exact sequence~\cite{Freed:2016rqq} 
\beq \label{eq:SES_general}
\text{Ext}^1_\Z \left(\pi_{d+1}(MT\SS), \Z \right) \hookrightarrow
H_{I\Z}^{d+2}\left(MT\SS \right)
\twoheadrightarrow
\Hom \left( \pi_{d+2}(MT\SS), \Z \right),
\eeq
where $\pi_k(X)$ is the $k^\text{th}$ stable homotopy group of a spectrum $X$.
For the special case of chiral fermion anomalies in $d$ dimensions, the invertible field theory in question is precisely the exponentiated $\eta$-invariant~\cite{Witten:2019bou} of Atiyah, Patodi and Singer~\cite{Atiyah:1975jf,Atiyah:1976jg,Atiyah:1976qjr} in dimension $d+1$. In that case, integrating the gauge-invariant {\em anomaly polynomial} $\Phi_{d+2}$ provides an element of the right factor $\Hom \left( \pi_{d+2}(MT\SS), \Z \right)$. Theories with $\Phi_{d+2}=0$, which are in the kernel of the right map and thus in the image of the left map, correspond to residual `global anomalies' which are thus captured by the left-factor. For theories with symmetry type $\SS=\Spin\times G$, where $G$ is the internal symmetry group, we have that $\text{Ext}^1_\Z \left(\pi_{d+1}(MT\SS), \Z \right) \cong \Hom \left( \Tor \Omega_{d+1}^\Spin(BG), \R/\Z \right)$ classifies global anomalies, where $\Omega_{\bullet}^\Spin$ denotes spin bordism. A straightforward corollary is that if the group $\Tor \Omega_{d+1}^\Spin(BG)$ vanishes there can be no global anomalies, which has been applied to various particle physics applications in recent years~\cite{Freed:2006mx,Garcia-Etxebarria:2017crf,Garcia-Etxebarria:2018ajm,Wan:2018bns,Seiberg:2018ntt,Hsieh:2019iba,Davighi:2019rcd,Wan:2019gqr,Davighi:2020uab,Lee:2020ewl,Davighi:2020kok,Debray:2021vob,Davighi:2022fer,Davighi:2022bqf,Wang:2022eag,Lee:2022spd,Davighi:2022icj,Debray:2023yrs}.

Concurrent with this development of the cobordism classification of anomalies, in the past decade the notion of `symmetry' has been generalised beyond the action of groups on local operators, in various exciting directions (for recent reviews, see Ref.~\cite{Cordova:2022ruw,Bah:2022wot}). This includes the notion of $p${\em-form} symmetries~\cite{Gaiotto:2014kfa}, which act not on local operators (the case $p=0$) but more generally on extended operators of dimension $p$. These symmetries couple to background fields that are $(p+1)$-form gauge fields. Such higher-form symmetries are found in a plethora of quantum field theories. The most well-known examples include a pair of $\U$ 1-form symmetries in 4d Maxwell theory, and a $\Z/N$ 1-form symmetry in 4d $\SU(N)$ Yang--Mills theory; the latter exhibits a mixed anomaly with charge-parity symmetry when the $\SU(N)$ $\theta$-angle equals $\pi$, which led to new insights regarding the vacua of QCD~\cite{Gaiotto:2017yup}.  Other important classes of generalised symmetry include {\em non-invertible} and {\em categorical} symmetries, as well as {\em subsystem} symmetries. These shall play no role in the present paper.

Soon after the notion of $p$-form symmetries was formalised in~\cite{Gaiotto:2014kfa}, it was realised that higher-form symmetries of different degrees can `mix', leading to the notion of $p${\em-group} symmetry. The most down-to-earth example of this symmetry structure, which in general is described using higher category theory (see {\em e.g.}~\cite{Gripaios:2022yjy}), is the notion of a {\em 2-group} symmetry whereby a 1-form symmetry $H^{[1]}$ and a 0-form symmetry $G^{[0]}$ combine non-trivially. This occurs, for example~\cite{Hsin:2020nts,Lee:2021crt}, when a bunch of Wilson lines (which are charged under a discrete 1-form centre symmetry) can be screened by a dynamical fermion (which is charged under a 0-form flavour symmetry).
2-group symmetries were first observed, in their gauged form, in string theory in the context of the Green--Schwarz mechanism, in which a $\U$ 1-form symmetry combines non-trivially with a diffeomorphism~\cite{Baez:2005sn,Sati:2009ic,Fiorenza:2010mh,Fiorenza:2012tb}. Their appearance as global symmetries in field theory was appreciated first by Sharpe~\cite{Sharpe:2015mja}, before many new examples were identified in Refs.~\cite{Cordova:2018cvg,Benini:2018reh}. Further instances of 2-group symmetry have since been discovered and analysed in 4d~\cite{Hsin:2020nts,Lee:2021crt,Bhardwaj:2021wif,Cvetic:2022imb,Carta:2022fxc}, 5d~\cite{Cvetic:2022imb,Benini:2018reh,Apruzzi:2021vcu,DelZotto:2022fnw,DelZotto:2022joo}, and 6d~\cite{Cordova:2020tij,DelZotto:2020sop,Apruzzi:2021mlh} quantum field theories. A 2-group symmetry structure has recently been studied in the Standard Model~\cite{Cordova:2022qtz}, in the limit of zero Yukawa couplings.

In this paper we study how the cobordism classification of anomalies can be applied to 2-group symmetries. As a first step, we here specialise to 2-groups for which both $G^{[0]}$ and $H^{[1]}$ are compact, connected, abelian groups {\em ergo} tori $\cong \U^n$; we refer to such a structure, at times, as a `toric 2-group'. In this case, there are  1-form and 2-form Noether currents $j^{(1)}_G$ and $j^{(2)}_H$ associated with the 0-form and 1-form symmetries respectively, and the non-trivial 2-group `mixing' between these symmetries is reflected in a fusion rule for these currents, schematically
\beq \label{eq:current-fusion}
j^{(1)}_G(x) j^{(1)}_G(0) \sim K(x) j^{(2)}_H(0) + \dots,
\eeq
for a known (singular) function $K(x)$.
Toric 2-group symmetries appear, for example, in chiral $\U$ gauge theories in 4d when there are mixed `operator-valued' anomalies between global chiral symmetry currents and the gauge current~\cite{Cordova:2018cvg}, as well as in 3d~\cite{Damia:2022rxw} and in hydrodynamics~\cite{Iqbal:2020lrt}. 
Toric 2-group symmetries thus naturally appear in fermionic systems, for which anomalies can be more fruitfully analysed using cobordism (compared to purely bosonic anomalies).

For such a toric 2-group $\G$, we study the bordism theory of manifolds equipped with `tangential $\G$-structure'. This can be defined analogously to a tangential $\SS$-structure for ordinary symmetry type $\SS$, given now the particular 2-group $\G$ plus an appropriate form of spin structure. When the 2-group symmetry does not further mix with the spacetime symmetry, this reduces to studying the spin bordism groups of the topological space that classifies $\G$-2-bundles.

Thankfully much is known about the classifying space of $\G$-2-bundles that will be of use to us, given the physics applications we have in mind. Quite generally, the classifying space of a 2-group $\G$ can be computed as $B|\G|$, where $|\G|$ is the geometric realization of the nerve of the 2-group $\G$ viewed as a category~\cite{bartels2006higher,Baez2008}, which is itself a topological {\em 1}-group --- albeit often an infinite-dimensional one. This means that Freed and Hopkins' classification of invertible field theories with fixed symmetry type can be applied directly, taking the tangential structure to be $\Spin \times |\G|$.

In the special case of a pure 1-form symmetry valued in an abelian group $H^{[1]}=A$, this formula for the classifying space of the corresponding 2-group reduces to $B(BA)$, recovering the well-known result for classifying abelian gerbes~\cite{brylinski2007loop}. This special case, which is of significant physical interest, has already been used in the physics literature to study 1-form anomalies, for many examples with variously $\Spin$, $\SO$, and $\O$ structures in Ref.~\cite{Wan:2018bns}. An extension of this to a pure 2-form symmetry and its anomalies in 6d gauge theories was studied via cobordism in Ref.~\cite{Lee:2020ewl}.

More generally for any 2-group $\G$, the classifying space $B|\G|$ always sits in a `defining fibration' 
\beq \label{eq:intro_fib}
B\left(BH^{[1]}\right) \to B|\G| \to BG^{[0]}\,,
\eeq 
which allows one to leverage the usual spectral sequence methods to compute {\em e.g.} the cohomology and/or the (co)bordism groups of $B|\G|$. For `non-trivial' 2-groups, {\em i.e.} those for which the fibration (\ref{eq:intro_fib}) is non-trivial, there is to our knowledge only a smattering of calculations in the literature; in particular, the case $B(B\Z/2^{[1]}) \to B|\G| \to B\O^{[0]}$ was studied by Wan and Wang~\cite{Wan:2018bns} (Section 4), and $B(B\U^{[1]}) \to B|\G| \to B\SU(2)^{[0]}$ was studied by Lee and Tachikawa~\cite{Lee:2020ewl} (Appendix B.6). In this paper we extend these works with a dedicated cobordism analysis of theories with 2-group symmetry and their anomalies. 

The content of this paper is as follows. After reviewing 2-groups and their appearances in quantum field theory in \S \ref{sec:2-groups_review}, we then describe the topological spaces $B|\G|$ that classify 2-group background fields, and show how these spaces can be computed, in \S \ref{sec:cobordism-with-2-group}. By considering appropriate tangential structures for 2-group symmetries, this then leads to a definition of the cobordism groups relevant for describing 2-group symmetries and their anomalies (\S \ref{sec:cobord_2grp}). We then apply this cobordism theory to study a small selection of examples in detail, focussing on toric 2-group symmetries as an especially tractable case.
The main examples that we study in this paper are as follows. In \S \ref{sec:Maxwell-revisited}, which is a warm up example, we study anomalies involving the pair of 1-form symmetries in 4d Maxwell theory, defined on non-spin manifolds. As well as the $\Z$-valued `local' mixed anomaly between the two 1-form symmetries~\cite{Gaiotto:2014kfa,Wang-Senthil:2016a,Hsin:2019fhf,Brennan:2022tyl}, there is a trio of $\Z/2$-valued global anomalies (involving gravity) that only appear on non-spin manifolds; the associated phases have been studied in~\cite{Jian:2020qab}. In \S \ref{sec:QED} we turn to our main example of toric 2-group symmetries, with non-trivial Postnikov class, occurring in 4d abelian chiral theories. We see, using cobordism, how turning on the Postnikov class transmutes a $\Z$-valued local anomaly into a torsion-valued global anomaly \`a la the Green--Schwarz mechanism -- we precisely recover the physics results derived by C\'ordova, Dumitrescu and Intriligator in Ref~\cite{Cordova:2018cvg}.  Finally, in \S \ref{sec:2d} we drop down to two dimensions, where anomalous 0-form symmetries can be recast as 2-group symmetries by using the trivially conserved 1-form symmetry whose 2-form current is simply the volume form~\cite{Sharpe:2015mja}. Our cobordism analysis highlights the prominent role played by spin structures; when the $\U^{[0]}$ anomaly coefficient is odd, the Green--Schwarz mechanism cannot in fact absorb the anomaly completely but leaves a $\Z/2$-valued anomaly, which we identify as an anomaly in the spin structure; we also show how this $\Z/2$-anomaly is trivialised upon passing to $\Spin_c$ bordism.\footnote{These examples emphasize that bordism computations are useful tools for keeping track of such subtle mod 2 effects, by automatically encoding various characteristic classes' normalisation that
can differ in spin {\em vs.} non-spin theories -- for example, the class $c_1^2$ is even on a spin
manifold but needn't be on a non-spin one.
}

\section{Review: 2-groups and their classifying spaces}
\label{sec:2-groups_review}

In this Section we recall the notion of a 2-group $\G$ (\S \ref{sec:2grp-Defn}), and briefly reprise their appearance as generalised symmetry structures in quantum field theory (\S \ref{sec:2grp-in-QFT}).

\subsection{What is a 2-group?} \label{sec:2grp-Defn}

We begin by sketching some equivalent definitions of a 2-group, for which our main reference is Ref.~\cite{Baez2008} by Baez and Stevenson. First and most concise is probably the definition of a 2-group $\G$ as a higher category. Loosely, a 2-group $\G$ is in this language a 2-category that contains a single object, with all 1-morphisms and 2-morphisms being invertible. 

(Slightly) less abstractly, a 2-group can also be described using `ordinary' category theory, without recourse to higher categories, and this can be done in two ways; either as a groupoid in the category of groups, or as a group in  the category of groupoids.\footnote{It will be implicitly assumed that all groups and groupoids discussed carry a topology.} Expanding on the former viewpoint, a 2-group is itself a category whose objects form a group $\tilde{G}^{[0]}$, and whose morphisms also form a group, where the latter can be written as a semi-direct product between $\tilde{G}^{[0]}$ and another group $\tilde{H}^{[1]}$, {\em viz.}
\begin{equation} \label{eq:groupoid}
\Obj(\G)=\tilde{G}^{[0]}, \qquad \Mor(\G)= \tilde{H}^{[1]} \rtimes \tilde{G}^{[0]},
\end{equation}
such that all the various maps involved in the definition of a groupoid are continuous group homomorphisms (see \S 3 of~\cite{Baez2008}).
We will usually be interested in situations where both $\tilde{G}^{[0]}$ and $\tilde{H}^{[1]}$ are Lie groups.

This definition involving a pair of groups $(\tilde{G}^{[0]},\tilde{H}^{[1]})$ is closely related to yet another, somewhat more practical, definition of a 2-group as a topological crossed module. The data specifying a topological crossed module consists of a quadruplet 
\begin{equation} \label{eq:xmod_tilde}
\G=(\tilde{G}^{[0]},\tilde{H}^{[1]},t,\alpha),
\end{equation} 
where $\tilde{G}^{[0]}$ and $\tilde{H}^{[1]}$ can be identified with the groups that appear in (\ref{eq:groupoid}), and $t:\tilde{H}^{[1]}\rightarrow \tilde{G}^{[0]}$, $\alpha:\tilde{G}^{[0]}\to \text{Aut}(\tilde{H}^{[1]})$ are continuous homomorphisms such that the pair of conditions
\begin{equation}
t\left( \alpha(g)(h)\right)=gt(h)g^{-1}, \qquad \alpha(t(h))(h^\prime)=hh^\prime h^{-1}
\end{equation}
hold $\forall g\in \tilde{G}^{[0]}, \; h,\, h^\prime\in \tilde{H}^{[1]}$, which can be thought of as equivariance conditions on the maps. To define a 2-group from this data, one takes the group of objects $\Obj(\G)=\tilde{G}^{[0]}$ and the group of morphisms $\Mor(\G)=\tilde{H}^{[1]} \rtimes \tilde{G}^{[0]}$, where the semi-direct product is defined using the map $\alpha$ as $(h,g) \cdot (h^\prime, g^\prime)=(h\alpha(g)(h^\prime), gg^\prime)$, and the various source, target, and composition maps from $\Mor(\G)$, as well as identity maps from $\Obj(\G)$, are defined using simple formulae (for which  we refer the interested reader to \S 3 of~\cite{Baez2008}.)

We can tie together this circle of definitions by linking this final crossed module definition of $\G$ to the first definition of $\G$ as a 2-category; in that picture, $\tilde{G}^{[0]}$ is the group of 1-morphisms, and $\tilde{H}^{[1]}$ is the group of 2-morphisms from the identity 1-morphism to all other 1-morphisms~\cite{Kapustin:2013uxa}.
When $\tilde{G}^{[0]}$ and $\tilde{H}^{[1]}$ are both Lie, and when the maps $t$ and $\alpha$ are both smooth, $\G$ is a Lie 2-group.

The final description of a 2-group $\G$ that we just gave, as a topological crossed module, is still of somewhat limited use in physics because the groups $\tilde{G}^{[0]}$ and $\tilde{H}^{[1]}$ are not preserved under equivalence of 2-groups as 2-categories. Moreover, a physical system with a 2-group symmetry (possibly with associated background fields) is only sensitive to the equivalence class of the 2-group at long distance, not the 2-group itself, as described by Kapustin and Thorngren~\cite{Kapustin:2013uxa}. It is therefore more convenient to work with equivalence classes of 2-groups from the beginning, which we describe next.

An equivalence class of 2-groups, with a particular 2-group $\G=(\tilde{G}^{[0]},\tilde{H}^{[1]},t,\alpha)$ above as a representative, can also be described by a quadruplet $(G^{[0]},H^{[1]}, \alpha,\beta)$ where~\cite{Kapustin:2013uxa} 
\begin{equation}
G^{[0]} := \text{coker}\,t, \qquad H^{[1]} := \ker t.
\end{equation}
The homomorphism $\alpha:G^{[0]} \to \text{Aut}(H^{[1]})$ is so named because it descends from the homomorphism $\alpha:G^{[0]}\to\text{Aut}(H^{[1]})$. The new component is the so-called \emph{Postnikov class} 
\beq
\beta\in H^3(BG^{[0]}; H^{[1]}).
\eeq
More precisely, two crossed modules $\G = (\tilde{G}^{[0]},\tilde{H}^{[1]},t,\alpha)$ and $\G^{\prime} = (\tilde{G}^{\prime[0]}, \tilde{H}^{\prime[1]}, t^{\prime}, \alpha^{\prime})$ are said to be equivalent when $\tilde{G}^{\prime[0]} \cong \tilde{G}^{[0]}$, $\tilde{H}^{\prime[1]} \cong \tilde{H}^{[1]}$ and there are homomorphisms (not necessarily isomorphisms) $h: \tilde{H}^{[1]} \to  \tilde{H}^{\prime[1]}$ and $g: \tilde{G}^{[0]}\to  \tilde{G}^{\prime[0]}$ that makes the diagram
\begin{equation} \label{eq:2step-ext}
  \begin{tikzcd}
   1\arrow[r] & H^{[1]}\arrow[d,equal] \arrow{r} &  \tilde{H}^{[1]}\arrow[d,"h"] \arrow[r,"t"] &\tilde{G}^{[0]}\arrow[d,"g"]\arrow[r] & G^{[0]}\arrow[d,equal]\arrow[r] & 1\\
   1\arrow[r] & H^{\prime[1]}\arrow[r] &  \tilde{H}^{\prime[1]}\arrow[r,"t^{\prime}"] & \tilde{G}^{\prime[0]}\arrow[r] & G^{\prime[0]}\arrow[r] &1
  \end{tikzcd}
\end{equation}
of two exact sequences commutative and compatible with the actions of $\tilde{G}^{[0]}$ on $\tilde{H}^{[1]}$ and $\tilde{G}^{\prime[0]}$ on $\tilde{H}^{\prime[1]}$. These equivalence classes are classified by the group cohomology $\mathcal{H}^3(G^{[0]},H^{[1]})$~\cite{Brown-K-S:1982}, which is isomorphic to the ordinary cohomology $H^3(BG^{[0]};H^{[1]})$.

For completeness, we also define the notion of \emph{2-group homomorphisms}, in terms of ordinary group homomorphisms between different elements of the associated crossed modules \cite{Baez2008}. Represent two 2-groups $\G$ and $\G^{\prime}$ by crossed modules $(\tilde{G}^{[0]},\tilde{H}^{[1]},t,\alpha)$ and $(\tilde{G}^{\prime[0]}, \tilde{H}^{\prime[1]}, t^{\prime}, \alpha^{\prime})$. A 2-group homomorphism $f:\G \to \G^{\prime}$, which is a functor such that $f:\text{Obj}(\G) \to \text{Obj}(\G^{\prime})$ and $f:\text{Mor}(\G)\to \text{Mor}(\G^{\prime})$ are continuous homomorphisms of topological groups, can be represented by a pair of maps $h:\tilde{H}^{[1]}\to \tilde{H}^{\prime[1]}$ and $g:\tilde{G}^{[0]} \to \tilde{G}^{\prime[0]}$ such that the diagram
\begin{equation}
  \begin{tikzcd}
    \tilde{H}^{[1]} \arrow[r,"t"]\arrow[d,"h"]& \tilde{G}^{[0]}\arrow[d,"g"]\\
    \tilde{H}^{\prime[1]}\arrow[r,"t^{\prime}"] & \tilde{G}^{\prime[0]}
  \end{tikzcd}
\end{equation}
is commutative and that $h,g$ are compatible with $\alpha,\alpha^{\prime}$:
\begin{equation}
h(\alpha(a)(b)) = \alpha^{\prime}(g(a)) h(b),
\end{equation}
for all $a \in \tilde{G}^{[0]},\, b \in \tilde{H}^{[1]}$.

\subsection{2-group symmetries in quantum field theory}
\label{sec:2grp-in-QFT}

Physically, an equivalence class $\G=(G^{[0]}, H^{[1]}, \alpha, \beta)$ of 2-groups describes a symmetry structure appearing in quantum field theories which have both a 0-form symmetry group $G^{[0]}$ and a 1-form symmetry group $H^{[1]}$. When the two symmetries are not independent, the Postnikov class $\beta$ of the corresponding 2-group symmetry is non-trivial, and the 0-form and 1-form parts cannot be analysed in isolation. 

While we will eventually specialise to the case of toric 2-groups in the bulk of this paper, we first recap how 2-groups appear in field theory more broadly. It is convenient to distinguish two broad categories of 2-groups which in general have different physical origin:
\begin{enumerate}
\item A \emph{continuous 2-group} $\G$ is one with a continuous 1-form symmetry group $H^{[1]}$. 
This class of 2-groups arises in field theory when, for example, the gauge transformation law for the 1-form symmetry background gauge field is modified so that there is no operator-valued 't Hooft anomaly involving the background gauge field for the 0-form symmetry, in an analogue of the Green--Schwarz mechanism for global symmetries \cite{Cordova:2018cvg}. The 0-form symmetry can here be abelian or non-abelian, connected or disconnected, and compact or non-compact. In this paper we focus our attention on the special case where both the 0-form and 1-form symmetry are connected, compact, and abelian groups, which we refer to as a `{\em toric 2-group}'.
\item A \emph{discrete 2-group} $\G$ is one with a discrete 1-form symmetry group $H^{[1]}$. This class of 2-group symmetries arises in gauge theories when the gauge Wilson lines, which are charged objects under the 1-form symmetry $H^{[1]}$, are not completely screened in the presence of the background gauge field for the 0-form symmetry group $G^{[0]}$. The local operator $\mathcal{O}$ that screens the gauge Wilson lines when the 0-form background gauge field is turned off is also charged under the 0-form symmetry. Hence, when the background gauge field of $G^{[0]}$ is turned on, $\mathcal{O}$ transmutes a gauge Wilson line into a flavour Wilson line \cite{Benini:2018reh,Bhardwaj:2021wif,Lee:2021crt}. 
\end{enumerate}

\section{Cobordism with 2-group structure}
\label{sec:cobordism-with-2-group}

In this Section we describe the background gauge fields on associated principal $\G$-2-bundles which play an important role in field theories with 2-group symmetries. We focus on describing the topological spaces $B|\G|$ that classify these background fields, and show how to calculate such classifying spaces in elementary examples, before introducing the cobordism groups that are central to this paper in \S \ref{sec:cobord_2grp}.

\subsection{Background fields and their classifying spaces}

A theory with a 2-group symmetry $\G$ can be coupled to a background gauge field (which a physicist might wish to decompose into components of an ordinary 1-form gauge field and a 2-form gauge field),\footnote{A more explicit description of the background fields will be sketched in \S \ref{sec:QED} where we discuss the 2-group symmetries appearing in 4d abelian chiral gauge theory.
} which is a connection on a principal $\G$-2-bundle in the sense defined in~\cite{bartels2006higher}.

Mathematically, Bartels moreover shows that such 2-bundles over a manifold $X$ are classified by the \v Cech cohomology group $\check{H}^1(X,\G)$ with coefficients valued in the 2-group $\G$. Baez and Stevenson prove~\cite{Baez2008} that there is a bijection
\begin{equation}
\check{H}^1(X,\G) \cong [X,B|\G|],
\end{equation}
where $[A,B]$ denotes the set of homotopy classes of maps from $A$ to $B$, and $|\G|$ is the geometric realization of the nerve $\mathcal{N}\G$ of the 2-group $\G$ when viewed as a groupoid as in (\ref{eq:groupoid}). 
The nerve $\mathcal{N}\G$ of the category $\G$ with $\Obj(\G)=G$ and $\Mor(\G)=H \rtimes G$ is a set of simplices that we can construct out of these objects and morphisms -- we will see some explicit examples shortly.

Quantum field theories with 2-group symmetry $\G$ are thus defined on spacetime manifolds $X$ equipped with maps to $B|\G|$, just as a theory with an ordinary symmetry $G$ is defined on spacetimes equipped with maps to $BG$.

\subsubsection{Elementary examples of $B|\G|$}

Since the classifying space $B|\G|$ will play a central role in what follows, we pause to better acquaint the reader with $B|\G|$ and how it can be computed by looking at some simple examples. The calculations in this subsection are purely pedagogical -- readers who are familiar with (or not interested in) such constructions might wish to skip ahead to \S \ref{sec:general-BG-fib}.

\subsubsection*{Pure 0-form symmetries }

In the case that there is no 1-form symmetry, and the 2-group symmetry simply defines an ordinary 0-form symmetry $G$, we expect to recover the usual classifying space of $G$. In this case, the 2-group corresponds to the quadruplet $\G=(G,0,0,0)$, to use the topological crossed module notation, and indeed
$B|(G,0,0,0)|=BG$.
To see this from the nerve construction, we first view $\G$ as a category, which is very simple in this case: $\text{Obj}(\G) = G$ with only identity morphisms at each element. So there are no non-degenerate $n$-simplices when $n>0$, while the set of 0-simplices is $G$. Hence, the geometric realisation of the nerve $|\G|$ is simply the group $G$ itself, and its classifying space is $BG$, as claimed.

\subsubsection*{Pure 1-form symmetries }

The case of a `pure 1-form symmetry', in which there is no 0-form symmetry at all, corresponds to 2-groups of the form $\G=(0,H,0,0)=:H[1]$, where $H$ is the (abelian) 1-form symmetry group.
The corresponding classifying space is~\cite{bartels2006higher}
\begin{equation} \label{eq:BBH}
B|H[1]| = B(BH),
\end{equation}
which we often denote simply $B^2 H$, which coincides with the well-known classifying space of an abelian gerbe~\cite{brylinski2007loop}. For example, when $H=\U$, $B|H[1]|$ is an Eilenberg--Maclane space $K(\Z,3)$.

\paragraph{An example: pure $\Z/2$ 1-form symmetry.}
For an explicit example that illustrates how to actually calculate the geometric realization of the nerve of such a $\G=H[1]$, we consider the simplest case where $H=\Z/2$. 
Viewing this as a category as in~(\ref{eq:groupoid}), we have that $\Obj(\Z/2[1])$ consists of only one element because $G$ is the trivial group, denoted simply by $\bullet$, while the group $\Mor(\Z/2[1])$ is just isomorphic to $H=\Z/2$. Diagrammatically, this structure can be represented as
\begin{equation*}
\begin{tikzcd}
\bullet \arrow["-1"', loop, distance=2em, in=35, out=325] \arrow["1"', loop, distance=2em, in=215, out=145]
\end{tikzcd}.
\end{equation*}
The nerve of $\Z/2[1]$ is then a simplicial set $\mathcal{N}\Z/2[1]$ built out of these objects and morphisms, whose low dimensional components are given as follows.
\begin{align*}
  \left(\mathcal{N}\Z/2[1]\right)_0 &= \left\{ \bullet\right\} \qquad &\text{(0-simplex)}\\
  \left(\mathcal{N}\Z/2[1]\right)_1 &= \left\{ \bullet \xrightarrow{-1} \bullet \right\} \qquad &\text{(1-simplex)}\\
  \left(\mathcal{N}\Z/2[1]\right)_2 &= \left\{ \bullet \xrightarrow{-1} \bullet \xrightarrow{-1} \bullet \right\} \qquad &\text{(2-simplex)}\\
                 &\text{etc.}
\end{align*}

This is enough information for us to work out the CW complex  cell structure for its geometric realisation, which recall is the topological space that we denote $|\Z/2[1]|$. The 0-cell is just a point. The 1-cell comes from the union of the 0-cell and an interval, with the gluing rule given by the 1-simplex \begin{tikzcd} \bullet \arrow["-1"', loop, distance=2em, in=35, out=325] \end{tikzcd} {\em i.e.} identifying both ends of the interval with the 0-cell. Hence, the 1-cell has the topology of a 1-sphere $S^1$. Similarly, the form of the 2-simplex, written more suggestively as
\begin{equation*}
  \begin{tikzcd}
\bullet \arrow[r, "-1", bend left=49, shift right] & \bullet \arrow[l, "-1", bend left=49]
\end{tikzcd},
\end{equation*}
tells us that the 2-cell is constructed by identifying the 1-cell $S^1$ as the boundary of a 2-ball, with antipodal points on the boundary identified. This gives the topology of the 2-dimensional real projective space $\R P^2$. Building inductively in this way, one can show that $|\Z/2[1]|$ is topologically $\R P^\infty$. On the other hand, we know that this is the same as the classifying space of $\Z/2$, {\it i.e.} $|\Z/2[1]|\cong B\Z/2$. Therefore, the classifying space of the $\Z/2$ 1-form symmetry bundle is $B|\Z/2[1]| \cong B^2\Z/2$. 

This argument can be generalised to any 2-group of the form $H[1] = (0,H,0,0)$. For any group $H$, it is proven that the geometrisation of the nerve $|H[1]|$ is the topological space $BH$ \cite{Segal1968}, and the classifying space of $H[1]$-2-bundles is thus $B^2H$, coinciding with the known classification of $H$-gerbes when $H$ is abelian.

\subsubsection*{2-groups with a trivial map}
\label{sec:2-groups-1form-trivial-action}

Moving up in complexity, let us now consider a 2-group with both the constituent 1-groups $\tilde G$ and $\tilde H$ being non-trivial, but with at least one of the maps in the crossed module definition being trivial. To wit, consider a particular 2-group $\G=(\tilde G,\tilde H,t,\alpha)$, written using the notation of (\ref{eq:xmod_tilde}), where the map $\alpha: \tilde G\rightarrow \text{Aut}(\tilde H)$ is trivial {\em i.e.} $\alpha(g)$ is the identity automorphism for any $g\in \tilde G$.

\paragraph{An example: $\tilde G=\Z/2$ and $\tilde H=\Z/2$.}

Probably the simplest example of a 2-group of this kind is one where $\tilde G$ and $\tilde H$ are both $\Z/2$. Since $\text{Aut}(\tilde H) \cong \text{Aut}(\Z/2)$ is trivial, $\alpha$ is automatically trivial. If $\tilde G$ and $\tilde H$ do not interact at all, meaning that the map $t$ is also trivial (the module is `uncrossed'), then the 2-group factorises as a product of a 0-form and a 1-form symmetry:
\begin{equation*}
  |\G| \cong \Z/2\times B\Z/2,
\end{equation*}
whose classifying space is just $B\Z/2\times B^2\Z/2$.

Since $t:\tilde H\rightarrow \tilde G$ must be a homomorphism, the only other option for $t$ is the identity map $t:-1\mapsto -1$. Setting $\G=(\Z/2,\Z/2,\text{id},0)$, then as a category $\text{Obj}(\G) \cong \Z/2$ and $\text{Mor}(\G)=G\times H \cong \Z/2 \times \Z/2$, whose action can be fully captured in the following diagram:
\begin{equation*}
\begin{tikzcd}
-1 \arrow["{(1,-1)}"', loop, distance=2em, in=215, out=145] \arrow[r, "{(-1,-1)}", bend left=49] & 1 \arrow["{(1,1)}"', loop, distance=2em, in=35, out=325] \arrow[l, "{(-1,1)}", bend left=49]
\end{tikzcd}
\end{equation*}
Since the morphisms shown in the diagram are all the morphisms in this category, we can identify the morphisms $(1,1)$ and $(1,-1)$ as the identity morphisms at the objects $1$ and $-1$, respectively.

The components of the nerve $\mathcal{NG}$ in low dimensions are given by
\begin{align*}
  \mathcal{NG}_0 &= \left\{ \bullet_{-1},\bullet_1\right\}\\
  \mathcal{NG}_1 &= \left\{\bullet_{-1}\xrightarrow{(-1,-1)}\bullet_{1}, \bullet_{1}\xrightarrow{(-1,1)}\bullet_{-1}\right\}\\
  \mathcal{NG}_2 &= \left\{\bullet_{-1}\rightarrow\bullet_1\rightarrow \bullet_{-1}, \bullet_1\rightarrow \bullet_{-1}\rightarrow \bullet_1\right\}
\end{align*}
From these data, it is easy to see that the 0-cell is just a pair of points, the 1-cell is a circle $S^1$,
\begin{tikzcd}
\bullet_{-1} \arrow[r, "{(-1,-1)}", bend left=60] & \bullet_1 \arrow[l, "{(-1,1)}", bend left=60]
\end{tikzcd}, and the 2-cell is a 2-sphere $S^2$, constructed by joining \begin{tikzcd}
\bullet_{-1} \arrow[r, bend left=60] & \bullet_1 \arrow[l, bend left=60]
\end{tikzcd} and \begin{tikzcd}
  \bullet_{1} \arrow[r, bend left=60] & \bullet_{-1} \arrow[l, bend left=60]
\end{tikzcd} as two ``hemispheres'', with the 1-cell at the equator. Continuing ad infinitum we obtain $|\G| \cong S^\infty$. But the infinite sphere is contractible, which means it is homotopy equivalent to a point. Hence, the classifying space for $\G$ must be trivial. 

This result should not be a surprise because the crossed module under consideration, namely $\G = (\Z/2,\Z/2,t,0)$ with $t$ an isomorphism, is equivalent to the trivial 2-group because both $G^{[0]}=\text{coker}\,t$ and $H^{[1]}=\ker t$ are trivial!

\subsubsection{General $B|\G|$ as a fibration} \label{sec:general-BG-fib}

Our discussion so far has been mostly pedagogical; the idea was to acquaint the reader with some basic computations of the classsifying spaces of 2-groups. Thankfully, more powerful tools are available to us for the computations that we will be interested in. First and foremost, for any 2-group $\G$ there is a fibration\footnote{The extensions encoded in (\ref{eq:2-group-fib}) are in fact classified by $H^3_{SM}(G^{[0]}, H^{[1]})$~\cite{schommer2011central}, the Segal--Mitchison group cohomology in degree 3 (a degree that usually classifies 2-step extensions). }
\begin{equation}
  \label{eq:2-group-fib}
  \begin{tikzcd}
    B(B H^{[1]}) \arrow{r} & B|\G| \arrow{d}\\
    & BG^{[0]}
  \end{tikzcd}.
\end{equation}
Note that, because the 1-form symmetry $H^{[1]}$ is always an abelian group, $BH^{[1]}$ is itself a (topological) group, and so $B(BH^{[1]})$ can be defined using the ordinary definition of the classifying space of a group. We will see in the next Subsection that this fibration is sufficient for our purpose of anomaly classification.

\subsection{Cobordism description of 2-group anomalies}
\label{sec:cobord_2grp}

In order to classify anomalies associated to a 2-group $\G$ via
cobordism, most of the structure contained in $\G$ as a 2-group can be
discarded. As discussed in the previous Subsection, $\G$-2-bundles are
classified by the topological space $B|\G|$, where $|\G|$ is the
geometric realization of the nerve (henceforth just `the nerve') of
$\G$, regardless of other information contained in $\G$. Thus, we can
use the nerve $|\G|$, itself a topological group, to define a
tangential structure on spacetime just as we would for an ordinary
0-form symmetry. For instance, if we are interested in a fermionic
theory with internal 2-group symmetry $\G$ (with no mixing with the
spacetime symmetry), we can take the symmetry type to be\footnote{If we do not
  require a spin structure, we could use $\SO$ instead of $\Spin$ in
  defining the symmetry type $\SS$ -- as we do in \S \ref{sec:Maxwell-revisited}.}
\begin{equation}
\SS = \Spin \times |\G|\;.
\end{equation}

Directly applying Freed--Hopkins' classification~\cite{Freed:2016rqq} of invertible,
reflection-positive field theories to this symmetry type, suggests that the cobordism group
\begin{equation}
H_{I\Z}^{d+2}(MT(\Spin\times|\G| ) \cong \text{Tors}\,\Omega^{\Spin}_{d+1}(B|\G|) \times \Hom \left( \Omega^{\Spin}_{d+2}(B|\G|), \Z \right)\;
\end{equation}
correctly classifies anomalies in $d$-dimensional fermionic theories with the
2-group symmetry $\G$. Our detection of anomaly theories is therefore
distilled into computing the spin bordism groups of $B|\G|$, for which
we can use the same methods as for an ordinary symmetry. 

In particular, our general strategy will be to build up the spin bordism groups of $B|\G|$ by using the `defining fibration' described
above, $B^2 H^{[1]} \to B|\G| \to BG^{[0]}$, as follows. We first use the Serre
spectral sequence (SSS) \cite{Serre1951HomologieSD} to compute the cohomology of $B|\G|$ from the
fibration. Then, one can use this newly computed cohomology to compute the
spin bordism groups of $B|\G|$, for example via the Adams spectral sequence
(ASS) \cite{Adams:1958}. Alternatively, we can convert the result to homology using the
universal coefficient theorem, and use it as an input of the Atiyah--Hirzebruch spectral sequence (AHSS) \cite{AHSS1961:ProcSymp61} 
\begin{equation}
E^2_{p,q} = H_p(B|\G|;\Omega^{\Spin}_q(\pt))\;.
\end{equation}
for another fibration $\pt \to B|\G| \to B|\G|$. Kapustin and Thorngren
already used the SSS associated to (\ref{eq:2-group-fib})
in~\cite{Kapustin:2013uxa} to compute the cohomology of
$B|\G|$. Moreover, this approach has been used to calculate bordism
groups relevant to 2-group symmetries with non-trivial Postnikov
class; Wan and Wang did so using the ASS in Ref.~\cite{Wan:2018bns},
while Lee and Tachikawa used the AHSS in Ref~\cite{Lee:2020ewl}. We usually adopt the AHSS-based approach in this paper.

\section{Maxwell revisited}
\label{sec:Maxwell-revisited}

As a (rather lengthy) warm-up example, let us discuss Maxwell theory. We consider a 4d $\U$ gauge theory without matter, and variants thereof in which we couple the action to various TQFTs without changing the underlying symmetry structure. The action for vanilla Maxwell theory is
\beq \label{eq:pure_Max}
S_{\text{Maxwell}} = -\frac{1}{4e^2}\int_{M_4} \dd a \wedge \star \dd a\, ,
\eeq
where $a$ denotes the dynamical $\U$ gauge field and $\star$ denotes the Hodge dual, assuming a metric $g$ on $M_4$.

This theory has two $\U$ 1-form global symmetries~\cite{Gaiotto:2014kfa}, referred to as `electric' and `magnetic' 1-form symmetries. Since each is associated with a continuous group, one can write down the corresponding conserved 2-form currents, which are
\beq \label{eq:je_jm}
j_e = \frac{2}{e^2} \dd a, \qquad j_m = \star\frac{\dd a}{2\pi}\, .
\eeq
The former is conserved ($\dd\star j_e=0$) by Maxwell's equations $\dd\star \dd a=0$, and the latter by $\dd^2=0$ (and using $\star\star=1$).\footnote{For Maxwell theory in $d$ dimensions, $j_e$ is always a 2-form current and so the electric symmetry is always 1-form, while $j_m$ is more generally a $(d-2)$-form, and so the magnetic symmetry is generally a $(d-3)$-form symmetry. }

It was already observed in~\cite{Gaiotto:2014kfa} that there is a mixed 't Hooft anomaly between the two 1-form symmetries, which is a local anomaly that can be represented by an anomaly polynomial $\dd B_e \wedge \dd B_m$, where $B_e$ and $B_m$ are background 2-form gauge fields for the electric and magnetic 1-form symmetries respectively. This kind of anomaly and its consequences concerning various quantum numbers are discussed in Ref.~\cite{Hsin:2019fhf,Brennan:2022tyl} (see also~\cite{Wang-Senthil:2016a}).  In Ref.~\cite{Jian:2020qab}, a variety of 5d SPT phases that have Maxwell-like theories on the 4d boundary, protected by electric and magnetic 1-form symmetries, were derived. Here we show how these results are simply captured by the cobordism  classification.

\subsection{Cobordism classification of 1-form anomalies}

Maxwell theory has two global $\U$ 1-form symmetries, and no 0-form global symmetries. 
The classifying space of this global symmetry structure $\G$ is that of an abelian $H$-gerbe with $H=\U_e\times \U_m$. From (\ref{eq:BBH}),
\beq
B|\mathbb{G}| = B^2 \U_e \times B^2 \U_m \, .
\eeq
Since there are no fermions, we do not need a spin structure, and so can define pure Maxwell theory on any smooth, orientable 4d spacetime manifold.\footnote{We do not account for any time-reversal symmetry, that would enable us to pass from oriented to unoriented bordism.} The 5d invertible field theories with this symmetry type are therefore classified by the generalized cohomology group
\beq \label{eq:MaxCobordism}
\Hom \left(\Tor \Omega_5^\SO(B|\mathbb{G}|), \R/\Z \right) \hookrightarrow H_{I\Z}^6\left(MT(\SO \times |\mathbb{G}| ) \right) \twoheadrightarrow  \Hom \left( \Omega_6^\SO(B|\mathbb{G}|), \Z\right)\, ,
\eeq 
where this sequence splits, but not canonically.
In terms of anomaly theories, the factor on the right of this SES captures local anomalies, and the factor on the left captures global anomalies; the abelian group that detects all anomalies is isomorphic to the direct product of the two.

\subsubsection*{Local anomalies}

In Appendix~\ref{app:maxwell} we compute that $\Omega^\SO_6\left((B^2\U)^2 \right) = \Z$, and so
\beq \label{eq:Om6_Max}
\Hom \left( \Omega_{6}^{\SO} (B^2 \U_e \times B^2 \U_m), \Z \right) \cong \Z \, ,
\eeq
corresponding to the integral of the degree-6 anomaly polynomial $dB_e \cup dB_m$ on a generator of $\Omega^\SO_6\left((B^2\U)^2 \right)$.
This detects the local mixed 't Hooft anomaly between the two 1-form symmetries discussed in~\cite{Gaiotto:2014kfa}.

\subsubsection*{Global anomalies}

In Appendix~\ref{app:maxwell} we also compute that $\Omega^\SO_5\left((B^2\U)^2 \right) = (\Z/2)^3$. Thus, the group classifying global anomalies is
\beq \label{eq:Om5_Max}
\Hom \left( \Tor \Omega_{5}^\SO (B^2 \U \times B^2 \U), \R/\Z \right)\cong (\Z/2)^3\, .
\eeq
In  terms of characteristic classes, these three factors of $\Z/2$ correspond to 
\begin{equation}
  w_2 \cup w_3,\quad  w_2\cup \tau_e, \quad w_2 \cup \tau_m,\nn
\end{equation}
where $w_{2,3}$ denote the second and third Stiefel--Whitney classes of the tangent bundle, and $\tau_{e,m}$ is a unique generator of $H^3(B^2\U_{e,m};\Z/2)$ ({\it c.f.} Appendix \ref{app:summ-class-results}) which can be thought of as a mod 2 reduction of an integral cohomology class represented by $\frac{dB_{e,m}}{2\pi}$.

\subsection{Phases of non-spin Maxwell theory}

The $w_2 w_3$ invariant corresponds to a global gravitational anomaly seen on non-spin manifolds, related to that in~\cite{Wang:2018qoy}.
The $w_2 \tau_{e,m}$ bordism invariants correspond loosely to 't Hooft anomalies for each of the 1-form symmetries, that prevents them from being gauged on certain (non-spin) gravitational backgrounds.
We discuss these in more detail next. 

\subsubsection*{The {\boldmath$w_2 \cup \tau_{e}$} phase}

To see the global anomaly that is captured by the $w_2 \cup \tau_{e}$,\footnote{For simplicity, we consider the case where the 1-form symmetry is electric. The corresponding story for the magnetic 1-form symmetry can be obtained by applying the electric-magnetic duality.} we must go beyond the vanilla Maxwell theory described by (\ref{eq:pure_Max}), and couple 4d Maxwell to a topological term. The modified action is
\begin{equation} \label{eq:fermi-monopole}
  \begin{split}
    S &=-\frac{1}{4e^2}\int_{M_4} \dd a \wedge \star \dd a + \pi \ii \int_{M_4} w_2(TM_4)\cup \rho_2 \left[ \frac{\dd a}{2\pi} \right]_{\Z}\\
    &= S_{\text{Maxwell}}+ \pi \ii \int_{M_4} w_2(TM_4)\cup c_1,
  \end{split}
\end{equation}
where $[ \cdot ]_{\Z}$ denotes an integral cohomology class and $\rho_2$ is mod $2$ reduction, and where we shall write $c_1$ for both the first Chern class of the $\U$ gauge
bundle and its mod $2$ reduction. The topological term couples a
background magnetic 2-form gauge field to the $\Z/2$ part of the
magnetic 1-form symmetry, and then equates it to the second
Stiefel--Whitney class of the tangent bundle. The electric 1-form
symmetry shifting $a \mapsto a + \lambda$, where $\lambda$ is a closed 1-form, remains intact.

This theory describes
fermionic-monopole electrodynamics; the topological term $\int w_2 \cup c_1$ forces all monopoles to become
fermionic~\cite{Thorngren:2014pza,Ang:2019txy}.  To see this, following~\cite{Wang:2018qoy}, consider
adding a magnetic monopole via an 't Hooft line operator of charge $1$
along $\ell \subset M_4$ that we excise from $M_4$ with the boundary condition
that $\int_{S^2_{\ell}} c_1 = 1$ on a small 2-sphere $S^2_{\ell}$ around
$\ell$. The theory is now defined on the complement $M^{\prime}_4$ of
$\ell$ in the presence of the 't Hooft operator, which is a manifold with
boundary $\partial M^{\prime}_4 \cong S^2_{\ell} \times \ell$. For the integral
$\int_{M^{\prime}_4}w_2(TM^{\prime}_4)\cup c_1$ to make sense on a manifold with
boundary, we need a trivialisation of the integrand on
$\partial M^{\prime}_4$.\footnote{This is because for an $n$-manifold with
  boundary $M^{\prime}$, the fundamental homology class $[M^{\prime}]$ is in
  $H_n(M^{\prime}, \partial M^{\prime})$. The integration of a cohomology class on
  $M^{\prime}$ is in fact a pairing between a cohomology class and the
  fundamental class $[M^{\prime}]$, so when we write
  $\int_{M^{\prime}} c$ for a cohomology class $c\in H^n(M^{\prime})$, we need to
  first find a class in $H^n(M^{\prime},\partial M^{\prime})$ that ``corresponds'' to
  $c$ to make sense of the pairing. This is always possible because
  the long exact sequence in cohomology
\begin{equation}
\ldots \rightarrow H^{n-1}(\partial M^{\prime}) \rightarrow H^n(M^{\prime},\partial M^{\prime}) \rightarrow H^n(M^{\prime}) \rightarrow 0\nn
\end{equation}
implies that any class $c \in H^{n}(M^{\prime})$ can always be lifted to a
class $C \in H^n(M^{\prime},\partial M^{\prime})$. The choice of this lift is the choice
of trivialisation of $c$ on $\partial M^{\prime}$.} Now, since $c_1$ is non-trivial on
$S^2_{\ell}$ due to the non-zero monopole charge, we
have to trivialise the $w_2(TM^{\prime}_4)$ factor. A trivialisation of the
second Stiefel--Whitney class is nothing but a spin structure. Since
there is a unique spin structure on $S^2$, a spin structure on
$S^2_{\ell} \times \ell$ is the same as a spin structure along $\ell$. Therefore, to
define the additional phase, we must choose a spin structure along the
monopole worldline: the monopole is a fermion. 

In order to see the 't Hooft anomaly afflicting $\U_e$, 
we now couple the theory to a background electric 2-form gauge field $B_e$, and promote the 1-form $\U_e$ global symmetry transformation
above to a `local' transformation by relaxing the condition that the 1-form parameter
$\lambda$ has to be closed. In fact, $\lambda$ does not strictly need to be a 1-form on $M_4$ -- more precisely, we now consider shifting the gauge field $a$ by any connection $\lambda$.
For a connection $\lambda$ with non-zero curvature $\dd \lambda$, the action (\ref{eq:fermi-monopole}) shifts under
$a \mapsto a+\lambda$, by
\begin{equation} \label{eq:w2tau_anom}
\delta S = \pi \ii \int_{M_4} w_2 \cup \rho_2\left[ \frac{\dd \lambda}{2\pi} \right]_\Z ,
\end{equation}
which encodes the 't Hooft anomaly associated with certain `large 2-form gauge transformations', and on certain gravitational backgrounds. 

For example, take $M_4$ to be $\CPP$. As usual, we can parametrize $\CPP$ with three complex coordinates $z_1$, $z_2$, and $z_3$, such that $\sum_{i=1}^3 z_i^\ast z_i = 1$ and with the equivalence $z_i \sim e^{i\alpha} z_i$ for $\alpha \in \R/\Z$. Define the 2-form
$\omega := \frac{i}{2} \partial \overline{\partial} \log (z_1^2 + z_2^2)$,
which is just the volume form on the $S^2 \cong \CP\subset \CPP$ submanifold defined by $z_3=0$.  Its cohomology class $[\omega]_\Z$ can be taken as a generator for $H^2(\CPP; \Z)$, and likewise $a:= \rho_2 [\omega]_\Z$ can be taken as a generator for $H^2(\CPP; \Z/2)$. We also have that $w_2 (T\CPP) = a$. If we shift $a$ by a connection $\lambda$ with $\left[\frac{\dd \lambda}{2\pi}\right]_\Z=[n\omega]_\Z=n[\omega]_\Z$, the shift in the action is 
\beq
\delta S = \pi \ii\langle [M_4], a^2 \rangle = n\pi\ii
\eeq
where $\langle \cdot, \cdot \rangle$ here denotes the pairing between mod 2 homology and cohomology, thus realising the $\Z/2$-valued global anomaly for odd values of $n$.

As usual, there is a dual description of this anomaly in terms of the 5d SPT phase that captures it via inflow to the 4d boundary. If $M_4$ were nullbordant, the phase would here be $S_{\text{SPT}}= \pi \ii \int_X w_2 \cup \tau_{e}$ where recall $\tau_e = \rho_2\left[\frac{dB_e}{2\pi}\right]_\Z$, for a 5-manifold $X$ such that $\partial X = M_4$ and to which the background fields are extended. 
To see this,
first note that one could cancel the anomalous shift \eqref{eq:w2tau_anom} by adding a `counter-term'
$S_{\text{ct.}} = -\pi \ii \int_{M_4} \widehat{w_2} \wedge  \frac{B_e}{2\pi}$ where $\widehat{w_2}$ is a closed 2-form constructed from $w_2$,\footnote{To construct $\widehat{w_2}$, first note that there is an integral lift $W_2\in H^2(M_4; \Z)$ of $w_2\in H^2(M_4; \Z/2)$ because $w_3=0$ for any orientable 4-manifold. (One sees this from the long exact sequence in cohomology associated to the coefficient sequence $0\to \Z \xrightarrow{\times 2} \Z \xrightarrow{\rho_2} \Z/2 \to 0$, for which the Bockstein connecting homomorphism $\beta:H^2(\cdot; \Z/2) \to H^3(\cdot; \Z):w_2 \mapsto \beta(w_2)$ where $\rho_2 \beta(w_2) = w_3$ ). Next, given the integral class $W_2$, construct any complex line bundle over $M_4$ with $c_1=W_2$, and take $\widehat{w_2}$ to be the curvature 2-form of that bundle.
} 
recalling that the
gauge transformation for $B_e$ is
$B_e \mapsto B_e + \dd \lambda$.
However, this `4d action' is not properly quantised. (It is `half-quantised', precisely because there is a $\Z/2$ anomaly.) One must instead write it as the 5d action
$S_{\text{ct.}} = -\pi \ii \int_X w_2 \cup \tau_{e}$.
Thus the fermionic-monopole electrodynamics has exactly the right anomaly to be a boundary state of the SPT phase given by $S_{\text{SPT}}= \pi \ii \int_X w_2 \cup \tau_{e}$ (see also Ref.~\cite{Jian:2020qab}).\footnote{This mixed anomaly can also be interpreted as a remnant of the mixed anomaly
between the electric and magnetic 1-form symmetry when only the $\Z/2$ subgroup of the electric 1-form symmetry is coupled to a background field $w_2$. }

\begin{figure}[t]
\begin{widepage}
\begin{centering}
  \includegraphics[width=1.0\textwidth]{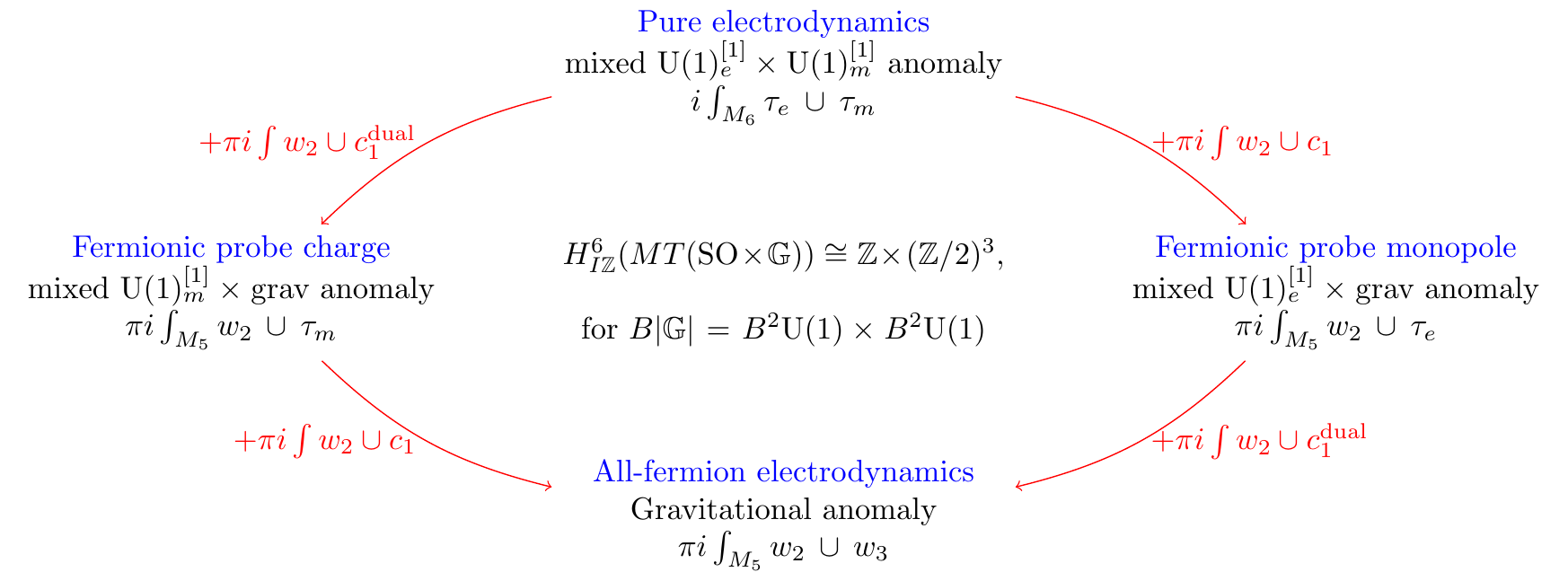}
\end{centering}
\end{widepage}
\caption{ A map of the four possible 't Hooft anomalies that can afflict a 4d theory with a pair of $\U$ 1-form global symmetries, defined on orientable manifolds without a spin structure. Each 't Hooft anomaly corresponds to a factor in the cobordism group $H_{I\Z}^6 \cong \Z \times (\Z/2)^3$, and is exhibited by Maxwell theory on its own (top, corresponding to the local $\Z$-valued anomaly) or coupled to a particular TQFT (left, bottom, and right, corresponding to the trio of global $\Z/2$-valued anomalies). The red arrows illustrate the relations between these four theories. For the global anomalies, these versions of 4d electromagnetism can also be realised as the boundary theories for 5d SPT phases. \label{fig:cartoon} }
\end{figure}

\paragraph{Bordism generator for the \boldmath{$w_2 \cup \tau_e$} anomaly.}

Of course, it is important to note that $\CPP$ is {\em not} nullbordant in $\Omega_4^\SO$, given the signature $\sigma = 1$ is a 4d bordism invariant. Given we saw the anomaly explicitly on $\CPP$, one cannot in fact realise it via the counterterm $S_{\text{ct.}} = -\pi \ii \int_X w_2 \cup \tau_{e}$ because there is no 5-manifold $X$ bounded by $M_4$. But as is well-known~\cite{Freed:2004yc}, the phase of the partition function $Z$ on such a non-nullbordant manifold suffers from an ambiguity, since it can always be shifted by a choice of generalized theta angle, {\em i.e.} a coupling to a non-trivial TQFT corresponding to an element in $\Hom(\Omega_4^\SO(\cdot), \R/\Z)$. Thus, one can fix $\arg Z[M_4, B_e] = \theta_0$ to a reference phase, and then $\arg  Z$  is uniquely defined on any other 4-manifold $(M_4^\prime, B_e^\prime)$ that is bordant to $(M_4, B_e)$.

This phase can be calculated by constructing a 5d mapping torus $\tilde X$ by taking a cylinder that interpolates between $(M_4, B_e)$ and $(M_4^\prime, B_e^\prime)$, and gluing its ends to make a closed 5-manifold. This $\tilde X$ will be a representative of the generator of the $\Z/2$ factor in $\Omega_5^\SO((B^2 \U)^2)$ that we have claimed is dual to $w_2 \cup \tau_e$.

In more detail, take $M_4 = \CPP$ and let $B_e^0$ denote any reference choice of background 2-form gauge field for the electric 1-form symmetry on  $M_4$. 
Next take the product manifold $\CPP \times [0,2\pi]$ with a product metric.
Over the interval $I = [0,2\pi]$ one implements (for example linearly)
the `large' 2-form gauge transformation 
$B_e \mapsto B_e^{\prime}:=B_e + n \omega$, $n \in \Z$.\footnote{We refer to this as a `large' gauge transformation because the gauge parameter $n \omega$ is closed but not exact, having a `winding number' of $n$.}
Since $B_e$ and $B_e^{\prime}$ are gauge-equivalent, the manifold $\CPP \times [0,2\pi]$ can be glued at 0 and $2\pi$ to make a closed 5-manifold with $\U$-2-bundle:
\beq
\tilde X := S^1 \times \CPP, \qquad B_e(\theta, z_i) = B_e^0 + \frac{t}{2\pi} n\omega,
\eeq
where $t\in I$. 

To verify that this mapping torus can be taken as a generator of the bordism group  $\Omega^\SO_5\left((B^2\U)^2 \right)$, it is enough to evaluate the bordism invariant `anomaly theory' $w_2 \cup \tau_e$ on $[\tilde X]$ and find a non-trivial value for the phase. Physically, evaluating $\langle w_2 \cup \tau_e, [\tilde X] \rangle$ computes the phase accrued by the partition function on $(\Sigma, B_e)$ upon undergoing the 2-form gauge transformation $B_e \to B_e^\prime$.
Doing the computation, we have that $\tau_e = \left[\frac{\dd t}{2\pi}\right]\cup \left[n\omega\right] = y \cup nx \in H^3(\tilde X; \Z/2)$,
where $x$ is the non-trivial element of $H^2(\tilde X; \Z/2)$ obtained by pulling back $a$ along the projection $\tilde X \to \CPP$, and
 $y$ is the non-trivial element of $H^2(\tilde X; \Z/2)$ obtained by pulling back the generator of $H^1(S^1; \Z/2)$ along $\tilde X \to S^1$. For $n$ odd, we have that $\tau_e = y\cup x$, the unique non-trivial element of $H^3(\tilde X; \Z/2)$. Thus $w_2 \cup \tau_e = x^2 y \in H^5(\tilde X; \Z/2)$ for $n$ odd, and so the anomaly theory evaluates to $1 \mod 2$ on this mapping torus, while it is trivial for $n$ even.

\subsubsection*{The \boldmath{$w_2 \cup \tau_m$} phase}

An identical account can be given for the magnetic 1-form symmetry, if one replaces the topological coupling $\pi \ii \int w_2 \cup c_1$ in (\ref{eq:fermi-monopole}) by  $\pi \ii \int w_2 \cup c_1^\text{dual}$, where $c_1^\text{dual}$ denotes the first Chern class of the electromagnetic dual of the gauge field. In that case, the 4d theory suffers from a 't Hooft anomaly afflicting the magnetic 1-form symmetry, which obstructs $\U_m[1]$ from being gauged on non-spin manifolds such as $\CPP$. The 5d anomaly theory is $\int_{M_5} w_2 \cup \tau_m$.

Physically, this same anomaly can be understood from a different perspective. Suppose one couples vanilla Maxwell theory (\ref{eq:pure_Max}) to a charge-1 fermion, defined on all orientable 4-manifolds using a $\Spin_c$ structure. The electric 1-form symmetry is explicitly broken, but one can still probe anomalies involving the magnetic 1-form symmetry and the gravitational background. The magnetic 1-form symmetry remains intact, even though we effectively have a `charge-$\frac{1}{2}$' monopole due to the constraint $c_1=\frac{w_2 \mod 2}{2}$ that follows the definition of $\Spin_c$, because this monopole is not dynamical but rather acts as a constraint on the field configurations that are summed in the path integral. The symmetry type is now $\Spin_c \times \G_m$, with $B|\G_m|=B^2 \U_m$, for which the cobordism group $H_{I\Z}^6(MT(\Spin_c \times \G_m)) \cong \Z/2 \times \Z^2$.\footnote{The two factors of $\Z$ arise simply because we have included the `$\U$ gauge symmetry' in our definition of the symmetry type, because it is now entangled with the spacetime symmetry. If one computed the reduced $\Spin_c$ bordism these factors would go, leaving only the $\Z/2$-valued $w_2 \cup \tau_m$ 't Hooft anomaly for the magnetic 1-form symmetry.} Of course, all the anomalies involving the electric 1-form symmetry are absent, as is the $w_2 w_3$ anomaly because the $\Spin_c$ requirement trivialises $w_3$, but the $w_2 \cup \tau_m$ global anomaly remains.

\subsubsection*{The \boldmath{$w_2 \cup w_3$} phase}

The final $\Z/2$-valued global anomaly that is detected by (\ref{eq:Om5_Max}) corresponds to the 5d SPT phase $\int_{M_5} w_2 \cup w_3$. To exhibit this anomaly, one starts from vanilla Maxwell theory and couples $w_2(TM_4)$ as a background gauge field to both the electric and magnetic $\U$ 1-form symmetries. 
This SPT phase is purely gravitational and has been extensively analysed in the literature. Its boundary states include the theory of all-fermion electrodynamics~\cite{Kravec:2014aza,Wang:2018qoy} and a fermionic theory with an emergent $\SU(2)$ gauge symmetry~\cite{Wang:2018qoy}.

\subsection{Scalar QED and 1-form anomaly interplay} \label{sec:SQED_interplay}

Now let us couple Maxwell theory, defined as before with an $\SO$ structure, to a charge-2 boson. This coupling to matter breaks the electric 1-form symmetry down to a discrete remnant, 
\beq
\U_e[1]  \rightarrow \Z/2_e[1] \, ,
\eeq
while the full $\U_m[1]$ magnetic 1-form symmetry is preserved. This version of scalar QED therefore furnishes us with two global 1-form symmetries, one discrete and one continuous, and no 0-form symmetries. The classifying space of the global symmetry $\G^\prime$ is therefore 
\beq
B|\G^\prime| = B^2 \Z/2_e \times B^2 \U_m\, ,
\eeq
and the 5d invertible field theories with this symmetry type are classified by the sequence (\ref{eq:MaxCobordism}) but with $|\G|$ replaced by $|\G^\prime|$.
In Appendix~\ref{app:charge2} we compute the relevant bordism groups to stitch together this generalized cohomology group.

\paragraph{Local anomalies:}
We compute from Appendix~\ref{app:charge2} that 
\beq
\Hom \left( \Omega_{6}^{\SO} (B^2\Z/2_e \times B^2\U_m), \Z \right) = 0 \, ,
\eeq
and so
there are {\em no} local anomalies; clearly, the mixed local 't Hooft anomaly `$dB_e \cup dB_m$' can no longer be realised now that one of the 1-form symmetries is broken to a discrete group, because the associated 2-form gauge field now has zero curvature.

\paragraph{Global anomalies:}
On the other hand, we find that
$\Omega^\SO_5\left(B^2\Z/2 \times B^2\U \right) = (\Z/2)^4$. The group
classifying global anomalies is now \beq \Hom \left( \Tor \Omega_{5}^\SO
  (B^2\Z/2_e \times B^2\U_m), \R/\Z \right)\cong (\Z/2)^4\, , \eeq and we
notice that there is an extra factor of $\Z/2$ corresponding to an
extra global anomaly. At the level of characteristic classes, we can
represent these four global anomalies in terms of products of
Stiefel--Whitney classes and cohomology classes of
$H^{\bullet}(B^2\U\times B^2\Z/2;\Z/2)$ that are in degree-5. There are five of them
in total:
\begin{equation} w_2\cup w_3, \quad w_2 \cup \tau_m, \quad w_2 \cup \sq^1u_2, \quad u_2\cup \tau_m
  ,\quad \sq^2\sq^1 u_2,\nn
\end{equation}
where $u_2$ is the unique generator of $H^2(B^2\Z/2;\Z/2)$ (see
Appendix \ref{app:summ-class-results}). However, once pulled back to
an orientable 5-manifold $X$, Wu's relation \cite{Wu1950} tells us
that $\sq^2\sq^1 u_2 = w_2(TX) \cup \sq^1u_2$, so it is not an
independent bordism invariant, and we end up with four invariants, as
expected from the bordism group computation.

\subsection*{The anomaly interplay}

Given an appropriate map $\pi$ between two spectra $MTH$ and $MTH^\prime$, there is an induced pullback map between the corresponding cobordism theories,  $\pi^\ast:H_{I\Z}^{d+2}\left(MTH^\prime \right)\to H_{I\Z}^{d+2}\left(MTH \right)$, that can be used to relate anomalies between the two theories. This idea of `anomaly interplay' has recently been used, for example, to relate local anomalies in 4d $\UU(2)$ gauge theory to Witten's $\SU(2)$ anomaly~\cite{Davighi:2020bvi}, to study anomalies in $\Z/k$ symmetries in 2d~\cite{Grigoletto:2021zyv,Grigoletto:2021oho} with applications to bootstrapping conformal field theories, and to derive anomalies in non-abelian finite group symmetries in 4d~\cite{Davighi:2022icj}. Physically, the idea of `pulling back anomalies' from one symmetry to another is not new, but goes back to Elitzur and Nair's analysis of global anomalies~\cite{Elitzur:1984kr}, following Witten~\cite{WITTEN1983422}.
In all these cases the symmetry type takes the form $\{\Spin \text{~or~} \SO \} \times G^{[0]}$, where $G^{[0]}$ is a 0-form symmetry. In that case, a map of spectra $\pi$ is induced by any group homomorphism $\pi:G\to G^\prime$.

In the present case, where we have theories with 1-form global symmetries, it is straightforward to adapt this notion of anomaly interplay (which is just pullback in cobordism). Again, the crucial fact we use is that the {\em nerves} associated with the 1-form symmetries are themselves just ordinary (topological) groups, between which we can define group homomorphisms.
Letting $\pi:\Z/2 \to \U:(1,-1) \to (1,e^{\ii\pi})$ denote the subgroup embedding, there is an associated map between symmetry types, $\pi:\SO\times |\G^\prime| \to \SO \times |\G|$ and an induced pullback map between the cobordism theories, 
\beq
\pi^\ast: H_{I\Z}^6(MT(\SO\times |\G^\prime|)) \longrightarrow H_{I\Z}^6(MT(\SO\times |\G|) )\, .
\eeq
Note that, since $H_{I\Z}^\bullet$ is a contravariant functor, the map between anomaly theories goes in the opposite direction to the subgroup embedding that we started with. Moreover, there is a pullback diagram for the whole short exact sequence characterizing $H_{I\Z}^6$, which encodes the notion of `anomaly interplay':
\begin{gather} \label{eq:MAX_interplay}
\begin{CD}
0 @>>> (\Z/2)^4 @>>> (\Z/2)^4 @>>> 0 @>>> 0 \\
@. @AAA     @A\pi^\ast AA   @AAA  @.     \\
0 @>>> (\Z/2)^3 @>>> (\Z/2)^3\times \Z @>>> \Z @>>> 0\, .
\end{CD}
\end{gather}
This anomaly interplay diagram can be used to track 't Hooft anomalies in the 1-form global symmetries, through the `integrating in' of a charge-2 boson.

Since all the generators of the cobordism groups can be represented by characteristic classes, we can represent the pullback $\pi^\ast$ by its action on the characteristic classes. The non-trivial action of $\pi^\ast$ is encoded in
\begin{align}
\pi^\ast: \, &\tau_e \cup \tau_m \mapsto u_2 \cup \tau_m, \label{eq:tetm_pullback}\\
 &w_2 \cup \tau_e \mapsto w_2 \cup \sq^1u_2. \label{eq:w2te_pullback}
\end{align}
Trivially, $\pi^\ast$ maps $w_2 \cup \tau_m$ and $w_2 \cup w_3$ to themselves.

The most interesting pullback relation is Eq. (\ref{eq:tetm_pullback}), which says that a $\Z$-valued local anomaly pulls back to a $\Z/2$-valued global anomaly. This is somewhat analogous to the interplay studied in~\cite{Davighi:2020bvi,Davighi:2020uab} between $\UU(2)$ local anomalies and $\SU(2)$ global anomalies for 4d 0-form symmetries, where the maps there corresponded to pulling back exponentiated $\eta$-invariants.

To see how the interplay works in this example involving 1-form symmetries, which is arguably simpler than the story for chiral fermion anomalies, we follow the general methodology set out in Ref.~\cite{Davighi:2020uab}. To wit, we start with a 5-manifold $M_5$ representative of a class in $\Omega_5^\SO(B^2\Z/2_e \times B^2\U_m)$ that is dual to $u_2 \tau_m$ ({\em i.e.} on which $u_2 \tau_m$ evaluates to 1 mod 2). We choose
\beq \label{eq:s2s2s1}
M_5 = S^2_e \times S^2_m \times S^1_\theta,
\eeq
equipped with a $\Z/2$ background 2-bundle that has support only on $S^2_e$, and a $\U_m$ background 2-bundle that has support only on $S^2_m \times S^1_\theta$. Specifically, the $\Z/2$ 2-form connection has non-trivial 2-holonomy round $S_e^2$, {\em viz} $\int_{S^2_e} u_2 =$ 1 mod 2. The $\U_m$ 2-form connection $B_m$ on $S^2_m \times S^1_\theta$ can be written
$B_m = B_m^0 + \frac{\theta}{2\pi} k \omega_m$, $k \in \Z$, where $\omega_m$ is the volume form on $S^2_m$ such that $\int_{S^2_m} \omega_m = 1$, $B_m^0$ is any connection on the $S^2_m$ factor, and $\theta \in [0,2\pi)$ parametrizes $S^1_\theta$. The important thing is the flux relation
\beq
\int_{S^2_m \times S^1_\theta} dB_m = k\int_{S^1_\theta} \frac{d\theta}{2\pi} \int_{S^2_m} \omega_m = k\, .
\eeq
Recalling $\tau_m = [dB_m/2\pi]_\Z$, integrating gives $\langle x_2 \cup \tau_m, [M_5] \rangle = k$ mod 2. When $k$ is odd, $M_5$ is not nullbordant, and the background fields cannot be simultaneously extended to any 6-manifold that $M_5$ bounds.

Now, to show that this global anomaly is the pullback under $\pi^\ast$ of the local anomaly $\tau_e \cup \tau_m$, we simply use the subgroup embedding $\pi:\Z/2 \to \U$ to embed the $\Z/2$ 2-connection in an $\U$ 2-connection $B_e$. One can take
\beq \label{eq:Be5mfd}
B_e = \frac{n}{2} \omega_e, \qquad n \in 2\Z+1, \qquad \int_{S^2_e} \omega_e = 1.
\eeq
The 5-manifold (\ref{eq:s2s2s1}) equipped with these structures $B_e$ and $B_m$ can be regarded as the pushforward in bordism of the $M_5$ that we started with, $\pi_\ast M_5$. This is nullbordant in $\Omega_5^\SO(B^2\U_e \times B^2\U_m)$; it can be realised as the boundary of a six-manifold $M_6 = D^3_e \times S^2_m \times S^1_\theta$, where $D^3_e$ is one half of a 3-sphere $S^3$ that is bounded by $S^2_e$, to which the electric 2-form connection (\ref{eq:Be5mfd}) can now be extended with $\int_{S^3} \frac{dB_e}{2\pi} = 2n$.\footnote{We emphasize that $B_e$ is {\em no longer} a flat connection when extended into the $S^3$ bulk, even though it restricts to a flat connection on the boundary of $D^3_e$.} Using Stokes' theorem, we thus have
\beq
\int_{S^2_e \times S^2_m \times S^1_\theta} \frac{B_e}{2\pi} \wedge \frac{dB_m}{2\pi} = \int_{D^3_e} \frac{dB_e}{2\pi} \int_{S^2_m \times S^1_\theta} \frac{dB_m}{2\pi} = nk\, ,
\eeq
obtaining the same phase from $\tau_e \tau_m$ as we did from $u_2 \tau_m$.

On the other hand, we only need to see that the characteristic class $\tau_e$ reduces to the class $\sq^1u_2$ when we restrict to a flat 2-bundle with mod 2 2-holonomy, in order to show that $w_2 \cup \tau_e$ pulls back to $w_2 \cup \sq^1u_2$. To see this most clearly, it is best to use the language of cochains instead of differential forms. By identifying $\U$ with $\R/\Z$, we represent the $\U$ 2-form gauge field $B_e$ by a real 2-cochain $b_e$. If the gauge field is flat, $\delta b_e$ must be trivial as a cochain valued in $\R/\Z$. In other words, $\delta b_e$ is an integral 3-cochain, whose cohomology class can be identified with $\tau_e$. To use this flat $\U$ 2-cochain to describe a $\Z/2$ 2-form gauge field, we further impose that it must have mod 2 2-holonomy, {\em i.e.} $b_e$ must be half-integral-valued, not just any real cochain. Then $\tilde{b}_e:= 2b_e$ is an integral cochain whose mod 2 reduction defines a cohomology class in $H^2(M_5;\Z/2)$ that coincides with $u_2$. As $\sq^1u_2 = \beta_2(u_2)$, where $\beta_2$ is the Bockstein homomorphism in the long exact sequence for cohomology
\begin{equation}
\ldots \rightarrow H^n(M_5;\Z/2) \rightarrow H^n(M_5;\Z/4) \rightarrow H^n(M_5;\Z/2) \xrightarrow{\beta_2} H^{n+1}(M_5;\Z/2) \rightarrow \ldots \nn
\end{equation}
induced by the short exact sequence $0 \rightarrow \Z/2 \rightarrow \Z/4 \rightarrow \Z/2 \rightarrow 0$,  it can be represented by the mod 2 reduction of $\frac{1}{2} \delta \tilde{b}_e$. But this is exactly the same as $\delta b_e$ which represents $\tau_e$. Therefore, the mod 2 reduction of $\tau_e$ is $\sq^1 u_2$ when we embed the $\Z/2$ 1-form symmetry inside the $\U$ 1-form symmetry.

There will be no further examples of anomaly interplay in the present paper; in particular, we do not consider any examples with non-trivially fibred 2-groups (although see footnote~\ref{foot:interplay} for some comments along these lines).

\section{QED anomalies revisited} \label{sec:QED}

The Maxwell examples in the previous Section exhibited only 1-form
global symmetries. In this Section, we move on to theories with both
0-form and 1-form global symmetries that are fused together in a
non-trivial 2-group structure. 

Quantum electrodynamics (QED) in 4d with certain fermion content
furnishes us with such a theory. Here both the 0-form and 1-form
symmetry groups are $\U$, and the 2-group structure is non-trivial
when there is a mixed 't Hooft anomaly between the global and gauged
0-form $\U$ currents, as was discovered
in~\cite{Cordova:2018cvg}. There is a further possible 't Hooft
anomaly in the 2-group structure that comes from the usual cubic $\U$
anomaly for the 0-form symmetry, but, rather than being a $\Z$-valued
local anomaly, the 2-group structure transmutes this cubic anomaly to a
discrete, $\Z/m$-valued global anomaly~\cite{Cordova:2018cvg}, with
$m$ given by the modulus of the integral Postnikov class of the
2-group. After reviewing the physics arguments for these statements,
we derive them from the cobordism perspective. Our spectral sequence calculations reproduce the order of the finite $\Z/m$ anomaly.

\subsection{From the physics perspective} \label{sec:U1U1-physics}

Consider a system of Weyl fermions coupled to a $\U_a$ gauge group with dynamical gauge field $a$. Assuming that a fermion with unit charge is present, the electric 1-form symmetry is broken completely. This leaves only the magnetic $\U^{[1]}$ 1-form symmetry, with 2-form current $j^m = \star\frac{f}{2\pi}$ as in (\ref{eq:je_jm}), where $f=\dd a$. Given enough Weyl fermions, one can find a global $\U^{[0]}$ 0-form symmetry that does not suffer from the ABJ anomaly. It is possible, however, that upon coupling a background gauge field $A$ to this global 0-form symmetry, there is an operator-valued mixed anomaly that is captured by the anomaly polynomial term
\begin{equation}
\Phi_6 \supset \frac{\kappa}{16\pi^3} f \wedge F \wedge F\, , 
\end{equation}
where $F=dA$ is the field strength for the background gauge field $A$. It describes, via the usual descent procedure, the shift in the effective action under the background gauge transformation $A \mapsto A + \dd \lambda^{(0)}$, with $\lambda^{(0)}$ a $2\pi$-periodic scalar, by
\begin{equation}
\label{eq:op-val-anom}
\delta S = \frac{\ii \kappa}{2}\int_{M_4} \lambda^{(0)} \frac{f \wedge F}{4\pi^2}\, .
\end{equation}
It was realised in Ref.~\cite{Cordova:2018cvg} that we should not interpret this term as an anomaly, but rather as a non-trivial 2-group structure. 

To see why this is the case, we first couple a 2-form background gauge field $B^{(2)}$, which satisfies the usual normalisation condition
\begin{equation}
  \label{eq:H-norm}
  \int_{M_3} \frac{\dd B^{(2)}}{2\pi} \in \Z
\end{equation}
on closed 3-manifolds,
to the magnetic 1-form symmetry via the coupling
\begin{equation}
S_{\text{coupling}} = \frac{\ii}{2\pi} \int_{M_4} f \wedge B^{(2)}.
\end{equation}
Then the potential anomaly \eqref{eq:op-val-anom} can be cancelled by modifying the background gauge transformation for $B^{(2)}$ from an ordinary 1-form gauge transformation $B^{(2)} \mapsto B^{(2)} + \dd \lambda^{(1)}$, that is independent of the 0-form gauge transformation of $A$, to a 2-group gauge transformation
\begin{equation}
\label{eq:2-gp-gauge transformation}
A \mapsto A + \dd \lambda^{(0)}, \qquad B^{(2)} \mapsto B^{(2)} + \dd \lambda^{(1)} + \frac{\hat{\kappa}}{2\pi} \lambda^{(0)} F,
\end{equation}
provided that we identify $\hat{\kappa}$ with $-\kappa/2$.\footnote{This is well-defined because it can be shown that $\kappa$ is always even.} Here the gauge transformation parameter $\lambda^{(1)}$ is a properly normalised $\U$ gauge field. 

This modified transformation mixes the magnetic $\U^{[1]}$ 1-form symmetry with the $\U^{[0]}$ 0-form symmetry, and encodes the 2-group structure at the level of the background fields. In our quadruplet notation, the 2-group is 
\begin{equation}
  \mathbb{G}= (\U,\U, 0,\hat{\kappa}),
\end{equation}
with the Postnikov class  $\hat{\kappa}\in H^3(B\U;\U) \cong \Z$, where we use the fact that $H^3(B\U;\U)$ with {\it continuous} $\U$ coefficients is isomorphic to $H^4(B\U;\Z)\cong \Z$ via the universal coefficient theorem. 
The nerve of such a 2-group is the extension
\begin{equation}
  |\mathbb{G}|=\U^{[0]}\times_{\hat{\kappa}} \U^{[1]}.
\end{equation}
The non-trivial 2-group structure can also be seen without turning on the background gauge fields explicitly. If we write $j$ for the 1-form current associated with $\U^{[0]}$, then the Ward identity takes the non-trivial form
\begin{equation}
\partial^{\mu} j_{\mu}(x) j_{\nu}(y) = \frac{\hat{\kappa}}{2\pi} \partial^{\lambda}\delta^{(4)}(x-y) j^m_{\nu\lambda}(y)\, ,
\end{equation}
where $j^m_{\nu\lambda}$ are the components of the magnetic 2-form current $j^m$. This Ward identity
shows how the 0-form and 1-form currents are fused, realising Eq.~(\ref{eq:current-fusion}).

The 2-group symmetry $\U^{[0]}\times_{\hat{\kappa}} \U^{[1]}$ can still suffer from an 't Hooft anomaly. 
Recall that, when $\hat{\kappa} =0$, the anomaly polynomial for our system of Weyl fermions is
\begin{equation}
\Phi_6 = \frac{\mathcal{A}_3}{6} c_1(F)^3 - \frac{\mathcal{A}_{\text{mixed}}}{24} p_1 (R) c_1(F),
\end{equation}
where $c_1(F)= F/2\pi$ is the first Chern class of the $\U^{[0]}$ bundle, and $p_1(R) = \frac{1}{8\pi^2} \Tr R\wedge R$ is the first Pontryagin class of the tangent bundle, where $R$ is the curvature 2-form. 
Before we continue, it is important to discuss the role of these anomaly coefficients in the (co)bordism context, still for the case $\hat\kappa = 0$. Na\"ively, one might think that the cubic anomaly coefficient $\mathcal{A}_3$ and the mixed $\U$-gravitational anomaly coefficient $\mathcal{A}_{\text{mixed}}$ are the two integers that classify anomalies according to cobordism, {\em i.e.} that $\mathcal{A}_3$ and $\mathcal{A}_{\text{mixed}}$ can be chosen as generators of the group
$\Hom \left(\Omega^{\Spin}_6 \left( B\U \right),\Z\right) \cong \Z \times \Z$.
However, this is not the correct identification because $\mathcal{A}_3$ and $\mathcal{A}_{\text{mixed}}$ are not quite independent. It can be shown that
\begin{equation}
\alpha_1 := c_1(F)^3 \qquad \text{and} \qquad \alpha_2 := \frac{c_1(F)}{6} \left( c_1(F)^2-\frac{1}{4} p_1(R) \right)
\end{equation}
are independent integral cohomology classes.\footnote{That $\alpha_1$ is integral is evident from the definition of $c_1(F)$. To see that $\alpha_2$ is integral, we observe that it is the anomaly polynomial for a Weyl fermion with charge $+1$, and then apply the Atiyah--Singer index theorem.} In terms of these basis generators, we can write $\Phi_6$ as
\begin{equation}
\Phi_6 = \frac{1}{6}\left( \mathcal{A}_3 - \mathcal{A}_{\text{mixed}} \right) \alpha_1 + \mathcal{A}_{\text{mixed}} \alpha_2.
\end{equation}
Again, by the Atiyah--Singer index theorem, $\Phi_6$ must be integral. Since $\alpha_1$ and $\alpha_2$ are integral basis generators, we can deduce that the two integers $(r,s)$ that label  $\Hom \left(\Omega^{\Spin}_6 \left( B\U \right),\Z\right)$ are
\begin{equation}
(r,s) = \left( \frac{1}{6}\left( \mathcal{A}_3-\mathcal{A}_{\text{mixed}} \right), \mathcal{A}_{\text{mixed}} \right).
\end{equation}

Continuing, let's switch the Postnikov class $\hat{\kappa}$ back on, and couple the background 2-form gauge field $B^{(2)}$ to the $\U^{[1]}$ 1-form symmetry. We are free to add a Green--Schwarz counter-term
\begin{equation} \label{eq:GS-term-4d}
S_{\text{GS}} = \frac{i n}{2\pi} \int_{M_4} B^{(2)} \wedge F, \qquad n \in \Z
\end{equation}
to the action.
Under the 2-group transformation \eqref{eq:2-gp-gauge transformation}, the effective action shifts if the anomaly coefficients $\mathcal{A}_3$ and $\mathcal{A}_{\text{mixed}}$ are non-zero, by
\begin{equation}
\delta S = \ii\frac{\mathcal{A}_3 + 6n\hat{\kappa}}{6} \int_{M_4} \lambda^{(0)} \frac{F \wedge F}{4\pi^2} - \ii\frac{\mathcal{A}_{\text{mixed}}}{24}\int_{M_4} \lambda^{(0)} p_1(R) .
\end{equation}
Since one may choose the counterterm coefficient $n$ to be any integer, one sees that $\mathcal{A}_3$ is only well-defined modulo $6\hat{\kappa}$~\cite{Cordova:2018cvg}. It follows that the integer $r=\frac{1}{6}\left( \mathcal{A}_3-\mathcal{A}_{\text{mixed}} \right)$ that we claim classifies the anomaly is not really valued in $\Z$, being well-defined only modulo $|\hat{\kappa}|$. This corresponds to a {\em global} anomaly in the 2-group symmetry, that is valued in the cyclic group
\beq \label{eq:Zm_group}
\Z/m, \qquad m=|\hat\kappa|.
\eeq
We emphasize that, even though $\mathcal{A}_3$ is well-defined modulo $6\hat{\kappa}$ (as was derived in~\cite{Cordova:2018cvg}), the discrete group that classifies the global anomaly has order $|\hat\kappa|$ (and not $6|\hat\kappa|$). The mixed gravitational anomaly remains a $\Z$-valued local anomaly.

\subsection{From the bordism perspective}
\label{sec:QED-bordism-perspective}

In this Section we show how the 't Hooft anomalies afflicting the 2-group global symmetry, that we have just described, can be precisely understood using cobordism. To do so, we need to compute the generalized cohomology groups
\begin{equation}
H^6_{I\Z}\left( MT \left( \Spin \times |\G|\right) \right) , \qquad |\G| =  U(1)^{[0]}\times_{\hat{\kappa}} U(1)^{[1]}\, ,
\end{equation}
which are built from $\Omega_5^\Spin(B|\G|)$ and $\Omega_6^\Spin(B|\G|)$, for each value of the Postnikov class $\hat\kappa$. These abelian groups detect and classify all possible anomalies for this 2-group symmetry type. 

To compute these bordism groups, we follow the general strategy outlined in \S\ref{sec:cobord_2grp}.  We first apply the
cohomological Serre spectral sequence to the fibration
\begin{equation}
  \label{eq:fib1}
  K(\Z,3)\to B|\G|\to K(\Z,2)
\end{equation}
to calculate the cohomology of $B|\G|$ in relevant low degrees. Then we  will convert the
result into homology groups, which are fed into the Atiyah--Hirzebruch
spectral sequence for the fibration  $\text{pt}\to B|\G|\to B|\G|$,
for which the second page is
$E^2_{p,q} = H_p\left(B|\G|;\Omega^\Spin_q(\text{pt})\right)$, to compute the spin bordism.

So, to begin, the $E_2$ page of the cohomological Serre spectral sequence for the
fibration \eqref{eq:fib1} is given by
\begin{equation}
  E_2^{p,q} = H^p\big(K(\Z,2);H^q(K(\Z,3);\Z)\big).
\end{equation}
Using
$H^\bullet(K(\Z,3);\Z) \cong \{\Z,0,0,\Z,0,0,\Z/2,0,\Z/3,\Z/2,\ldots\}$, given in
Table \ref{tab:integral-cohomology-KZ3} of Appendix
\ref{app:summ-class-results}, we can construct the $E_2$ page as shown
in the left-hand side of Fig. \ref{fig:SerreSSU1U1}. In fact, what is
shown there is the $E_4$ page, since the entries are sparse enough
that there are no non-trivial differentials in the region we are
interested in until page $E_4$.
\begin{figure}[h]
  \centering
  \includegraphics[scale=0.7]{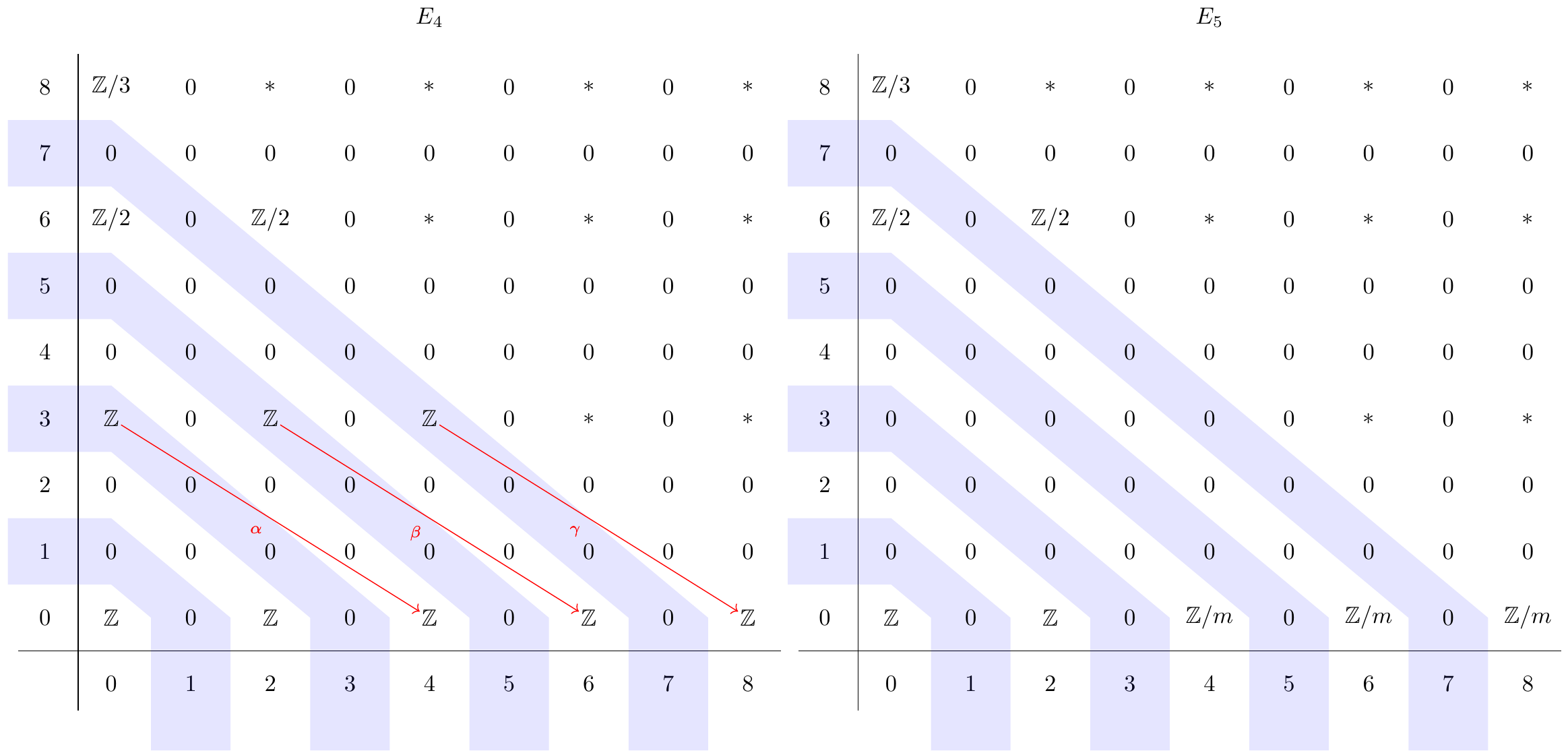}
  \caption{Fourth and fifth pages of the Serre spectral sequence to compute the integral cohomology of the fibration $K(\Z,3)\to B|\G|\to K(\Z,2)$.}
  \label{fig:SerreSSU1U1}
\end{figure}

The differentials $\alpha$, $\beta$, and $\gamma$ shown on the left-hand diagram of
Fig. \ref{fig:SerreSSU1U1} are linear in the Postnikov class
$\hat{\kappa}\in \Z$ (see the Appendix of Ref. \cite{Kapustin:2013uxa}).
More precisely, we write
\begin{equation}
  H^3(K(\Z,3);\Z)\cong \Hom (\Z^{[1]},\Z) \cong \Z \nn
\end{equation}
using the universal coefficient theorem and the fact that
$K(\Z,3) = B^2\U^{[1]} = B^3\Z^{[1]}$. Here, we include the
superscript to emphasise that this $\Z$ comes from our $\U$ 1-form
symmetry part. From this, we can write the entry $E^{0,3}_4$ in the
Serre spectral sequence as $\Hom (\Z^{[1]},\Z)$. Then, the
differential $\alpha$ is given by `contraction' (adopting the terminology
of \cite{Kapustin:2013uxa}) with the Postnikov class
\begin{equation}
  \begin{split}
    \alpha: \Hom (\Z^{[1]},\Z) &\to
    H^4(B\U^{[0]};\Z) \\
    x &\mapsto x \circ \hat{\kappa}
\end{split}
\end{equation}
where we make use of the fact that the Postnikov class is a
cohomology class
\begin{equation}
  \hat{\kappa} \in H^3(B\U^{[0]},\U^{[1]})\cong H^4(B\U^{[0]};\Z^{[1]})\;. \nn
\end{equation}
Similar arguments apply for the differentials $\beta$ and $\gamma$. Thus,
$\alpha$, $\beta$, and $\gamma$ map $1$ to $\pm\hat{\kappa}$, resulting in the
$E_5$ page as shown, where $m:= |\hat{\kappa}|$.  We can then read off the
integral cohomology groups to be
\begin{equation}
  \label{eq:cohom-BG}
  H^\bullet(B|\G|;\Z) \cong \{\Z,0,\Z,0,\Z/m ,0, e(\Z/2,\Z/m),0, e(\Z/6,\Z/m),\ldots\}\, 
\end{equation}
where the notation $e(A,B)$ denotes an extension of $A$ by $B$, {\it
  viz.} a group that fits in the short exact sequence
$B \hookrightarrow e(A,B) \twoheadrightarrow A$. We compute the mod 2 cohomology by the same method
in Appendix \ref{app:mod-2-cohomology}.

Heuristically, one can also argue for the form of $\alpha$ and $\beta$ as
follows ({\it c.f.} Appendix B.6 of Ref. \cite{Lee:2020ewl}). Let us
start with a representative of a generator of the cohomology group
$H^3(K(\Z,3),\Z)\cong \Z$, and ask what it becomes when $K(\Z,3)$ is the
fibre of $B|\G|$.  Given the normalisation condition
(\ref{eq:H-norm}), the 3-form
\begin{equation}
  \label{eq:htilde}
  \tilde h:=\dd B^{(2)}/2\pi
\end{equation}
represents the generator of $H^3(K(\Z,3);\Z)\cong \Z$. This is true when
$K(\Z,3)$ stands on its own. 
However, when we pass to the 2-group $\G$ and take $K(\Z,3)$ to be the fibre of
$B|\G|$, then $\tilde h$ is not gauge-invariant under (\ref{eq:2-gp-gauge
  transformation}) for a general Postnikov class, and cannot be a
representative of any cohomology class for $B|\G|$. We can remedy this
by modifying the definition of $\tilde h$ to
\begin{equation}
  \label{eq:h}
  h:= \tilde h - \frac{\hat \kappa}{2\pi} A \wedge F .
\end{equation}
The trade off is that the gauge-invariant $h$ is no longer closed; instead
\begin{equation} \label{eq:dh}
\dd h = - \frac{\hat{\kappa}}{4\pi^2} F \wedge F =-\hat{\kappa} c_1\cup c_1 ,
\end{equation}
where $c_1:= \frac{F}{2\pi}$ is the first Chern class of the
$U(1)^{[0]}$ bundle, and $c_1 \cup c_1$ can be taken as the generator of
$H^4(K(\Z,2); \Z)$.  The 2-group relation (\ref{eq:dh}) implies that
both the differentials $\alpha$ and $\beta$ in Fig.~\ref{fig:SerreSSU1U1} map
$1$ to $-\hat{\kappa}$.\footnote{The extra minus sign comes from the convention used to
  define the Postnikov class.}

From this cohomological starting point, we can proceed to compute the spin bordism groups. 
We find it is helpful to split the discussion into the cases where $m$ is even or odd, for which the bordism group calculations are tackled using different tricks. As a warm up, we first consider the simplest case where $m$ is zero, corresponding to a 0-form and 1-form symmetry that do not mix.

\subsubsection{Zero Postnikov class}

We first consider the trivial toric 2-group where $B|\G|$ is a simply a product space
$K(\Z,2) \times K(\Z,3)$. As $B|\G|$ is a product, we can use the K\"unneth
theorem to determine the homology groups of $K(\Z,2) \times K(\Z,3)$ from
the homology groups of each factor, given by Eq.~\eqref{eq:integral-cohomology-of-BU1} and  Table
\ref{tab:integral-cohomology-KZ3} in Appendix~\ref{app:summ-class-results}. We obtain
\begin{equation}
H_{\bullet}(K(\Z,2)\times K(\Z,3);\Z) \cong \left\{ \Z,0,\Z,\Z,\Z,\Z \times \Z/2, \Z, \Z \times \Z/6, \Z \times \Z/2,\ldots \right\}
\end{equation}
and construct the second page of the AHSS as shown in
Fig.~\ref{fig:AHSS-U1xBU1}, with non-trivial differentials on the
$E^2$ page indicated by coloured arrows. 

\begin{figure}[h]
  \centering
  \includegraphics[scale=0.8]{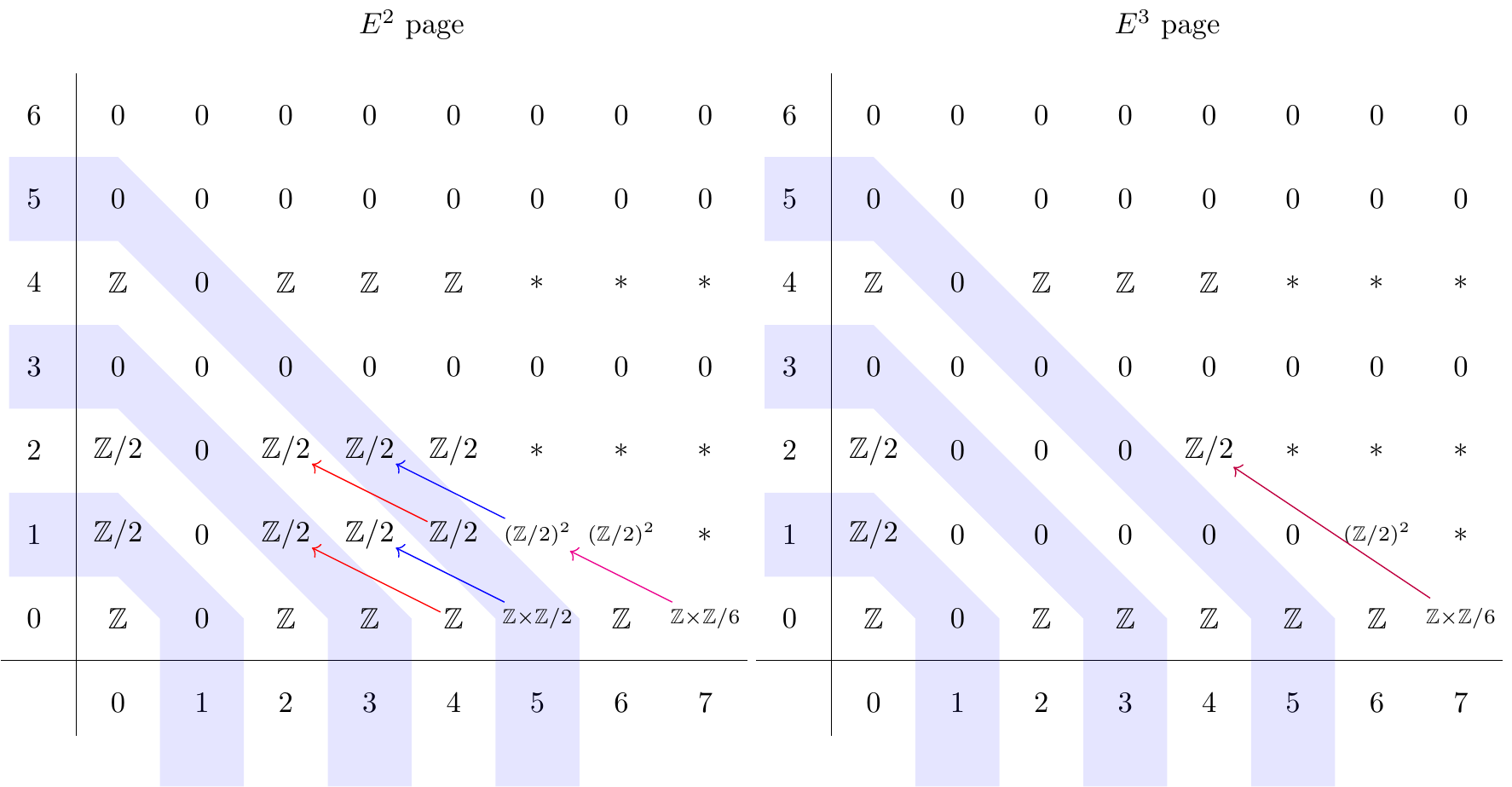}
  \caption{The $E^2$ and $E^{3}$ pages for the Atiyah--Hirzebruch spectral sequence for $\Omega^\Spin_{\bullet}(K(\Z,2) \times K(\Z,3))$. }
  \label{fig:AHSS-U1xBU1}
\end{figure}

The differentials
on the zeroth and the first rows are the composition
$\widetilde{\text{Sq}^2}\circ \rho_2$ and $\widetilde{\text{Sq}^2}$,
respectively, where $\rho_2$ is reduction modulo 2. To compute the action
of these differentials, we need to know how the Steenrod squares act
on the mod 2 cohomology ring of the product $K(\Z,2)\times K(\Z,3)$. From
the mod 2 cohomology rings of $K(\Z,2)$ and $K(\Z,3)$ given by
\eqref{eq:mod2-cohomology-of-BU1} and
\eqref{eq:mod2-cohomology-ring-KZ3}, we obtain
\begin{equation}
H^{\bullet}(K(\Z,2)\times K(\Z,3);\Z/2) \cong \Z/2[c_1,\tau_3, \sq^2\tau_3,\sq^4\sq^2\tau_3,\ldots],
\end{equation}
where $c_1$ and $\tau_3$ are the unique generators in degree $2$ and
$3$, respectively.  In our diagram, the red maps correspond to
$\text{Sq}^2 c_1 = c_1^2$, the blue maps to $\text{Sq}^2 \tau_3$ being a
generator of $H^{\bullet}(K(\Z,3);\Z/2)$, and the magenta maps to
$\text{Sq}^2(c_1 \tau_3) = c_1^2 \tau_3$. The differential depicted in the
$E^3$ page from $E_{7,0}^3$ in the diagram must be non-trivial by
comparing $\Omega^{\Spin}_{6}(K(\Z,2) \times K(\Z,3))$ with the result computed
with the Adams spectral sequence. We can then read off the bordism
groups in lower degrees, which we collect below in Table
\ref{tab:abelian-2grp-bordism-odd}. We piece together the cobordism
group
\begin{equation}
H_{I\Z}^{6}(MT(\Spin \times \U \times B\U)) \cong \Z \times \Z,
\end{equation}
which classifies anomalies for this symmetry type. The pair of
$\Z$-valued local anomalies anomalies just corrresponds to the usual
cubic $c_1^3$ and mixed gravitational $c_1 p_1$ anomalies associated
with the $\U$ 0-form symmetry. The presence of the 1-form symmetry
here plays no role in anomaly cancellation for this dimension.


\subsubsection{Even Postnikov class}

Now we turn to the case where $m=|\hat\kappa|$ is a non-zero even integer. 
Continuing from the cohomology calculation above, summarised in Eq.~(\ref{eq:cohom-BG}), there are two options for the extension
$e(\Z/2,\Z/m)$; either the trivial extension $\Z/2 \times \Z/m$ or the
non-trivial one $e(\Z/2, \Z/m) \cong \Z/(2m)$, which are non-isomorphic. By comparing the mod 2
cohomology calculated from applying the universal coefficient theorem
to \eqref{eq:cohom-BG} and the one calculated directly from the Serre
spectral sequence, one can show that the correct extension is the direct
product $\Z/2 \times \Z/m$.  The integral homology for $B|\G|$ is then
\begin{align}
  H_{\bullet} \left( B|\G|;\Z \right) &\cong \left\{ \Z,0,\Z,\Z/m, 0, \Z/m \times \Z/2, 0, e(\Z/6,\Z/m), \ldots \right\} \label{eq:homZ-BG}
\end{align}
and the
mod 2 homology is 
\begin{equation}
H_{\bullet}(B|\G|;\Z/2) \cong \left\{ \Z/2, 0, \Z/2, \Z/2, \Z/2, \Z/2\times \Z/2 , \ldots \right\}\, ,
\end{equation}
The $E^2$ page of the AHSS is then given by
Fig.~\ref{fig:AHSS-evenm}. We observe that there are 
non-vanishing differentials already on the $E^2$ page (which will
not be the case when we turn to the case of odd $m$). These
non-trivial $E^2$ differentials make the spectral sequence easier to compute, as follows.

\begin{figure}[htbp]
\centerline{\includegraphics[scale=0.85]{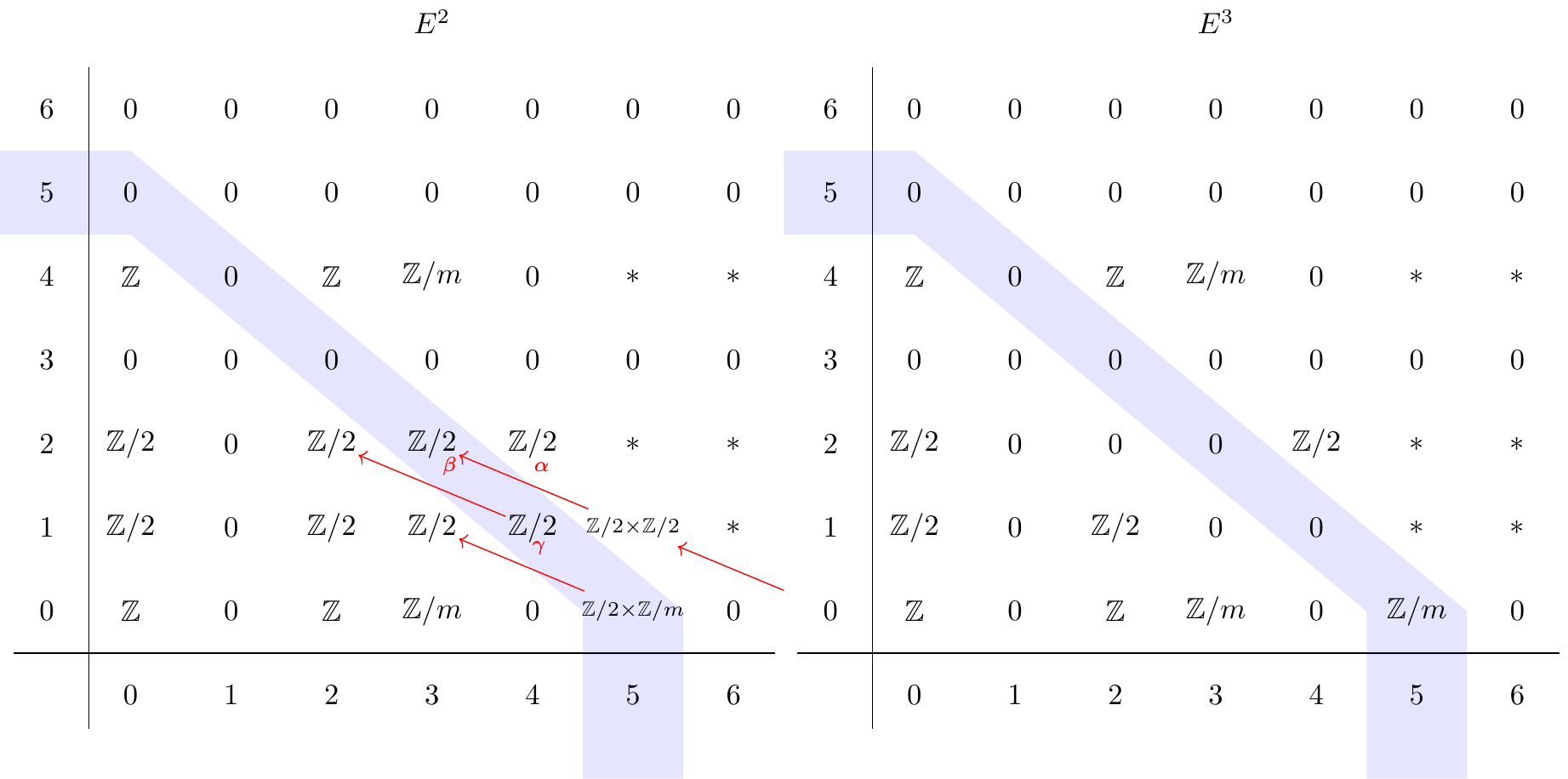}}
\caption[]{\label{fig:AHSS-evenm} The $E^2$ page of the AHSS for the
  fibration $\text{pt} \to B|\G| \to B|\G|$ when $\G$ is the 2-group
  $\U^{[0]} \times _{\hat{\kappa}} \U^{[1]}$ with $m:= |\hat{\kappa}|$ even.}
\end{figure}

The differentials $\alpha$ and $\beta$ on the second page are duals
$\widetilde{\sq^2}$ of the Steenrod square $\sq^2$, while the
differential $\gamma$ is a mod 2 reduction followed by the dual of
$\mathrm{Sq}^2$ \cite{teichner1993signature}. The Steenrod square dual
to $\beta$ sends a unique generator in $H^2(B|\G|;\Z/2)$ to a unique
generator in $H^4(B|\G|;\Z/2)$ ({\it c.f.} Appendix
\ref{app:mod-2-cohomology}), so $\beta$ must be the non-trivial map from
$\Z/2$ to $\Z/2$, killing off both factors. Similarly, the Steenrod
square dual to $\alpha$ acts on the unique generator of
$H^3(B|\G|;\Z/2)$ that comes from the generator
$\tau_3\in H^3(K(\Z,3);\Z/2)$ of the fibre, which we will also label by
$\tau_3$. The image is a generator of $H^5(B|\G|;\Z/2)$ that comes from
$\text{Sq}^2 \tau_3$ of $H^5(K(\Z,3);\Z/2)$. It generates the $\Z/2$
factor in the mod 2 cohomology that is a reduction from the $\Z/2$
factor in the integral cohomology, and not the $\Z/m$ factor since it
is the $\Z/2$ factor that arise from the fibre's contribution.
Therefore, both $\alpha$ and $\gamma$ must be non-trivial, with
$\ker \gamma \cong \Z/m$, as indicated in the $E^3$ page in
Fig. \ref{fig:AHSS-evenm}.

Before continuing, we pause here to emphasize the importance of using
(spin) cobordism, rather than just cohomology, to study 't Hooft anomalies in these theories.
At the level of cohomology, the existence of a non-trivial $\Z/2$-valued cohomology class 
$\text{Sq}^2 \tau_3 \in H^5(K(\Z,3);\Z/2)$ might suggest there is a non-trivial SPT phase on
5-manifolds $M_5$, or equivalently a non-trivial anomaly theory for the corresponding 4d theory,
with partition function
\begin{equation}
Z = \exp \left( \pi\ii \int_{M_5} \text{Sq}^2 \tau_3 \right),
\end{equation}
where $\text{Sq}^2\tau_3$ is pulled back to $M_5$. But, one can use Wu's relation
to trade the Steenrod square operation for a Stiefel--Whitney class~\cite{Wu1950}, {\em viz.}
\begin{equation}
  \label{eq:w2t3-assumed}
  Z =  \exp \left( \pi\ii \int_{M_5} w_2(TM_5)\cup \tau_3 \right).
\end{equation}
If we now restrict to $M_5$ being a spin manifold, then $w_2(TM_5)$ is trivial and we immediately
learn that this SPT phase is trivial. Of course, this `trivialisation' of the SPT phase corresponding
to the cohomology class in $H^5(K(\Z,3);\Z/2)$ is automatically captured by the spectral sequence computation for spin bordism,
by the non-triviality of the map $\gamma$ on turning from the second to the third page of the AHSS.

Continuing with the bordism computation, we now find that
the entries in the range $p+q \le 5$ stabilise on this $E^3$ page, whence we
can read off the spin bordism groups for 2-group symmetry
$\G = \U^{[0]} \times_{\hat{\kappa}} \U^{[1]}$ with even Postnikov class
$\hat{\kappa}$ up to degree-5, which are listed in Table
\ref{tab:abelian-2grp-bordism-odd}. There is still an undetermined
group extension in $\Omega^{\Spin}_3(B|\G|)$ but we will argue from the
physics point of view in Section \ref{sec:2d} that it must be the
non-trivial extension $\Z/(2m)$.  In summary, our results are those
given in Table~\ref{tab:abelian-2grp-bordism-odd}.

\begin{table}[h]
  \centering
  \begin{tabular}{|c|ccccccc|}
\hline
    $i$ & $0$ & $1$ & $2$ & $3$ & $4$ & $5$ &$6$\\
    \hline
    $\Omega^{\Spin}_i(B|\G|)$, $\hat{\kappa}\neq 0$ & $\Z$ & $\Z/2$ & $\Z \times \Z/2$ & $\Z/(2|\hat{\kappa}|)$ & $\Z$ & $\Z/|\hat{\kappa}|$ & $\Z \times \text{Tors}$  \\
    $\Omega^{\Spin}_i(B|\G|)$, $\hat{\kappa}=0$ & $\Z$ & $\Z/2$ & $\Z \times \Z/2$ & $\Z$ & $\Z\times \Z$ & $\Z$ & $\Z\times \Z \times \text{Tors}$  \\
\hline
  \end{tabular}
  \caption{Spin bordism groups for the twisted 2-group symmetry
    $\G=\U^{[0]}\times_{\hat{\kappa}}\U^{[1]}$. Here,
    $\text{Tors}$ denotes an undetermined pure torsion part. For
    comparison, we also include the bordism groups in the case
    $\hat\kappa=0$, which corresponds to a direct product of 0-form and
    1-form symmetries. One can thus see how, for both 2d and 4d
    quantum field theories, local anomalies in the untwisted case
    morph into global anomalies when the Postnikov class $\hat\kappa$ is
    turned on, thanks to a Green--Schwarz-like mechanism.  }
  \label{tab:abelian-2grp-bordism-odd}
\end{table}

\subsubsection{Odd Postnikov class} \label{ss:odd_post}

We now turn to the case where $m$ is odd, for which the spin bordism computation turns out to be rather more difficult, technically.
Firstly, when $m$ is odd, the extension $e(\Z/2, \Z/m)$ is unambiguously
$\Z/m \times \Z/2$ (which is, for odd $m$, isomorphic to $\Z/(2m)$). The integral homology of $B|\G|$ is as written in Eq.~(\ref{eq:homZ-BG}), and the mod 2 homology groups are now
\begin{align}
  H_{\bullet} \left( B|\G|;\Z/2 \right) &\cong \left\{ \Z/2, 0, \Z/2, 0, 0, \Z/2, \ldots \right\} \label{eq:homZ2-BG}
\end{align}
by the universal coefficient theorem. Now we feed these results into the
AHSS for the trivial point fibration  $\text{pt}\to B|\G|\to B|\G|$.
The $E^2$ has no non-trivial differentials in the range of interest. 

\begin{figure}[htbp]
\centerline{\includegraphics[scale=0.85]{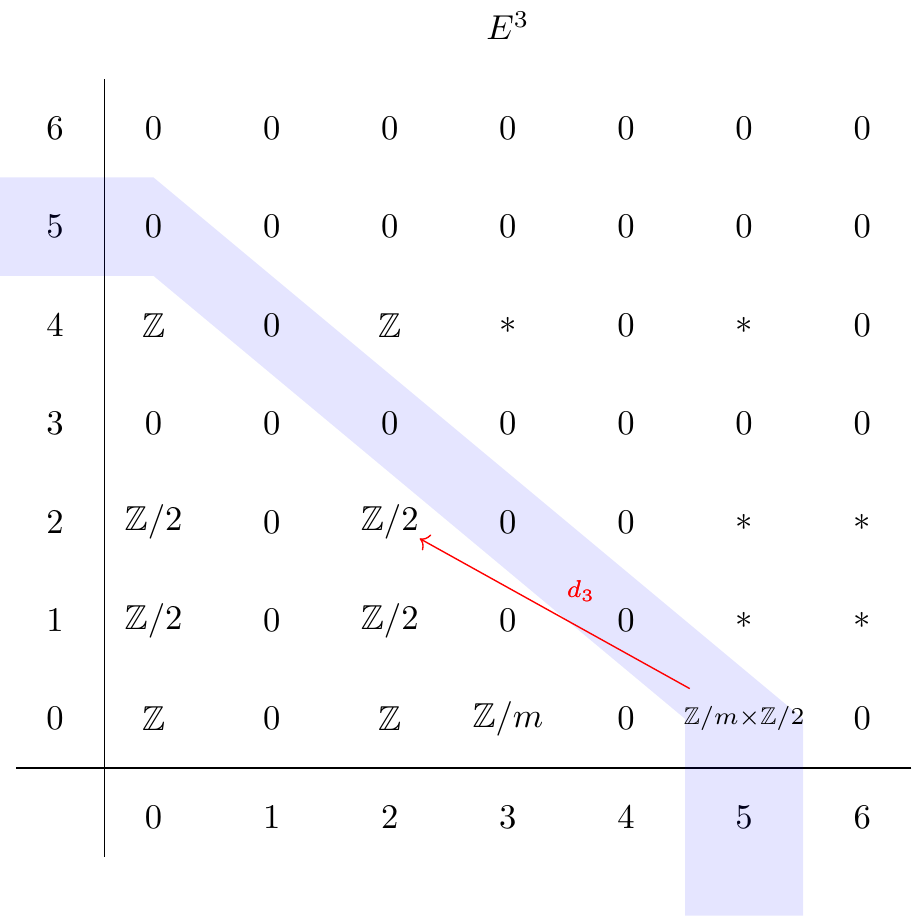}}
\caption[]{The $E^3$ page of the AHSS for the fibration
  $\text{pt} \to B|\G| \to B|\G|$ when $\G$ is the 2-group
  $\U^{[0]} \times _{\hat{\kappa}} \U^{[1]}$ with
  $m := |\hat{\kappa}|$ odd.\label{fig:AHSS-oddm} }
\end{figure}

The $E^3$ page is shown in
Fig. \ref{fig:AHSS-oddm}, 
with a single potentially non-trivial differential in the range of
interest, $\dd_3: E^3_{5,0} \to E^3_{2,2}$. If
this differential were trivial, then $\Omega^{\Spin}_5(B|\G|)$ would equal $\Z/m\times \Z/2$, and $\Omega^{\Spin}_4(B|\G|)$ would equal $\Z \times\Z/2$. 
The factor of $\Z/m$ in $\Omega^{\Spin}_5(B|\G|)$ would correspond to the global anomaly that we described in the previous Section, which we know to be valued in $\Z/m$ as in (\ref{eq:Zm_group}). The extra factor of $\Z/2$ would presumably correspond to a further global anomaly that we have not seen so far by physics arguments.
On the other hand, if this differential were the non-trivial map, then 
 $\Omega^{\Spin}_5(B|\G|)$ would equal $\Z/m$, agreeing precisely with the physics account, and $\Omega^{\Spin}_4(B|\G|)$ would just equal $\Z$.

It turns out that this all-important $d_3$ differential is {\em non}-trivial.
But being a differential on the third page, there are no straightforwardly-applicable formulae (analogous to the formulae in terms of Steenrod squares that are available on the second page) that we can use to compute it directly. Indeed, our usual AHSS plus ASS techniques are not sufficient to constrain this differential. This differential can nonetheless be evaluated using different arguments,\footnote{We are very grateful to Arun Debray for sharing this ingenious argument to compute this differential. Appendix~\ref{app:Arun} is written by Arun Debray. }
which we present in Appendix~\ref{app:Arun}. The gist of the argument is as follows.

The central character is a long exact sequence in bordism groups,\footnote{See Appendix E of Ref~\cite{Debray:2023yrs} for a related theorem. The proof of the theorem used in this paper, and that of Ref~\cite{Debray:2023yrs}, will appear in future work~\cite{DebrayNEW}.} analogous to the Gysin long exact sequence in ordinary homology, of the form
\begin{equation}
\dots \to \Omega^{\Spin}_d (S(V)) \to \Omega^{\Spin}_d(B|\G|) \to \Omega^{\Spin}_{d-2}((B|\G|)^{V-2}) \to \Omega^{\Spin}_d (S(V)) \to \dots
\end{equation}
Here $V \to B|\G|$ is the pullback bundle (which is rank 2) of the tautological line bundle $L \to B\U^{[0]}$ along the quotient map $q:B|\G| \to B\U^{[0]}$, $S(V) \to B|\G|$ is the associated sphere bundle of $V$, and $(B|\G|)^{V-2}$ is the Thom spectrum of the virtual bundle. This long exact sequence in bordism is extremely powerful: by proving that the groups  $\Omega^{\Spin}_5 (S(V))$ and $\Omega^{\Spin}_{3}((B|\G|)^{V-2})$ both lack free and 2-torsion summands, one learns that the group $\Omega^{\Spin}_5(B|\G|)$ in the middle, which is our bordism group of interest, also lacks 2-torsion. The $\Z/2$ factor discussed above is therefore absent (from which we learn that it must be killed by the differential $d_3$ which is therefore non-trivial).

Putting things together, we can thus extract all the spin bordism groups for 2-group symmetry
$\G = \U^{[0]} \times _{\hat{\kappa}} \U^{[1]}$, with odd Postnikov class
$\hat{\kappa}$, up to degree-5. 

\subsubsection{The free part}

When $m$ is non-trivial, regardless of its parity, we can also show
that the free part of $\Omega^{\Spin}_6(B|\G|)$ is given by $\Z$ coming
from the entry $E^r_{2,4}$. First, observe that the free part of the
$\Omega^{\Spin}_6(B|\G|)$, if there is one at all, can only come from this
entry as any other entry from the diagonal $p+q=6$ is either trivial
or pure torsion. Next, we need to show that $E^{\infty}_{2,4}$ is
non-trivial. The only differential that could kill it along the way is
$d_5:E^5_{7,0} \to E^5_{2,4}$. But this differential must be trivial
because $E^5_{7,0}$ is pure torsion (as $E^2_{7,0}$, given by
$H_7(B|\G|;\Z) \cong e(\Z/6,\Z/m)$ in \eqref{eq:homZ-BG}, is pure
torsion), and $\Hom(G,\Z) = 0$ if $G$ is pure torsion. Therefore,
there is one free factor in the diagonal $E^{\infty}_{p,q}$ with
$p+q=6$, and we obtain
\begin{equation}
\Hom\left( \Omega^{\Spin}_6(B|\G|), \Z \right) \cong \Z .
\end{equation}

From these results, we piece together the cobordism group that classifies anomalies for this global 2-group symmetry, valid for even and odd $m \neq 0$:
\beq
\boxed{
H^6_{I\Z}\left( MT \left( \Spin \times |\G|\right) \right) \cong \Z/m \times \Z\, .}
\eeq
In terms of the anomaly coefficients, the space of anomaly theories is classified by
\beq
(r,s) = \left(\frac{1}{6}\left(\mathcal{A}_3 - \mathcal{A}_{\mathrm{mixed}} \right) \mod m,  \mathcal{A}_{\mathrm{mixed}}  \right)\in \Z/m \times \Z\, , \qquad m=|\hat\kappa|\, ,
\eeq
in agreement with the results of the previous Subsection. To reiterate, the local anomaly associated with $c_1^3$ ({\em i.e.} whose coefficient is the sum of $U(1)^{[0]}$ charges cubed) becomes a global anomaly when the 0-form and 1-form symmetries are fused into a 2-group defined so that the mixed 't Hooft anomaly ($\sim f \wedge F \wedge F$) vanishes. The mixed gravitational anomaly associated with $p_1 c_1$ remains a $\Z$-valued local anomaly.

The result for the spin bordism group in degree-3 in Table~\ref{tab:abelian-2grp-bordism-odd} is also interesting, implying the presence of an novel global anomaly structure for the corresponding symmetry in two dimensions. We discuss this from a physics perspective in \S \ref{sec:2d}.

\section{Abelian 2-group enhancement in two dimensions}
\label{sec:2d}

In this Section, we consider 2d avatars of the 4d 2-group anomalies
that we have been discussing. In this lower-dimensional version we
will find some interesting differences.

In \S \ref{sec:QED-bordism-perspective}, we computed the spin bordism
groups for the 2-group symmetry
$\G=\U^{[0]} \times_{\hat{\kappa}} \U^{[1]}$. In particular, we obtained
$$\Omega_3^\Spin(B|\G|) \cong \Z/(2m)$$
where $m=|\hat{\kappa}|$. With the fourth bordism group given by $\Z$, we can
put things together to get the cobordism group
\begin{equation} H^4_{I\Z}\left( M \left( \Spin \times \G\right) \right) \cong
  \Z/(2m) \times \Z\, .
\end{equation}
This is similar to the result for $H^6_{I\Z}$ that was pertinent to
anomalies in 4d, except for the fact that the global anomaly is here
`twice as fine' as the $\Z/m$ anomaly in 4d.\footnote{The
  $\Z$-valued local anomaly here is simply the gravitational anomaly
  associated with $-\frac{1}{24}p_1(R) \Tr_{\bf{R}} 1 \subset \Phi_4$ in the
  degree-4 anomaly polynomial, which can always be cancelled by adding
  neutral fermions and so plays no further role in our discussion.}
That extra division by 2 will correspond to a subtle new global
anomaly associated with the spin structure, which will be our main
interest in this Section.

Before we discuss the global $\Z/(2m)$ anomaly in more depth, let us first discuss the physics interpretation of a 2-group symmetry $\G$ in two dimensions, which is rather different to the 4d case.
Recall that in 4d, the 1-form symmetry $\U^{[1]}$ was identified with the magnetic 1-form symmetry, with 2-form current proportional to $\star f$. But in general $d \geq 3$ spacetime dimensions, the magnetic symmetry is a $(d-3)$-form symmetry, and so there is no such symmetry in 2d.
There is, however, a trivially conserved `1-form symmetry' $\U^{[1]}_\text{top}$ whose 2-form current is simply $j_\text{top} = \mathrm{vol}_2$, the volume form. This symmetry does not act on any line operators in the theory and so one should not think of it as a physical 1-form symmetry -- but it will play a role in what follows.

Now suppose there is a global 0-form symmetry $G^{[0]}$ with background gauge field $A$ with curvature $F$. If we first cancel the pure gravitational anomaly by adding neutral fermions, the local anomalies in $G^{[0]}$ are captured by the degree-4 anomaly polynomial
\beq
\Phi_4 = \Tr_{\bf{R}} \left(\frac{F \wedge F}{8\pi^2} \right)\, .
\eeq
The usual 't Hooftian interpretation of $\Phi_4 \neq 0$ would be that, if one tries to gauge $G^{[0]}$, the anomaly `breaks' the symmetry $G^{[0]}$ in the quantum theory. However, Sharpe re-interprets $\Phi_4 \neq 0$ as indicating a weaker 2-group symmetry structure in the 2d quantum theory~\cite{Sharpe:2015mja}, corresponding to the extension
$B\U \hookrightarrow \G \twoheadrightarrow G^{[0]}$.
In other words, in 2d one can invoke the auxiliary $\U$ 1-form symmetry to trade an anomaly in a 0-form symmetry for a 2-group structure. This is of course just a simpler version of the 4d situation described in \S \ref{sec:U1U1-physics}, in which a mixed 't Hooft anomaly corresponding to a term $\sim f \wedge F \wedge F \supset \Phi_6$ was re-interpreted as signalling a weaker 2-group symmetry.

But in this Section we will see that, in the case that one needs a spin structure to define the field theory, it is {\em not always possible} to completely absorb the 't Hooft anomaly associated with $\Phi_4$ by a well-defined 2-group structure. Rather, there can be a residual $\Z/2$-valued anomaly left over, which one should interpret as an 't Hooft anomaly in the 2-group symmetry itself.

\subsection{`Spin structure anomalies' in two dimensions} \label{sec:spin-struct-anom-2d}

To see how this works, let us now be more precise and set $G^{[0]}=\U^{[0]}$. Letting $\{q_i\}$ denote the $\U^{[0]}$ charges of a set of (say, left-handed) chiral fermions, the anomaly polynomial is
\beq
\Phi_4 = \frac{1}{2} \mathcal{A}_{2}\, c_1 \cup c_1\, , \qquad \mathcal{A}_{2}:=\sum_i q_i^2\, ,
\eeq
where $c_1=F/2\pi$ is the first Chern class of the $\U^{[0]}$ bundle. By inflow, one can use the associated Chern--Simons form $I_3 = \frac{\mathcal{A}_{2}}{8\pi^2}A\wedge F$ to compute the variation of the effective action under the background gauge transformation $A\to A+ d\lambda^{(0)}$,
\beq
\delta S_0 = -\ii\frac{\mathcal{A}_{2}}{4\pi} \int_{M_2} \lambda^{(0)} F\, .
\eeq
Now, let us try to implement the philosophy above, and trivialise this anomaly by promoting $\U^{[0]}$ to a 2-group global symmetry by fusing with  $\U^{[1]}_\text{top}$. We turn on a background 2-form gauge  field $B^{(2)}$ that couples to the trivially conserved current $j_\text{top}= \mathrm{vol}_2$, via the coupling
\beq \label{eq:2dGS}
S_{\text{coupling}} = \ii \int_{M_2} \star j_\text{top} B^{(2)} = \ii\int_{M_2} B^{(2)}, 
\eeq
because $\star \mathrm{vol}_2=1$. 

We consider the $\U^{[0]}$ and $\U^{[1]}_\text{top}$ symmetries to be fused via the toric 2-group structure $\G=\U^{[0]} \times_{\hat{\kappa}} U(1)^{[1]}_{\text{top}}$, where  $\hat\kappa \in \Z$ is the Postnikov class. 
This prescribes the by-now-familiar transformation (\ref{eq:2-gp-gauge transformation}) on the background gauge fields, under which
\beq
\delta S_{\text{coupling}} = \ii\frac{\hat\kappa}{2\pi} \int_{M_2} \lambda^{(0)} F\, .
\eeq
We see that the local anomaly associated with the $F\wedge F$ term in $\Phi_4$ is completely removed iff
\beq \label{eq:SpinAnom}
\hat\kappa = \frac{\mathcal{A}_2}{2}.
\eeq
But, since the Postnikov class $\hat\kappa$ is neccessarily integral, this is possible only when $\mathcal{A}_2$ is an even integer.
Thus, there is an order 2 anomaly remaining when $\mathcal{A}_2$ is odd.

This chimes perfectly with our computation of the bordism group $\Omega_3^\Spin(B|\G|)$ in \S \ref{sec:QED-bordism-perspective}, that we quoted at the beginning of this Section, which implies that there is in general a mod $2|\hat\kappa|$ global anomaly. If one imagines fixing the 2-group structure, in other words fixing an integer $\hat\kappa$, then one can vary the fermion content and ask whether there is a global 't Hooft anomaly for the 2-group $\G$. We can think of the coupling term (\ref{eq:2dGS}) as a Green--Schwarz counterterm, which effectively shifts the anomaly coefficient by
\beq
\mathcal{A}_2 \to \mathcal{A}_2 + 2n\hat\kappa\, , \qquad n \in \Z, 
\eeq
and so the residual anomaly is clearly 
\beq
\mathcal{A}_2 \mod 2m\,, \qquad m = |\hat\kappa|.
\eeq
This is the $\Z/(2m)$-valued global anomaly in the 2-group symmetry that is detected by $\Omega_3^\Spin(B|\G|)$. Even if we choose the minimal value for $\hat\kappa$, a mod 2 anomaly persists iff
\beq
\sum_i q_i^2 = 1 \mod 2 \implies N_o = 1 \mod 2\, ,
\eeq
where $N_o$ is the number of fermions with odd $\U^{[0]}$ charges.

One can offer a different perspective on this anomaly by thinking about a generator for $\Omega_3^\Spin(B|\G|)$, which makes it more transparent how this mod 2 anomaly is related to the requirement of a spin structure. If we look back at the AHSS in Fig.~\ref{fig:AHSS-oddm}, we see that the factor of $\Z/2$ that ends up in $\Omega_3^\Spin(B|\G|)$ comes from the $E^2_{2,1}$ element, which stabilizes straight away from the $E^2$ page. Since $E^2_{2,1} = H_2\left(B|\G|; \Omega_1^\Spin(\pt)\right)$, this suggests a generator for this $\Z/2$ factor can be taken to be a mapping torus 
\beq
M^3 = S^2 \times S^1, 
\eeq
with $c_1(F) \in 2\Z+1$ on the $S^2$ factor, corresponding to an odd-charged monopole, and the non-bounding spin structure on the $S^1$, corresponding to a non-trivial element  of $\Omega_1^\Spin(\pt) \cong \Z/2$. This would suggest that the transformation by $(-1)^F$, which counts fermion zero modes, is anomalous on $S^2$ with an odd-charged monopole configuration for the background gauge field.\footnote{We do not believe, however, that this anomaly can be detected from the torus (in this case, just a circle) Hilbert space, using the methods set out in Ref.~\cite{Delmastro:2021xox}. This is because a system with one Weyl fermion with unit $\U$ charge (together with a Weyl fermion of opposite chirality to cancel the gravitational anomaly) will have an even number of Majorana zero modes, which allows construction of a $\Z/2$-graded Hilbert space even in the sector twisted by $-1\in \U$. (The failure to construct a $\Z/2$-graded Hilbert space is one sign of anomalies involving the spin structure, as studied in~\cite{Delmastro:2021xox}.)}

We can of course see this from an elementary calculation. For our set of left-handed chiral fermions with charges $\{q_i\}$, the index of the Dirac operator on $S^2$, which recall is $\mathrm{Ind}(i\slashed{D}_2) = n_L - n_R$ the number of (LH minus RH) zero modes of the Dirac operator, is equal to
\beq
\mathrm{Ind}(i\slashed{D}_2) = \int_{S^2} \Phi_2 = \int_{S^2} \Tr_{\bf{R}} \left(\frac{F}{2\pi}\right) = c_1(F) \sum_i q_i
\eeq
by the 2d Atiyah--Singer index theorem. Thus, the total number of zero modes $N:=n_L + n_R$, which is congruent to $n_L - n_R$ mod 2, satisfies
\beq
N = c_1(F) \sum_i q_i \mod 2\, .
\eeq
Choosing $c_1(F)$ to be odd on $S^2$, as we are free to do, we see that $(-1)^F$, which counts these zero modes, flips the sign of the partition function when 
\beq
\sum_i q_i = 1 \mod 2 \implies N_o = 1 \mod 2.
\eeq
For such fermion content, the 2-group symmetry and $(-1)^F$ are therefore equivalently anomalous. (From the 2-group perspective, the anomalous transformation corresponds simply to the element $e^{\ii \pi} \in \U^{[0]}$.\footnote{
The existence of this global $\Z/2$-valued anomaly, despite the fact the anomalous $\U$ gauge transformation is clearly connected to the identity, further evidences the fact that global anomalies do not require the existence of `large gauge transformations' captured by, say, $\pi_d(G)$ (when restricting the spacetime topology to be a sphere). This disagreement between homotopy- and bordism-based criteria for global anomalies was discussed in Ref.~\cite{Davighi:2020kok}. In this case, the global anomaly arises because an otherwise local anomaly is only incompletely dealt with by the Green--Schwarz mechanism.
})

This should be contrasted with the situation in four dimensions that we examined in \S \ref{sec:U1U1-physics}. In that case, the corresponding formula for the number of fermion zero modes, on say $M_4 = S^4$, is $\mathrm{Ind}(i\slashed{D}_4) = \frac{1}{2} c_1^2 \sum_i q_i^2$. But, for a spin 4-manifold, the Atiyah--Singer index theorem implies that $c_1^2$ is an {\em even} integer, and moreover cancelling the local mixed gravitational anomaly $\implies \sum_i q_i = 0 \implies \sum_i q_i^2 = 0\mod 2$. The index and therefore the number of zero modes is then necessarily even, meaning there is no analogous anomaly in $(-1)^F$. This is why the spin bordism calculation exposes `only' a mod $|\hat\kappa|$ anomaly in the 4d case, but the finer mod $2|\hat\kappa|$ in 2d.\footnote{A more mundane way to understand this difference is therefore simply as a consequence of the normalisation of the various characteristic classes on spin manifolds of the relevant dimension.}

\paragraph{Example: the 3-4-5-0 model.}
To furnish an explicit example of a theory featuring this irremovable anomaly,
  we consider the so-called ``3-4-5-0 model''; that is, a $\U$ gauge theory
  in 2d with two left-moving Weyl fermions $\psi_1$, $\psi_2$ with charges
  $3$ and $4$, and two right-moving Weyl fermions $\chi_1$, $\chi_2$ with
  charges $5$ and $0$. The theory is free of both gauge and
  gravitational anomalies, and so we can ask about its global symmetries. 
There is a somewhat trivial $\U_{\chi_2}$ global symmetry
  rotating the neutral fermion $\chi_2$ (with charge $1$). Na\"ively, there are two further $\U$ global symmetries, call them $\U_1$ and
  $\U_2$, that act on the remaining fermions, whose charges $Q_1$ and
  $Q_2$ can be chosen independently in one basis as follows:
  \begin{table}[h!]
    \centering
    \begin{tabular}{c||ccc}
      Field & $Q_{\text{gauge}}$ & $Q_1$ & $Q_2$\\
      \hline
      $\psi_1$ & $3$ & $3$ & $1$\\
      $\psi_2$ & $4$ & $4$ & $3$\\
      $\chi_1$ & $5$ & $5$ & $3$
    \end{tabular}
  \end{table}

\noindent
  One immediately sees that $Q_1$ coincides with the gauge
  charge. Thus, the $\U_1$ transformation can be undone by a gauge
  transformation. Thus, the correct global symmetry is a product
  $\U_2\times \U_{\chi_2}$ with charges given by:
\begin{table}[h!]
  \centering
  \begin{tabular}{c||ccc}
    Field & $Q_{\text{gauge}}$ & $Q_2$ & $Q_{\chi_2}$\\
    \hline
    $\psi_1$ & $3$ & $1$ & $0$\\
    $\psi_2$ & $4$ & $3$ & $0$ \\
    $\chi_1$ & $5$ & $3$ & $0$\\
    $\chi_2$ & $0$ & $0$ & $1$
  \end{tabular}
\end{table}

\noindent
The corresponding 't Hooft anomaly for each $\U$ factor of the global symmetry is odd as desired.  

\subsection{Non-spin generalisation}

We have argued that the order 2 anomaly just described, for a 2d theory with 2-group symmetry $\G=\U^{[0]} \times _{\hat{\kappa}} \U^{[1]}$, is intrinsically related to the requirement of a spin structure and the use of spin bordism. More broadly, all the 2-group anomalies that we study are sensitive to the full choice of tangential structure -- this is especially the case when computing the order of a finite global anomaly. 

To develop this idea further, we consider a theory defined with the same 2-group structure, but now using a Spin$_c$ structure. Recall that in $d$ dimensions the group Spin$_c(d)$ is $[\Spin(d)\times \U]/\Z/2$, defined by identifying the element $(-1)^F\in \Spin(d)$ with the order 2 element of a $\U$ symmetry, that we here take to be an auxiliary global symmetry used to define spinors; using a Spin$_c$ rather than a Spin structure allows us to define our fermionic theory on non-spin manifolds (indeed, all orientable manifolds up to and including dimension 4 admit a $\Spin_c$ structure).

\begin{table}[h]
  \centering
  \begin{tabular}{|c|ccccccc|}
\hline
    $i$ & $0$ & $1$ & $2$ & $3$ & $4$ & $5$ & $6$ \\
    \hline
    $\Omega^{\Spin_c}_i(\pt)$ & $\Z$ & $0$ & $\Z$ & $0$ & $\Z^2$ & $0$ & $\Z^2$ \\ 
    $\Omega^{\Spin_c}_i(B|\G|)$ & $\Z$ & $0$ & $\Z^2$ & $\Z/|\hat{\kappa}|$ & $\Z^3$ & $e(\Z/2|\hat{\kappa}|, \Z/|\hat{\kappa}|)$ & $\Z^4$ \\ 
\hline
  \end{tabular}
  \caption{$\Spin_c$ bordism groups for a point~\cite{bahri1987eta}, and for 2-group symmetry $\G=\U^{[0]} \times _{\hat{\kappa}} \U^{[1]}$.}
  \label{tab:SpinC-bord}
\end{table}

\begin{figure}[htbp]
\centerline{\includegraphics[scale=0.85]{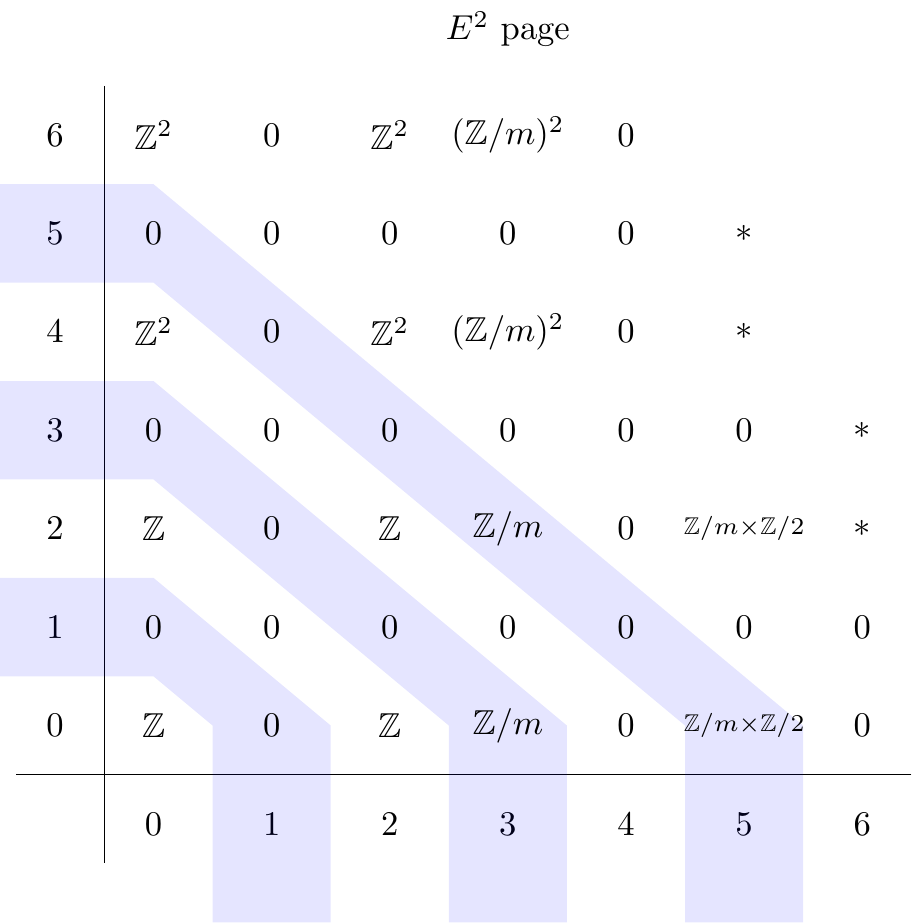}}
\caption[]{\label{fig:AHSS-SpinC} The $E^2$ page of the AHSS for spin$_c$ bordism, for the
  fibration $\text{pt} \to B|\G| \to B|\G|$, where $\G$ is the 2-group
  $\U^{[0]} \times _{\hat{\kappa}} \U^{[1]}$ with $m:= |\hat{\kappa}|$. Up to the $p+q=6$ diagonal there are no non-vanishing differentials whatsoever, and the entries shown in this region stabilize to the last page.}
\end{figure}

Having already computed the integral homology (\ref{eq:homZ-BG}) of $B|\G|$ in \S \ref{sec:QED-bordism-perspective}, and using the well-known result for the $\Spin_c$-bordism groups of a point~\cite{bahri1987eta} (see the first row in Table~\ref{tab:SpinC-bord}), 
it is straightforward to use the AHSS
\beq
E^2_{p,q} = H^p\big(B|\G|; \Omega^{\Spin_c}_q (\pt)\big) \implies \Omega^{\Spin_c}_q (B|\G|)
\eeq
to compute the $\Spin_c$ bordism groups of interest. The $E^2$ page is shown in Fig.~\ref{fig:AHSS-SpinC}. Things could not be much simpler; in the region of interest, all differentials are vanishing and all the elements on $E^2$ up to the $p+q=6$ diagonal stabilize to the last page. The bordism groups are recorded in the second row of Table~\ref{tab:SpinC-bord}. There is an extension problem that we cannot resolve using the AHSS for $\Omega_5$, but it is not relevant to our discussion here.

It is curious that $\Omega^{\Spin_c}_3(B|\G|)\cong\Z/m$ compared to $\Omega^{\Spin}_3(B|\G|)\cong \Z/(2m)$. One might na\"ively expect the same classification because the anomaly coefficient $\mathcal{A}_2$ is shifted by $2n\hat{\kappa}$ for any $n\in \Z$ from the 2-group transformation of the Green--Schwarz counter-term $\ii n\int_{M_2} B^{(2)} $. The puzzle can be resolved by correctly identifying the four integers in $\Omega^{\Spin_c}_4(B\U)\cong \Z^4$ that properly classify all perturbative anomalies, before turning on the 2-group structure. 

Let $F$ be the field strength for the background gauge field of the $\U$ 0-form global symmetry as before, and let $G$ denote the field strength associated to the $\Spin_c$ connection.\footnote{We use the convention for $\Spin_c$ connections, typical in the physics literature (see {\em e.g.}~\cite{Seiberg:2016rsg}), whereby $2G$ corresponds to a properly normalised $\U$ connection, and $G$ itself can be `half-integral'.
} Consider $N$ left-moving fermions with charges $\{q_i\}$ under the ordinary $\U$ global symmetry and with charges $\{g_i\}$ under the $\Spin_c$ connection. There is no constraint on the $q_i$, but the $g_i$ must all be odd integers. The anomaly polynomial associated to
this matter content is given by
\begin{equation}
\Phi_4 = \frac{1}{2}\sum_{i=1}^Nq_i^2 c_1(F)^2 + \sum_{i=1}^Ng_iq_i c_1(F)c_1(G) + \frac{1}{2}\sum_{i=1}^N g_i^2 c_1(G)^2 - N \frac{p_1(R)}{24},
\end{equation}
where $p_1(R)$ is the first Pontryagin class of the tangent bundle. The four anomaly coefficients that accompany each term in $\Phi_4$ are not strictly independent from one another, so they cannot be taken as labels for each factor of $\Z$ in $\Omega^{\Spin}_4(B\U)$. 

To extract a set of correct integer labels, we first choose an integral basis for the anomaly polynomial in terms of linear combinations of $c_1(F)^2$, $c_1(F) c_1(G)$, $c_1(G)^2$, and $p_1(R)$. One option is
\begin{equation}
\alpha_1 = c_1(F)^2,\quad\alpha_2 = \frac{1}{2}c_1(F^2)+c_1(F)c_1(G),\quad\alpha_3 =  \frac{1}{2}c_1(G)^2 - \frac{p_1(R)}{24}, \quad \alpha_4 = \frac{p_1(R)}{3}.
\end{equation}
We know that $\alpha_4$ is integral because $\int_{M_4}p_1(R)/3$ is just the signature of $M_4$ which is integral (and takes a minimal value of $1$, {\em e.g.} on $M_4=\C P^2$) for any orientable 4-manifold $M_4$. To see the integrality of $\alpha_3$, we apply the Atiyah--Singer index theorem for a $\Spin_c$ connection coupled to a single left-moving fermion with $q=0$ and $g=1$. Next, applying the same index theorem to a single left-moving fermion with $q=1$ and $g=1$ tells us that $\alpha_2$ is an integer.  Finally, the class $\alpha_1$ is integral by definition, and indeed it takes a minimal value of $1$. To see this, take $M_4 = \CPP$ with $F$ such that $\int F/2\pi = 1$ on a $\C P^1$ subspace. Hence $c_1(F)$ is the generator $x$ of $H^2(\CPP;\Z)$, so $c_1(F)^2=x^2=1\in H^4(\CPP; \Z)$ generates the top cohomology, which evaluates to 1 on the fundamental class $[\CPP]$.

In terms of these basis elements, the anomaly polynomial can be written as $\Phi_4 = \sum_{i=1}^4n_i\alpha_i$, where the variables $n_i$ are
\begin{equation}
n_1 = \frac{1}{2}\sum_{i=1}^Nq_i(q_i-g_i), \quad n_2 = \sum_{i=1}^N q_ig_i, \quad n_3 = \sum_{i=1}^Ng_i^2, \quad n_4 = \frac{1}{8}\sum_{i=1}^N(g_i^2-1),
\end{equation}
and our arguments above tell us that each $n_i$ must be an integer. This is clearly the case for $n_2$ and $n_3$. To see that $n_1$ is an integer, recall that $g_i$ are all odd. So when $q_i$ is odd, $(q_i-g_i)$ is even, and vice versa. So their product must be even. Similarly, to see that $n_4$ is an integer, we observe that $g_i-1$ and $g_i+1$ are both even and differ by 2, so one of them must be divisible by $4$. Hence their product is divisible by $8$. 

When we turn the 2-group structure back on, the shift in anomaly coefficients induced by the Green--Schwarz $S \supset \ii n\int_{M_2} B^{(2)}$ counter-term amounts to a shift 
\beq
n_1 \mapsto n_1 + n\hat{\kappa}.
\eeq 
It is now obvious that the local anomaly labelled by the integer $n_1$ is truncated down to a mod $|\hat{\kappa}|$ anomaly.
If one first takes care to cancel the anomalies involving the $\Spin_c$ connection, which amounts to choosing charges $\{g_i\}$ such that $n_2$, $n_3$, and $n_4$ vanish, then we find that the 2-group structure can indeed be used to completely capture the remaining 't Hooft anomaly in the ordinary $\U^{(0)}$ 0-form global symmetry, in the sense proposed by Sharpe~\cite{Sharpe:2015mja}. This should be contrasted with the residual order 2 anomaly (\ref{eq:SpinAnom}) that cannot be cancelled in this way for the theory defined with a spin structure.\footnote{At the level of  (\ref{eq:SpinAnom}), in the $\Spin_c$ case we can think of the charges $q_i$ for the ordinary $\U^{(0)}$ 0-form symmetry as being `twisted' by the $\Spin_c$ charges $\{g_i\}$, which are necessarily non-zero. This makes the `effective' anomaly coefficient $2n_1$, analogous to $\mathcal{A}_2$ appearing in (\ref{eq:SpinAnom}), be always even. }

\section{Conclusion and outlook}
\label{sec:conclusion-outlook}

In this paper, we explored anomalies in 2-group global symmetries $\G$ in
field theories in $d=4$ and $d=2$ dimensions using the cobordism
classification. We compute the bordism
groups $\Omega^H_{d+1}(B|\G|)$ and $\Omega^H_{d+2}(B|\G|)$ of the classifying
space $B|\G|$ of the nerve of $\G$, with the tangential structure
$H = \SO$ or $\Spin$. The torsion part of the former group gives us
the global anomaly while the free part of the latter gives the
perturbative anomaly. Throughout the work, we focus on 2-groups $\G$ whose
0-form and 1-form symmetry parts are compact, connected abelian groups {\em ergo} tori, which we refer to as `toric 2-groups'.

As a warm-up example, we first looked at Maxwell's theory in 4d, on a general orientable manifold (not necessarily spin). The
2-group symmetry in this theory contains only the 1-form symmetry
part, being a product of the $\U$ electric 1-form symmetry and the
$\U$ magnetic 1-form symmetry. The familiar mixed anomaly~\cite{Gaiotto:2014kfa} between
these two $\U$ factor can be given in terms of a bordism invariant in the free part of $\Omega_6$. There are also three global anomalies, probed by the torsion subgroup of $\Omega_5$, which can only be seen on non-spin manifolds and can be interpreted as mixed anomalies between the 1-form symmetries and gravity. We
also studied anomalies when one of the $\U$ factor breaks explicitly
to its cyclic subgroup due to the presence of a charged scalar, and
related them to anomalies in pure Maxwell's theory via anomaly
interplay.

When both the 0-form and 1-form symmetry constituents of $\G$ are
$\U$, with the Postnikov class given by $\hat{\kappa}\in \Z$, we used a variety of techniques to prove that
the relevant spin bordism groups are
\begin{equation}
\Omega^{\Spin}_5(B|\G|) \cong \Z/|\hat{\kappa}|, \quad \Hom \left( \Omega^{\Spin}_6(B|\G|), \Z \right) \cong \Z\;.
\end{equation}
The free part of the sixth bordism group signals the existence of a
perturbative anomaly in 4d, which is just the mixed anomaly
between the $\U$ 0-form symmetry and gravity. On the other hand, the
perturbative cubic anomaly for $\U$ 0-form symmetry that would be
there were $\hat{\kappa}$ to vanish now transmutes into a discrete global
anomaly captured by the torsion part of the 5th bordism group. These
cobordism computations are in perfect agreement with 
the physics argument put forward by
C{\'o}rdova, Dumitrescu, and Intriligator~\cite{Cordova:2018cvg}.

We next discussed the fate of anomalies for a $\U$ 0-form symmetry in
2d, for both spin and non-spin manifolds. We revisited the
construction, originally due to Sharpe~\cite{Sharpe:2015mja}, whereby
this 0-form symmetry is enhanced to a 2-group symmetry by mixing it
with a $\U$ 1-form symmetry whose current is given by the top form on
the underlying manifold. When the anomaly coefficient $\mathcal{A}_2$
for the pure $\U$ anomaly is even, we found that, by choosing the
Postnikov class of the enhanced 2-group to be
$\hat{\kappa}=\mathcal{A}_2/2$, the anomaly can be fully absorbed away. On
the other hand, in the presence of the spin structure, when
$\mathcal{A}_2$ is odd and restricting to spin-manifolds, the most one
can do is reduce the anomaly down to an order 2 discrete anomaly,
that we intrepret as an anomaly in the spin structure.

There are many future directions to pursue that are
natural extensions of this work. One obvious alley is to look at more
general classes of 2-group beyond our `toric' assumption, {\it e.g.} 2-groups whose 0-form symmetry is a
non-abelian Lie group, and/or 2-groups whose 0-form or 1-form symmetry is
a finite group
\cite{Benini:2018reh,DelZotto:2020sop,Apruzzi:2021vcu,Bhardwaj:2021wif,Apruzzi:2021mlh,Lee:2021crt,
  DelZotto:2022fnw, DelZotto:2022joo,Bhardwaj:2022scy,Carta:2022fxc}. 
One can rigorously analyse anomalies in these 2-group symmetries using the cobordism
classification. Doing so involves an extension of the mathematical machinery 
developed in this paper, and would offer a non-trivial check on 
other methods of seeing the anomalies that have been explored
recently in the literature \cite{Bhardwaj:2022dyt}.
One can also study
how these anomalies relate to the anomalies studied in this paper
through a generalised version of anomaly interplay.\footnote{A version of anomaly interplay valid for theories with 2-group symmetry can be sketched as follows. In principle, one ought to start by defining a smooth 2-homorphism between a pair of 2-groups $\G$ and $\G^\prime$, which would induce a map between the corresponding bordism spectra. But our description of tangential structures in \S \ref{sec:cobord_2grp} suggests a shortcut: it is good enough to start with an ordinary (1-)homomorphism between the topological groups $\pi:|\G| \to |\G^\prime|$, discarding all information not captured by the nerves. This induces a map between cobordism groups going the other way. These ideas will be developed in future work. \label{foot:interplay}
} 

Another avenue to pursue concerns a higher-dimensional generalisation 
of the 2d 2-groups discussed in \S~\ref{sec:2d}, as follows. We saw that one
can cancel (almost) all pure $\U$ anomalies in 2d by enhancing the $\U$ 0-form symmetry to a
2-group by making use of the top-form $\U$ current. One then wonders
whether there is an analogous story for `top-group' symmetry
structures in higher dimensions, {\it e.g.} 4-group symmetry in 4d,
in which a `topological 3-form symmetry' with $j_4 = \text{vol}_4$ is
fused with the ordinary 0-form symmetry to kill the cubic
anomaly. Moreover, the remnant order 2 anomaly in the 2d case is a
direct result of the non-trivial role played by the spin structure. It
remains to be seen whether such an interplay between `top-group'
symmetries and spin-structure anomalies persists in higher dimensions.

For all these anomalies, there is of course a dual interpretation concerning symmetry protected topological (SPT) phases of matter that are protected by these generalised symmetry types, which are rigorously classified by cobordism. Furthermore, it is conceivable that by extending the bordism computations to more generalised symmetries one could uncover novel anomalies/phases, and suggest corresponding new constraints on the dynamics of theories with such symmetry.

\acknowledgments{ We are deeply grateful to Arun Debray for several stimulating discussions,
and for supplying a mathematical argument that enabled us to complete this paper.
We thank Pietro Benetti Genolini, Philip Boyle
  Smith, and Ben Gripaios for helpful discussions. We are
  also grateful to I\~naki Garc\'ia Etxebarria and Sakura
  Sch\"afer-Nameki for asking questions that partly motivated this
  project. This work was partially carried out while both authors were
  at DAMTP, University of Cambridge, in which time NL was supported by
  David Tong’s Simons Investigator Award, and we were both supported
  by the STFC consolidated grant ST/P000681/1. JD has received funding
  from the European Research Council (ERC) under the European Union’s
  Horizon 2020 research and innovation programme under grant agreement
  833280 (FLAY), and by the Swiss National Science Foundation (SNF)
  under contract 200020-204428. NL is supported by the Royal Society of
  London and by the STFC consolidated grant in `Particles, Strings and
  Cosmology’ number ST/T000708/1.}

\appendix

\section{Additional cohomology computations}
\label{app:some-cohom-comp}

In this Appendix, we collect results and computations for various cohomology groups.

\subsection{Cohomology groups of Eilenberg--Maclane spaces}
\label{app:summ-class-results}

In this Subsection, we collect classic results for cohomology groups
of various Eilenberg--Maclane spaces $K(G,n)$. The mod 2 cohomology
results below are obtained by Serre in the classic paper
\cite{Serre1953aP}.

\subsubsection*{\underline{$K(\Z/2,2)$}}

The mod 2 cohomology ring of an Eilenberg--Maclane space $B^2\Z/2=K(\Z/2,2)$ is
given by
\begin{equation}
\label{eq:mod2-cohomology-KZ2-2}
H^{\bullet}(K(\Z/2,2);\Z/2) \cong \Z/2[u_2,\sq^1 u_2,\sq^2\sq^1u_2,\ldots, \sq^{2^k}\sq^{2^{k-1}}\ldots\sq^2\sq^1u_2,\ldots],
\end{equation}
where $u_2$ is the unique generator in $H^2$.

\subsubsection*{\underline{$K(\Z,2)$}}

The classifying space $B\U$ of $\U$ is a $K(\Z,2)$ space. We have
\begin{equation}
  \label{eq:integral-cohomology-of-BU1}
  H^{\bullet}(K(\Z,2);\Z) = H^{\bullet}(B\U;\Z) \cong \Z[c_1],
\end{equation}
where $c_1$ is the unique generator in $H^2$ called the universal first Chern class. The mod 2 version is
given simply by
\begin{equation}
\label{eq:mod2-cohomology-of-BU1}
H^{\bullet}(K(\Z,2);\Z/2) \cong \Z/2[c_1],
\end{equation}
where $c_1$ is now the mod 2 reduction of the universal first Chern
class. The action of the Steenrod squares are given by $\sq^1 c_1 =0$,
$\sq^2 c_1=c_1^{2}$ which follows directly from the axioms.

\subsubsection*{\underline{$K(\Z,3)$}}

The mod 2 cohomology ring of an Eilenberg--Maclane space $K(\Z,3)$ is
given by
\begin{equation}
  H^{\bullet} \left( K(\Z,3);\Z/2 \right) \cong \Z/2[\tau_3, \text{Sq}^2\tau_3, \text{Sq}^4\text{Sq}^2\tau_3, \text{Sq}^8\text{Sq}^4\text{Sq}^2\tau_3,\ldots],
  \label{eq:mod2-cohomology-ring-KZ3}
\end{equation}
where $\tau_3$ is the unique generator in $H^3(K(\Z,3);\Z/2)$.

Integral cohomology groups of $K(\Z,3)$ can be obtained by applying
the universal coefficient theorem to the homology groups computed in
Ref. \cite{breen2014derived}. The resulting cohomology groups are
shown in Table \ref{tab:integral-cohomology-KZ3} below.
\begin{table}[h]
  \centering
  \begin{tabular}{|c|ccccccccc|}
    \hline
    $n$ & $0$ & $1$& $2$& $3$& $4$& $5$& $6$& $7$& $8$\\
    \hline
    $H^n(K(\Z,3);\Z)$ & $\Z$ & $0$& $0$& $\Z$& $0$& $0$& $\Z/2$& $0$& $\Z/3$\\
    \hline
  \end{tabular}
  \caption{Integral cohomology groups up to degree 8 of $K(\Z,3)$}
  \label{tab:integral-cohomology-KZ3}
\end{table}

\subsection{Mod 2 cohomology ring of the classifying space of an abelian 2-group}
\label{app:mod-2-cohomology}

In this Subsection, we use the Serre spectral sequence to calculate
the mod 2 cohomlogy ring of $B|\G|$ when $\G$ is the abelian 2-group
$\U^{[0]}\times_{\hat{\kappa}} \U^{[1]}$ studied in \S
\ref{sec:QED-bordism-perspective} with the Postnikov class $\hat{\kappa}$.

The cohomological Serre spectral sequence induced by the fibration
$K(\Z,3) \to B|\G| \to K(\Z,2)$ is
\begin{equation}
  E_2^{p,q} = H^p \left( K(\Z,2);H^q \left( K(\Z,3);\Z/2 \right) \right) \Rightarrow H^{p+q} \left( B|\G|;\Z/2 \right).
\end{equation}
The entries on the $E_2$ page can be computed from the results given
in \S \ref{app:summ-class-results}. There is no non-trivial
differentials on the $E_2$ or $E_3$ pages in the range of degrees we
are interested in, and the $E_2$ entries in this range propagate to
the $E_4$ page, shown in Fig. \ref{fig:mod2SerreSSoddm}. As explained
in \S \ref{sec:QED-bordism-perspective}, the differentials $\alpha$ and
$\beta$ on the $E_4$ page are given by contraction with the Postnikov
class $\hat{\kappa}$. Thus, they are non-trivial if and only if
$\hat{\kappa}$ is odd. In this case, the resulting $E_5$ page is shown on
the right-hand side of Fig. \ref{fig:mod2SerreSSoddm}. When
$\hat{\kappa}$ is even, the entries stabilise already on the The entries in
the range of interest stabilise on this page. On the other hand, when
$\hat{\kappa}$ is even, the entries stabilise already on the $E_4$ page. In
either case, we can read off the mod 2 cohomology groups to be
\begin{table}[h]
  \centering
  \begin{tabular}{|c|ccccccc|}
    \hline
    $n $ & $0$& $1$& $2$& $3$& $4$& $5$& $6$\\
    \hline
    $H^n(B|\G|;\Z/2)$, \quad $\hat{\kappa}$ odd  & $\Z/2$ & $0$& $\Z/2$& $0$& $0$& $\Z/2$& $\Z/2$\\
    $H^n(B|\G|;\Z/2)$, \quad $\hat{\kappa}$ even & $\Z/2$ & $0$& $\Z/2$& $\Z/2$& $\Z/2$& $\Z/2\times\Z/2$&$\Z/2 \times \Z/2$\\
    \hline
  \end{tabular}
  \caption{The mod 2 cohomology groups of $B|\G|$ where $\G$ is the 2-group $\U^{[0]}\times _{\hat{\kappa}} \U^{[1]}$ with the Postnikov class $\hat{\kappa}$}
  \label{tab:mod2-cohomology-groups-U1U1}
\end{table}

\begin{figure}[htbp]
\centerline{\includegraphics[scale=0.8]{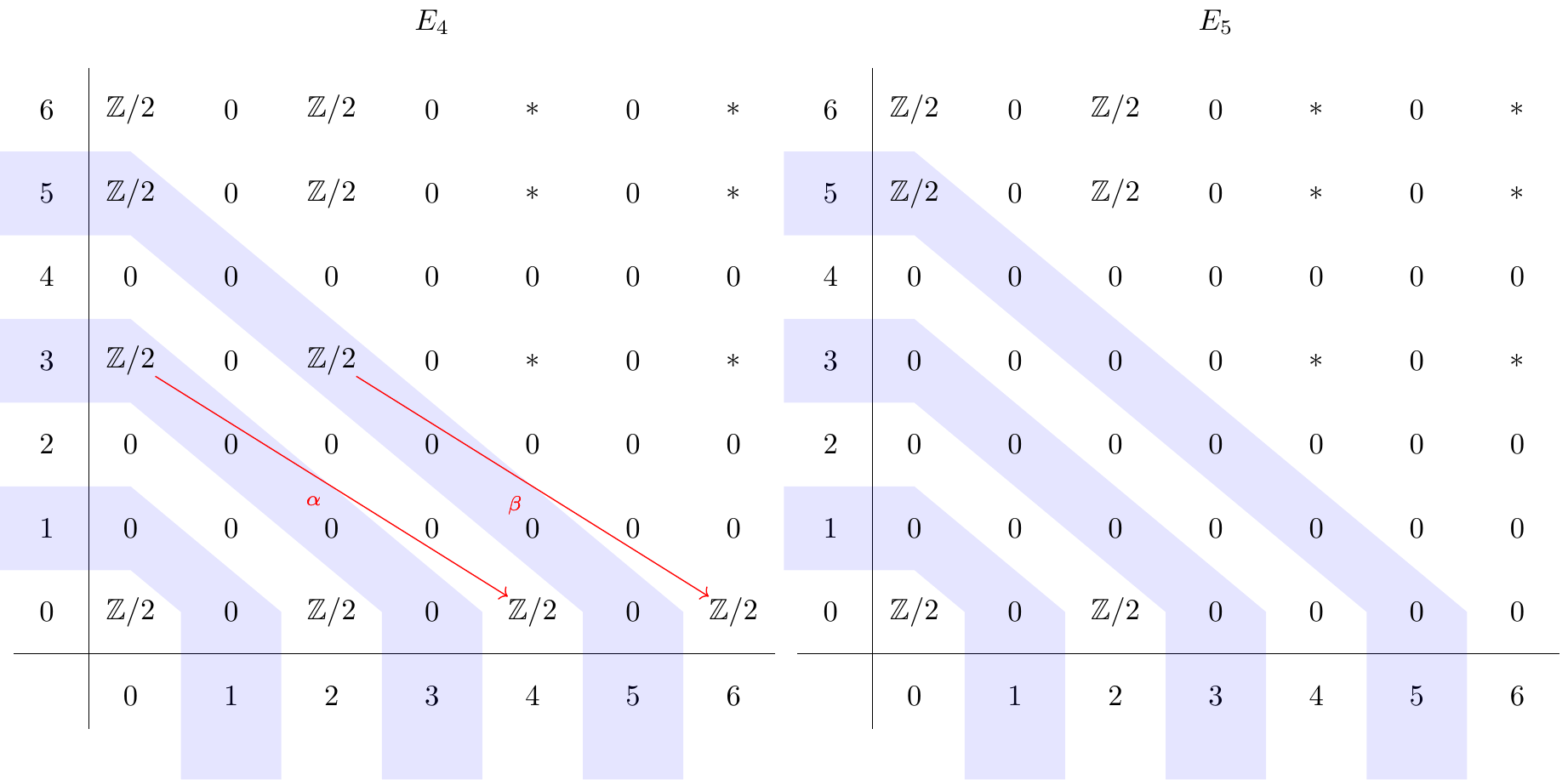}}
\caption[]{\label{fig:mod2SerreSSoddm} The Serre spectral sequence for
  the mod 2 cohomology for $B|\G|$ when
  $\G = \U^{[0]} \times _{\hat{\kappa}} \U^{[1]}$ with an odd Postnikov class
  $\hat{\kappa}$.}
\end{figure}

On the pages of the Serre spectral sequence, there are non-trivial
$\sq^2$ actions on the mod 2 cohomology rings of the fibre $K(\Z,3)$
and the base $K(\Z,2)$ that converge to $\sq^2$ actions on
$H^{\bullet}\left( B|\G|;\Z/2 \right)$ (Theorem 6.15 of
\cite{McCleary2000aB}). In particular, in the even $\hat{\kappa}$ case, we
have
\begin{equation}
  \begin{split}
  \sq^2:  H^2 \left( B|\G|;\Z/2 \right) & \to H^4 \left( B|\G|;\Z/2 \right)\\
    c_1 &\mapsto c_1^2
  \end{split}
\end{equation}
and
\begin{equation}
  \begin{split}
  \sq^2:  H^3 \left( B|\G|;\Z/2 \right) & \to H^5 \left( B|\G|;\Z/2 \right)\\
    \tau_3 &\mapsto \sq^2 \tau_3
  \end{split}
\end{equation}
where $c_1$, $\tau_3$, $c_1^2$, and $\sq^2 \tau_3$, are now taken to be the
unique generators for $H^i(B|\G|;\Z/2)$, $i=2,3,4,5$, respectively.

\section{Additional bordism calculations}
\label{sec:bordism-calculation}

\subsection{Maxwell} \label{app:maxwell}

Here we compute the bordism groups for the symmetry type \beq H= \SO \times
\G, \qquad B|\G|=B^2\U_e \times B^2\U_m\, , \eeq relevant to `free Maxwell
theory', {\em i.e.} $\U$ gauge group with no dynamical matter, which
has intact electric and magnetic $\U$ 1-form global symmetries,
discussed in \S \ref{sec:Maxwell-revisited}. Since
$\U \cong S^1 \cong B\Z$, the classifying space for a $U(1)$ 1-form symmetry
$B^2\U \cong K(\Z, 3)$ is an Eilenberg--Maclane space whose cohomology is
known. Using the K\"unneth theorem we find the integral homology of
$B|\G|= B^2\U \times B^2\U$, as recorded in Table~\ref{tab:MAXhom}. We also
need the oriented bordism groups of a point~\cite{10.2307/1970136}, which
are recorded in Table~\ref{tab:SO-bord}.

\begin{table}[h]
  \centering
  \begin{tabular}{|c|cccccccc|}
    \hline
    $i$ & 0 & 1 & 2 & 3 & 4 & 5 & 6 & 7\\
    \hline
	$H_i(K(\Z,3);\Z)$ & $\Z$ & $0$ & $0$ & $\Z$ & $0$ & $\Z/2$ & $0$ & $\Z/3$ \\
	$H_i(K(\Z,3);\Z/2)$ & $\Z/2$ & $0$ & $0$ & $\Z/2$ & $0$ & $\Z/2$ & $\Z/2$ & $0$ \\
\hline
   $H_i((B^2\U)^2;\Z)$ & $\Z$ & $0$ & $0$ & $\Z^2$ & $0$ & $(\Z/2)^2$ & $\Z$ & $(\Z/3)^2$ \\
   $H_i((B^2\U)^2; \Z/2)$ & $\Z/2$ & $0$ & $0$ & $(\Z/2)^2$ & $0$ & $(\Z/2)^2$ & $(\Z/2)^3$ & $0$ \\
\hline
  \end{tabular}
  \caption{Integral and mod 2 homology groups of $(B^2\U)^2=\left(K(\Z,3)\right)^2$, relevant to Maxwell theory.}
  \label{tab:MAXhom}
\end{table}

\begin{table}[h]
  \centering
  \begin{tabular}{|c|ccccccccc|}
    \hline
    $i$ & $0$ & $1$ & $2$ & $3$ & $4$ & $5$ & $6$ & $7$ & $8$ \\
    \hline
    $\Omega^{\SO}_i(\pt)$ & $\Z$ & $0$ & $0$ & $0$ & $\Z$ & $\Z/2$ & $0$ & $0$ & $\Z^2$ \\
    \hline
  \end{tabular}
  \caption{Oriented bordism groups for a point~\cite{10.2307/1970136}.}
  \label{tab:SO-bord}
\end{table}

We next write down the AHSS for oriented bordism associated to the
trivial fibration of $|\G|$ by a point, for which the second page is
given by \beq E^2_{p,q} = H_p\left((B^2\U)^2; \Omega_q^{\SO}(\pt) \right)\,
.  \eeq This is shown in Fig.~\ref{fig:AHSS_Maxwell}.

\begin{figure}[h]
\centering
  \includegraphics[width=0.7\textwidth]{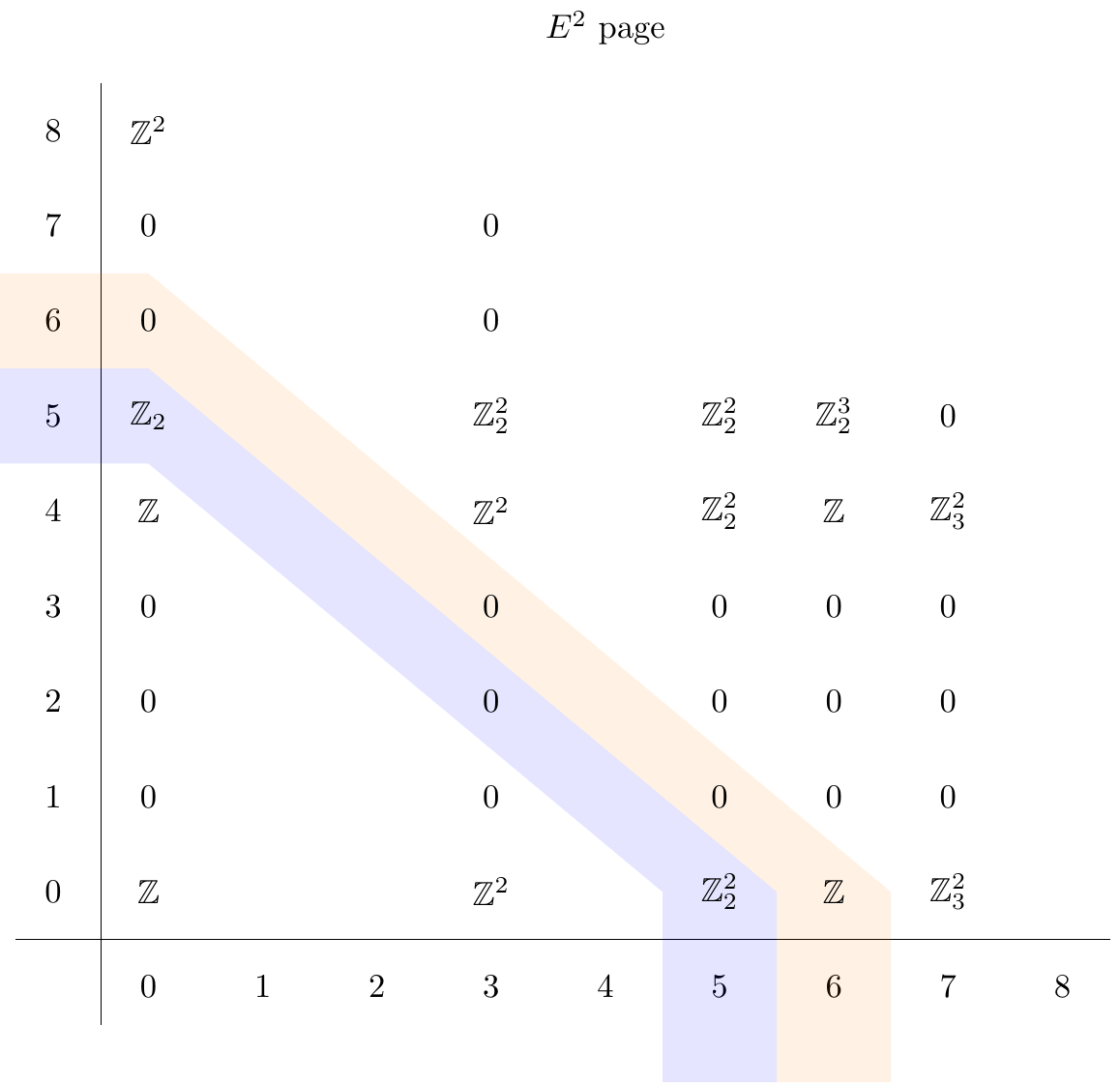}
\caption{The second page of the AHSS for $\Omega^{\SO}_\bullet\left((B^2\U)^2 \right)$.    \label{fig:AHSS_Maxwell} }
\end{figure}

We want to read off the fifth and sixth bordism groups, which receive contributions from the $p+q=5$ (blue) and $p+q=6$ (orange) diagonals respectively. For the sixth bordism group, there is a single factor $E^2_{6,0}=\Z$ to consider. The only possible differential is $d^5:E^5_{6,0} \to E^5_{0,5}\cong \Z/2$, whose kernel is obviously isomorphic to $\Z$, meaning that $E^\infty_{6,0}=\Z$. Thus, 
\beq
\Omega^\SO_6\left((B^2\U)^2 \right) = \Z \, .
\eeq
For the fifth bordism group, the $E^2_{5,0}=(\Z/2)^2$ element stabilizes to the last page, because the only possible differential is $d^4:E^4_{5,0} \to E^4_{0,4} :\Z/2 \mapsto \Z$ which is therefore the zero map. For the $E^2_{0,5}$ element, it is again the $d^5:E^5_{6,0} \to E^5_{0,5}$ map that is important. We can in fact deduce that this differential must be zero simply because it hits the zeroth column which records the bordism groups of a point, and since there is a canonical split $\Omega^\SO_\bullet(X)\cong \Omega^\SO_\bullet (\pt) \oplus \widetilde{\Omega}^\SO_\bullet(X)$ for $X$ connected, where the second factor is the reduced bordism of $X$, the $E_{0,5}$ factor must stabilize to the last page.\footnote{Equivalently, we could have used the AHSS for reduced bordism to get the same result.} This also makes it obvious that there can be no non-trivial extension, and so we can read off
\beq
\Omega^\SO_5\left((B^2\U)^2 \right) = (\Z/2)^3 \, .
\eeq
It is easy to read off the lower-degree bordism groups too, and we summarize the results in Table~\ref{tab:Maxwell-bord}.

\begin{table}[h]
  \centering
  \begin{tabular}{|c|ccccccc|}
\hline
    $i$ & $0$ & $1$ & $2$ & $3$ & $4$ & $5$ & $6$  \\
    \hline
    $\Omega^{\SO}_i\left((B^2U(1))^2 \right)$ & $\Z$ & $0$ & $0$ & $\Z^2$ & $\Z$ & $(\Z/2)^3$ & $\Z$ \\ 
\hline
  \end{tabular}
  \caption{Oriented  bordism groups up to degree 6, for a pair of $\U$ 1-form symmetries.}
  \label{tab:Maxwell-bord}
\end{table}

\subsection{Scalar QED with charge-2 boson } \label{app:charge2}

Next we compute the bordism groups for 
\beq
H= \SO \times \G^\prime, \qquad B|\G^\prime| = B^2\Z/2_e \times B^2\U_m\, ,
\eeq
relevant to QED with a charge-2 boson, which has the effect of breaking the $\U_e$ 1-form symmetry down to a discrete $\Z/2_e$ subgroup. Using the K\"unneth theorem we find the integral homology of $B^2\Z/2 \times B^2\U$, as recorded in Table~\ref{tab:SQEDhom}.

\begin{table}[h]
  \centering
  \begin{tabular}{|c|cccccccc|}
    \hline
    $i$ & 0 & 1 & 2 & 3 & 4 & 5 & 6 & 7\\
    \hline
	$H_i(B^2 \Z/2;\Z)$ & $\Z$ & $0$ & $\Z/2$ & $0$ & $\Z/4$ & $\Z/2$ & $\Z/2$ & $\Z/2$ \\
	$H_i(B^2\Z/2;\Z/2)$ & $\Z/2$ & $0$ & $\Z/2$ & $\Z/2$ & $\Z/2$ & $(\Z/2)^2$ & $(\Z/2)^2$ & $(\Z/2)^2$ \\
\hline
   $H_i(B^2\Z/2 \times B^2\U;\Z)$ & $\Z$ & $0$ & $\Z/2$ & $\Z$ & $\Z/4$ & $(\Z/2)^3$ & $\Z/2$ & ${\scriptstyle (\Z/2)^2\times \Z/3 \times \Z/4}$ \\
   $H_i(B^2\Z/2 \times B^2\U; \Z/2)$ & $\Z/2$ & $0$ & $\Z/2$ & $(\Z/2)^2$ & $\Z/2$ & $(\Z/2)^4$ & $(\Z/2)^4$ & $(\Z/2)^4$ \\
\hline
  \end{tabular}
  \caption{Integral and mod 2 homology groups of $B^2\Z/2$ and $B^2\Z/2 \times B^2\U$. Recall the homology of $B^2 \U$ is recorded in the first two lines of Table~\ref{tab:MAXhom}.}
  \label{tab:SQEDhom}
\end{table}

We next write down the AHSS for oriented bordism associated to the trivial fibration $\pt \to B^2\Z/2 \times B^2\U \to B^2\Z/2 \times B^2\U$, for which the second page is given by (see Fig.~\ref{fig:AHSS_ScalarQED}),
\beq
E^2_{p,q} = H_p\left(B^2\Z/2 \times B^2\U; \Omega_q^{\SO}(\pt) \right)\, .
\eeq
\begin{figure}[h]
\centering
  \includegraphics[width=0.7\textwidth]{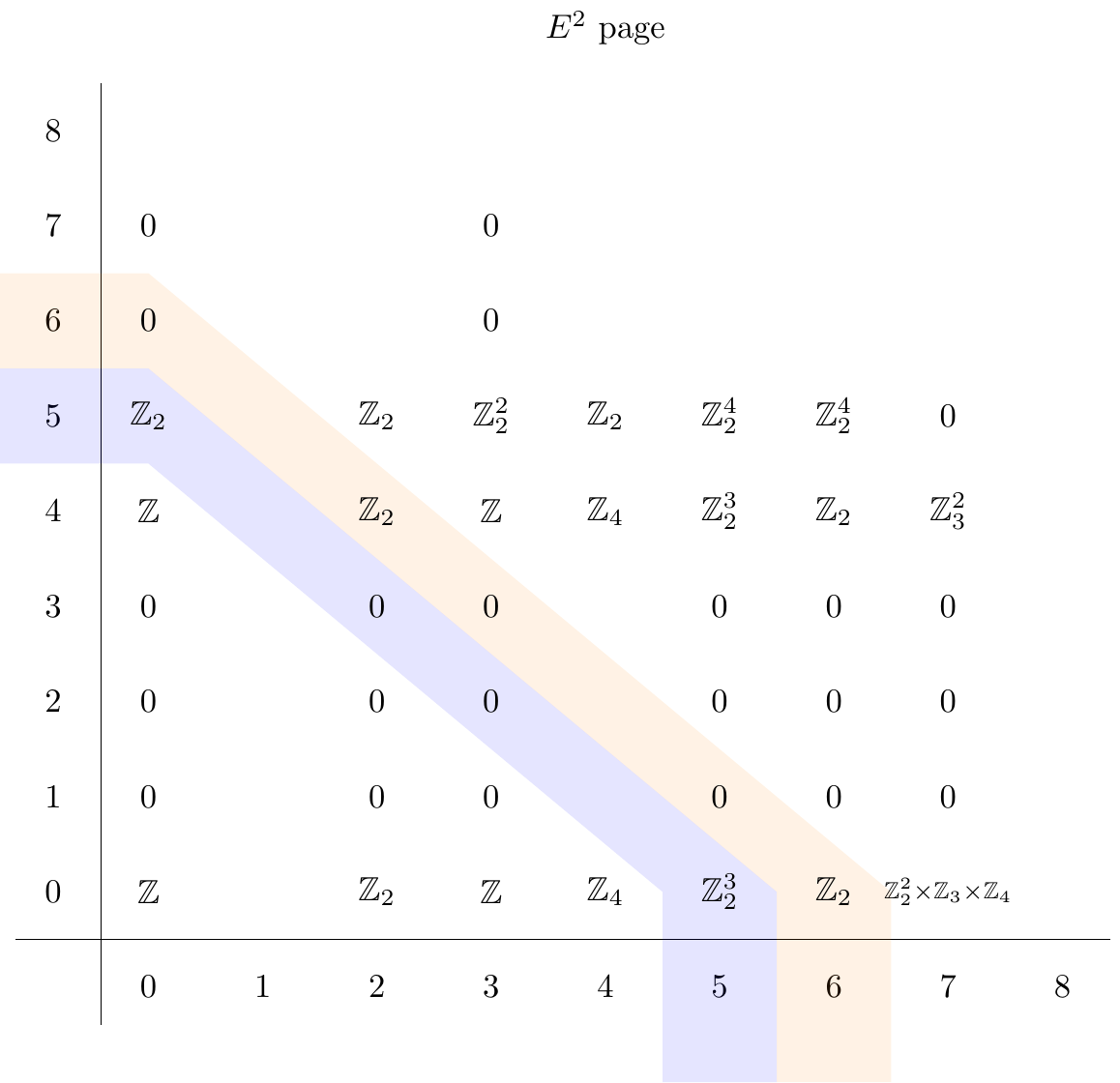}
\caption{The second page of the AHSS for $\Omega^{\SO}_\bullet\left(B^2\Z/2 \times B^2\U \right)$.    \label{fig:AHSS_ScalarQED} }
\end{figure}
We want to read off $\Omega_5$ and $\Omega_6$ from the stabilization of elements on the $p+q=5$ and $p+q=6$ diagonals respectively. For the fifth bordism group, very similar arguments to those of~\ref{app:maxwell} tell us that $E_{5,0}^\infty=(\Z/2)^3$ and $E_{0,5}^\infty=\Z/2$, and the canonical split $\Omega^\SO_\bullet(X)\cong \Omega^\SO_\bullet (\pt) \oplus \widetilde{\Omega}^\SO_\bullet(X)$ again tells us that the  extension problem is trivial, giving:
\beq
\Omega^\SO_5\left(B^2\Z/2 \times B^2\U \right) = (\Z/2)^4 \, .
\eeq
For $\Omega_6$, there is a potentially non-vanishing differential on the fifth page, $d^5:E_{7,0}\to E_{2,5}:(\Z/2)^2 \times \Z/3 \times \Z/4 \mapsto \Z/2$. Without knowing this differential, one cannot read off $\Omega_6$ -- the most we can conclude is that it equals $\Z/2$, $\Z/4$, or $\Z/2 \times \Z/2$. In any case, we have that $\Omega^\SO_6\left(B^2\Z/2 \times B^2\U \right)$ is pure torsion, which is all that's relevant for discussing the corresponding anomalies in 4d.

\section{Computation of a key differential} \label{app:Arun}

This Appendix is written by Arun Debray.

In this appendix, we finish the computation from \S\ref{ss:odd_post}, showing that if $\G$ is a $2$-group with
0- and 1-form parts both $\U$ and $k$-invariant equal to $m$ times the generator of $H^4(B\U;\Z) \cong \Z$,
where $m$ is odd, then $\Omega_5^\Spin(B|\G|)\cong \Z/m$. After the calculation in \S\ref{ss:odd_post}, the only
ambiguity is in the $2$-torsion, which could be either $0$ or $\Z/2$, depending on the value of a $d_3$ in the
Atiyah-Hirzebruch spectral sequence in Fig.~\ref{fig:AHSS-oddm}. This differential resisted several of our
standard techniques to address it, as did an analogous differential in the Adams spectral sequence computing the
$2$-completion of $\Omega_5^\Spin(B|\G|)$.

Here we finish the computation of the $2$-torsion subgroup of $\Omega_5^\Spin(B|\G|)$ in a different
way, resolving the differentials only implicitly. Specifically, we fit $\Omega_\bullet^\Spin(B|\G|)$ into a long exact
sequence~\eqref{our_specific_LES} and compute enough other terms in the long exact sequence to determine
$\Omega_5^\Spin(B|\G|)$. This long exact sequence is a version of the Gysin sequence for spin bordism, and one of
its homomorphisms has a geometric interpretation as a \textit{Smith homomorphism}, a map between bordism groups
defined by sending a manifold $M$ to a smooth representative of the Poincar\'e dual to a characteristic class of
$M$. The general version of~\eqref{our_specific_LES} was written down in~\cite{Debray:2023yrs} to study Smith
homomorphisms and their applications to anomalies of quantum field theories; see there and~\cite{Cordova:2019wpi,Hason:2020yqf} for
more on Smith homomorphisms and their role in physics.

\begin{theorem}[\cite{Debray:2023yrs}]
\label{Smith_thm}
Let $V\to X$ be a vector bundle of rank $r$ over a CW complex $X$, $\pi\colon S(V)\to X$ be the sphere bundle of
$V$, $X^{V-r}$ be the Thom spectrum of the rank-zero virtual vector bundle $V-\underline\R^r$, and $E_\bullet$ be a
generalised homology theory. Then there is a long exact sequence
\begin{equation}
\label{smith_LES}
	\dotsb\longrightarrow E_n(S(V))\overset{\pi_*}{\longrightarrow} E_n(X)\overset{\psi}{\longrightarrow}
	E_{n-r}(X^{V-r}) \longrightarrow E_{n-1}(S(V))\longrightarrow\dots
\end{equation}
\end{theorem}
\noindent The map $\psi$ is the Smith homomorphism, but we will not need that fact; see~\cite{Debray:2023yrs} for more information.
\begin{proof}
We will show that the Thom space $\mathrm{Th}(X, V)$ of $V$ is the homotopy cofibre of $\pi$, so there is the Puppe
long exact sequence $\dotsb\to E_n(S(V))\to E_n(X)\to \widetilde E_n(\mathrm{Th}(X, V))\to E_{n-1}(S(V))\to\dotsb$.
Once this is established, \eqref{smith_LES} follows: $X^{V-r}\simeq \Sigma^{-r}\Sigma^\infty_+ \mathrm{Th}(X, V)$,
and $\Sigma_+^\infty$ induces an isomorphism for any generalised homology theory and $\Sigma^{-r}$ shifts
generalised homology down in grading by $r$, yielding~\eqref{smith_LES}.\footnote{The long exact
sequence~\eqref{smith_LES}, and indeed this whole appendix, could have been done with $\mathrm{Th}(X, V)$, but once
we apply bordism, it is conceptually helpful to replace $\mathrm{Th}(X, V)$ by $X^{V-r}$ ---
$\widetilde\Omega_k^\Spin(\mathrm{Th}(X, V))$ is a bordism group of $(k-r)$-dimensional manifolds, and
$\Omega_k^\Spin(X^{V-r})$ is a bordism group of $k$-manifolds. The degree shift must occur somewhere, and we
believe our choice is clearest.}

The homotopy cofibre of an inclusion $A\hookrightarrow X$ of CW complexes such that the image of $A$ is a union of
cells is the quotient $X/A$; in general, one can compute the homotopy cofibre of a map $f\colon A\to X$ between CW
complexes by changing $f$ and $X$ by a homotopy so that $f$ is indeed a cellular inclusion. For $\pi\colon S(V)\to
X$ from the theorem statement, replace $\pi$ with the inclusion $i\colon S(V)\hookrightarrow D(V)$ of $S(V)$ into
the disc bundle $D(V)$ of $V\to X$; the map $\pi'\colon D(V)\to X$ is a homotopy equivalence and $\pi'\circ i =
\pi$, and the CW structures on $X$ and the standard CW structure on $D^n$ can be used to put CW structures on
$S(V)$ and $D(V)$ such that the image of $i$ is a union of cells. Thus the homotopy cofibre of $\pi\colon S(V)\to
V$ is the quotient $D(V)/S(V)$, which is by definition $\mathrm{Th}(X, V)$.
\end{proof}
We are interested in $E_\bullet = \Omega_\bullet^\Spin$, and specifically in $2$-torsion. To avoid worrying about torsion for
other primes, we \emph{localise at $2$}, meaning we tensor with the ring $\Z_{(2)}$ of rational numbers whose
denominators in lowest terms are odd. If $A$ is a finitely generated abelian group, so that $A$ is a direct sum of
a free abelian group $\Z^r$ and cyclic groups $\Z/p_i^{e_i}$ of prime-power order, localising at $2$ has the effect
of replacing $\Z^r$ with $(\Z_{(2)})^r$, preserving all $2^{e_i}$-torsion, and sending all odd-prime-power torsion
to zero. Since we are trying to determine the $2$-torsion subgroup of $\Omega_5^\Spin(B|\G|)$, which is a finitely
generated abelian group, we lose no relevant information by localising at $2$. Moreover, since $\Z_{(2)}$ is a flat
ring, tensoring with $\Z_{(2)}$ preserves long exact sequences ({\em i.e.}\ $2$-localisation is exact), so we can still
apply~\eqref{smith_LES} after $2$-localising.

For any $2$-group $\mathbb H$, the classifying space $B|\mathbb H|$ has a CW structure ({\em e.g.}\ because it is the
geometric realization of a simplicial space), so we can invoke Theorem~\ref{Smith_thm} for $X = B|\G|$. Quotienting $\G$
by its 1-form symmetry defines a map to the 0-form part of $\G$, which on classifying spaces induces a map
$q\colon B|\G|\to B\U$. Let $V\to B|\G|$ be the pullback of the tautological complex line bundle $L\to B\U$ by $q$.
Then Theorem~\ref{Smith_thm}, together with exactness of $2$-localisation, gives us a long exact sequence
\begin{equation}
\label{our_specific_LES}
	\dotsb\to \Omega_5^\Spin(S(V)) \otimes \Z_{(2)} \to \Omega_5^\Spin(B|\G|)\otimes
	\Z_{(2)} \to \Omega_3^\Spin((B|\G|)^{V-2}) \otimes\Z_{(2)} \to\dots
\end{equation}
We will show in Theorem~\ref{s_and_q_Smith} that $\Omega_5^\Spin(S(V)) \otimes \Z_{(2)}$ and
$\Omega_3^\Spin((B|\G|)^{V-2}) \otimes\Z_{(2)}$ both vanish; exactness then implies $\Omega_5^\Spin(B|\G|)\otimes
\Z_{(2)} = 0$ as well, meaning $\Omega_5^\Spin(B|\G|)$ has no $2$-torsion, which is what we want to prove in this
appendix.

Before doing so, we first need to figure out what $S(V)$ is:
\begin{lemma}
The sphere bundle $S(V)$ is homotopy equivalent to $K(\Z, 3)$, and this homotopy equivalence identifies the bundle
map $S(V)\to B|\G|$ with the map $K(\Z, 3)\to B|\G|$ given by inclusion of the 1-form $\U$ symmetry.
\end{lemma}
\begin{proof}
Let $q\colon B|\G|\to B\U$ be the classifying map for $V$. Then the sphere bundle of $V$ is the pullback of the
sphere bundle of the tautological line bundle $L\to B\U$, {\em i.e.}\ there is a homotopy pullback square
\begin{equation}
\begin{tikzcd}
	{S(V)} & {S(L)} \\
	{B|\G|} & B\U
	\arrow["q", from=2-1, to=2-2]
	\arrow[from=1-1, to=2-1]
	\arrow[from=1-2, to=2-2]
	\arrow[from=1-1, to=1-2]
\end{tikzcd}
\end{equation}
The sphere bundle of a complex line bundle is the associated principal $\U$-bundle; for $L\to B\U$, this is the
tautological bundle $E\U\to B\U$, meaning $S(L)$ is contractible. In general, a homotopy pullback in which one of
the two legs is contractible is the homotopy fibre (denoted $\mathrm{hofib}(\text{--})$) of the other leg, meaning
the map $S(V)\to B|\G|$ can be identified with the canonical map $\mathrm{hofib}(q)\to B|\G|$. Recall that $q$
arose from a short exact sequence of topological groups
\begin{equation}
\label{group_SES}
	1\longrightarrow G\overset{i}{\longrightarrow} H\overset{p}{\longrightarrow} K\longrightarrow 1
\end{equation}
by applying the classifying space functor; specifically, $G = |\U[1]|$, $H = |\G|$, and $K = \U$, and $q = Bp$. The
classifying space functor turns short exact sequences of groups into fibre sequences of spaces, meaning that
for~\eqref{group_SES} the homotopy fibre of $Bp$ is homotopy equivalent to $BG$, and this homotopy equivalence
identifies the canonical map $\mathrm{hofib}(Bp)\to BH$ with $Bi\colon BG\to BH$. For our specific choices of $G$,
$H$, and $K$, this tells us that up to homotopy equivalence, the map $S(V)\to B|\G|$ is the map $B|\U[1]|\simeq
K(\Z, 3)\to B|\G|$ as claimed in the theorem statement.
\end{proof}
\begin{lemma}
\label{s_and_q_Smith}
$\Omega_5^\Spin(K(\Z, 3))\otimes\Z_{(2)}$ and $\Omega_3^\Spin((B|\G|)^{V-2})\otimes \Z_{(2)}$ both vanish.
\end{lemma}
\begin{proof}
We check both of these using the Adams spectral sequence; the relevant differentials and extension questions are
trivial for degree reasons, so these computations are straightforward analogues of Beaudry and Campbell's computations
in Ref.~\cite{beaudry2018guide}. We direct the reader to~\cite{beaudry2018guide} for more on our proof strategy and notation. The Adams spectral
sequence computes $2$-completed spin bordism, not $2$-localised spin bordism, but because $\Omega_5^\Spin(K(\Z,
3))$ and $\Omega_3^\Spin((B|\G|)^{V-2})$ are finitely generated abelian groups, this distinction is not important:
the $2$-localisation of a finitely generated abelian group $A$ vanishes if and only if the $2$-completion of $A$
vanishes, if and only if $A$ lacks both free and $2$-torsion summands.

For $K(\Z, 3)$, we need as input the $\cA(1)$-module structure on $H^\bullet(K(\Z, 3);\Z/2)$ in degrees $6$ and
below. This was computed by Serre~\cite[\S 10]{Serre1953aP}: $H^\bullet(K(\Z, 3);\Z/2)$ is a polynomial algebra on generators
$\tau_3$ in degree $3$, $\Sq^2\tau_3$ in degree $5$, and other generators in degrees too high to matter to us. Using this,
one can compute that if $\uQ$ denotes the $\cA(1)$-module $\cA(1)/(\Sq^1, \Sq^2\Sq^3)$, then there is an
isomorphism of $\cA(1)$-modules
\begin{equation}
\label{A1_KZ3}
	H^\bullet(K(\Z, 3);\Z/2)\cong \textcolor{BrickRed}{\Z/2} \oplus \textcolor{MidnightBlue}{\Sigma^3\uQ}
		\oplus P,
\end{equation}
where $P$ is concentrated in degrees $7$ and above, hence is irrelevant for our computation. We draw this
decomposition in Fig.~\ref{KZ3_Adams}, left. Liulevicius~\cite[Theorem 3]{Liu62} computes
$\Ext_{\cA(1)}(\textcolor{BrickRed}{\Z/2})$ and Beaudry--Campbell~\cite[Figure 29]{beaudry2018guide} compute
$\Ext_{\cA(1)}(\textcolor{MidnightBlue}{\uQ})$, so we can draw the $E_2$-page of the Adams spectral sequence in
Fig.~\ref{KZ3_Adams}, right. The $E_2$-page vanishes in topological degree $5$, {\em i.e.}\ for $t-s = 5$, so the
$E_\infty$-page must also vanish in that degree, and we conclude $\Omega_5^\Spin(K(\Z, 3))\otimes\Z_{(2)} = 0$.

\begin{figure}[h!]
\centering
\begin{subfigure}[c]{0.2\textwidth}
	\begin{tikzpicture}[scale=0.6, every node/.style = {font=\tiny}]
        \foreach \y in {0, ..., 6} {
                \node at (-2, \y) {$\y$};
        }
		\begin{scope}[BrickRed]
  			\fill (0, 0) circle (3pt);
			\node[right] at (0, 0) {$1$};
		\end{scope}
		\begin{scope}[MidnightBlue]
    		\questionupsidedon (0, 3, );
			\node[right] at (0, 3) {$\tau_3$};
		\end{scope}
	\end{tikzpicture}
\end{subfigure}
\begin{subfigure}[c]{0.4\textwidth}
\begin{sseqdata}[name=KZ3, classes = fill, scale=0.6, xrange={0}{6}, yrange={0}{4}, Adams grading,
>=stealth,
x label = {$\displaystyle{s\uparrow \atop t-s\rightarrow}$},
x label style = {font = \small, xshift = -16.5ex, yshift=5.5ex},
class labels = { left = 0.07em, font=\small }
]
\begin{scope}[BrickRed]
	\class(0, 0)\AdamsTower{}
	\class(1, 1)\structline(0, 0)(1, 1)
	\class(2, 2)\structline
	\class(4, 3)\AdamsTower{}
\end{scope}
\begin{scope}[MidnightBlue]
	\class(3, 0)\AdamsTower{}
\end{scope}
\end{sseqdata}
\printpage[name=KZ3, page=2]
\end{subfigure}
\caption{
Left: the $\cA(1)$-module structure on $H^\bullet(K(\Z, 3);\Z/2)$ in degrees $6$ and below. Right: the $E_2 =
E_\infty$-page of the Adams spectral sequence computing the $2$-completion of $\Omega_\bullet^\Spin(K(\Z, 3))$.
There is nothing in topological degree $5$, so the $2$-completion of
$\Omega_5^\Spin(K(\Z, 3))$ vanishes.}
\label{KZ3_Adams}
\end{figure}

For $(B|\G|)^{V-2}$, using the description of the $\cA(1)$-module structure on $H^\bullet(B|\G|;\Z/2)$ in low degrees
from \S\ref{app:mod-2-cohomology} and the way Stiefel--Whitney classes twist the Steenrod squares of a Thom spectrum
(see~\cite[\S 3.3]{beaudry2018guide}), we obtain an isomorphism of $\cA(1)$-modules
\begin{equation}
	H^\bullet((B|\G|)^{V-2};\Z/2)\cong C\eta \oplus P',
\end{equation}
where $C\eta\coloneqq \Sigma^{-2}\widetilde H^\bullet(\mathbb{CP}^2;\Z/2)$ and $P'$ is concentrated in degrees $5$ and
above, hence is irrelevant to our computation. We draw this isomorphism in Fig.~\ref{GThom_Adams}, left.
\begin{figure}[h!]
\centering
\begin{subfigure}[c]{0.2\textwidth}
	\begin{tikzpicture}[scale=0.6, every node/.style = {font=\tiny}]
        \foreach \y in {0, ..., 4} {
                \node at (-2, \y) {$\y$};
        }
  		\fill (0, 0) circle (3pt);
  		\fill (0, 2) circle (3pt);
		\sqtwoR(0, 0, );
		\node[right] at (0, 0) {$U$};
	\end{tikzpicture}
\end{subfigure}
\begin{subfigure}[c]{0.4\textwidth}
\begin{sseqdata}[name=twistG, classes = fill, scale=0.6, xrange={0}{4}, yrange={0}{3}, Adams grading,
>=stealth,
x label = {$\displaystyle{s\uparrow \atop t-s\rightarrow}$},
x label style = {font = \small, xshift = -13ex, yshift=5.5ex},
class labels = { left = 0.07em, font=\small }
]
\class(0, 0)\AdamsTower{}
\class(2, 1)\AdamsTower{}
\class(4, 2)\AdamsTower{}
\end{sseqdata}
\printpage[name=twistG, page=2]
\end{subfigure}
\caption{
Left: the $\cA(1)$-module structure on $H^\bullet((B|\G|)^{V-2};\Z/2)$ in degrees $4$ and below, where $V$ is the complex
line bundle associated to the map $q\colon B|\G|\to B\U$ given by quotienting $\G$ by its 1-form symmetry. Right:
the $E_2 = E_\infty$-page of the Adams spectral sequence computing the $2$-completion of
$\Omega_\bullet^\Spin((B|\G|)^{V-2})$. There is nothing in topological degree $3$, so the $2$-completion of
$\Omega_3^\Spin((B|\G|)^{V-2})$ vanishes.}
\label{GThom_Adams}
\end{figure}
$\Ext_{\cA(1)}(C\eta)$ is computed in~\cite[Example 4.5.6 and Figure 22]{beaudry2018guide}, so we can draw the $E_2$-page of
the Adams spectral sequence in Fig.~\ref{GThom_Adams}, right. The $E_2$-page is empty in topological degree $3$, so the
$E_\infty$-page is also empty, and we conclude $\Omega_3^\Spin((B|\G|)^{V-2})\otimes \Z_{(2)} = 0$.
\end{proof}


\bibliography{references} \bibliographystyle{JHEP}

\end{document}